\documentclass[11pt,a4paper]{article}
\usepackage{graphicx}
\usepackage{amsmath}
\usepackage{natbib}
\usepackage{multirow}
\usepackage{arydshln}

\usepackage{listings}
\usepackage{lstbayes}

\RequirePackage[OT1]{fontenc}
\RequirePackage{amsthm,amsmath,natbib}
\RequirePackage{hyperref}

\usepackage{stmaryrd}
\usepackage{enumitem}
\usepackage{caption}
\usepackage{subcaption}
\usepackage{amssymb}
\usepackage{bbm}
\usepackage{bm}
\usepackage{verbatim}

\usepackage{xcolor}

\usepackage{tikz}
\usepackage{tikz-qtree}
\usetikzlibrary{bayesnet}

\newcommand{\R}{\mathbb{R}}

\newcommand{\N}{\mathbb{N}}

\newcommand{\E}{\mathbb{E}}
\newcommand{\I}{\mathbb{I}}
\newcommand{\1}{\mathbbm{1}}

\newcommand{\avec}{\textbf{a}}
\newcommand{\bvec}{\textbf{b}}
\newcommand{\yvec}{\textbf{y}}
\newcommand{\amean}{a_{\cdot}}
\newcommand{\bmean}{b_{\cdot\cdot}}
\newcommand{\bimean}{{b}_{i\cdot}}
\newcommand{\ymean}{y_{\cdot\cdot\cdot}}
\newcommand{\yimean}{y_{i\cdot\cdot}}
\newcommand{\yijmean}{y_{ij\cdot}}

\newcommand{\gavec}{\boldsymbol{\gamma}}
\newcommand{\etavec}{\boldsymbol{\eta}}
\newcommand{\gamean}{\gamma_{\cdot}}
\newcommand{\etamean}{\eta_{\cdot\cdot}}
\newcommand{\etai}{\eta_{i\cdot}}

\newcommand{\ratio}[2]{r_{#1,#2 }}
\newcommand{\var}[1]{\sigma_{#1}^2}
\newcommand{\tvar}[1]{\tilde{\sigma}_{#1}^2}
\newcommand{\precision}[1]{\tau_{#1}}
\newcommand{\tprec}[1]{\tilde{\tau}_{#1}}

\newcommand{\tE}{\tilde{\E}}

\newcommand{\balpha}{{\boldsymbol{\alpha}}}

\newcommand{\bbeta}{{\boldsymbol{\beta}}}
\newcommand{\bgamma}{{\boldsymbol{\gamma}}}

\newcommand{\mbeta}{{\bar{\beta}}}

\newcommand{\tbeta}{\tilde{\beta}}
\newcommand{\tbbeta}{\tilde{\bbeta}}
\newcommand{\by}{\boldsymbol{y}}
\newcommand{\ba}{{\boldsymbol{a}}}
\newcommand{\ma}{\bar{a}}

\newtheorem{defi}{Definition}
\newtheorem{rmk}{Remark}
\newtheorem{theorem}{Theorem}

\newtheorem{ex}{Example}
\newtheorem{lemma}{Lemma}
\newtheorem{coroll}{Corollary}

\newtheorem{taggedmodelx}{Model}
\newenvironment{taggedmodel}[1]
 {\renewcommand\thetaggedmodelx{#1}\taggedmodelx}
 {\endtaggedmodelx}

\newtheorem{taggedsamplerx}{Sampler}
\newenvironment{taggedsampler}[1]
 {\renewcommand\thetaggedsamplerx{#1}\taggedsamplerx}
 {\endtaggedsamplerx}

\begin{document}
\title{Multilevel linear models, Gibbs samplers and multigrid decompositions}

\author{Giacomo Zanella\footnote{
Department of Decision Sciences, BIDSA and IGIER, Bocconi University, via Roentgen 1, 20136
Milan, Italy.
\href{mailto:giacomo.zanella@unibocconi.it}{giacomo.zanella@unibocconi.it}
}
\ and Gareth Roberts\footnote{
Department of Statistics, University of Warwick, Coventry, CV4 7AL, UK.
\href{mailto:gareth.o.roberts@warwick.ac.uk}{gareth.o.roberts@warwick.ac.uk}
}}

\maketitle

\begin{abstract}
We study the convergence properties of the Gibbs Sampler in the context of posterior distributions arising from Bayesian analysis of conditionally Gaussian hierarchical models. 
We develop a multigrid approach to derive analytic expressions for the convergence rates of the algorithm 
for various widely used model structures, including nested and crossed random effects.
Our results apply to multilevel models with an arbitrary number of layers in the hierarchy, while most previous work was limited to the two-level nested case.
The theoretical results provide explicit and easy-to-implement guidelines to optimize practical implementations of the Gibbs Sampler, such as indications on which parametrization to choose (e.g.\ centred and non-centred), which constraint to impose to guarantee statistical identifiability, and which parameters to monitor in the diagnostic process.
Simulations suggest that the results are informative also in the context of non-Gaussian distributions and more general MCMC schemes, such as gradient-based ones.
\end{abstract}


\section{Introduction}\label{sec:intro}

Markov chain Monte Carlo (MCMC) is established as the computational work\-horse of most Bayesian statistical analyses for complex models.
For hierarchical models with conditionally conjugate priors, the Gibbs sampler \citep{gelsmi90,smirob93} remains one of the most natural algorithm of choice, thanks to its simplicity of implementation and low computational cost per iteration (thanks to conjugacy and conditional independence).
Nonetheless, speed of convergence of the resulting Markov chain can be a major issue and can be highly sensitive to the model structure and the implementation details, such choice of parametrization \citep{hills1992parameterization,GelfandSahuCarlin1995} or identifiability constraints \citep{vines1996fitting,gelfand1999identifiability,xie2006measures}.
%
This work provides a contribution towards gaining a quantitative understanding of the interaction between Bayesian hierarchical structures and the behaviour of MCMC algorithms, which lies at the heart of the practical success of Bayesian statistics.

While there is some previous work in the area
\citep{RobertsSahu1997,mengvand,RobertsPapaspiliopoulosSkold2003,jones2004sufficient, PapaspiliopoulosRobertsSkold2007,mengyu},
current theoretical understanding of the interaction 
 between Bayesian hierarchical models and MCMC convergence
 is still very limited, and almost nothing is known for models of hierarchical depth greater than two.
The present paper offers a contribution 
 towards such an understanding, focusing on theory for Gaussian hierarchical models and seeking quantitative results. 
In particular, we derive analytic expressions for the convergence rates of the Gibbs Sampler for various multilevel linear models and explore the dependence of these rates on the model structure, the choice of parametrization and the introduction of identifiability constraints.
The theoretical results given in this paper extend and improve substantially on existing literature \citep{RobertsSahu1997,mengyu,bass2016comparison,GaoOwen2017}
both in terms of generality of hierarchical structure and the availability of explicit rates.
 We also show by simulations that the understanding gained from the Gaussian case can be extrapolated to more general settings.
%
%

In general, the Gibbs sampler can be elegantly described in terms of orthogonal projections \citep{.amit:1991:stochastic+convergence+distributions,.amit:1996:properties+convergence+perturbations,diaconis2010}.
While in principle this theory provides the tools to extract practical convergence information for Gibbs samplers in the context of multivariate Gaussian distributions, in order to apply it to practically used Bayesian multilevel models one needs detailed knowledge of the spectrum of non-trivial high-dimensional matrices, which has drastically limited its applicability to derive analytic results.
In this paper we combine this general framework with a novel multigrid decomposition approach that allows us to focus on low-dimensional Markov chains and derive explicit analytic results concerning Gibbs sampler rates of convergence for multilevel linear models, such as nested and crossed random effect models with an arbitrary number of layers and/or factors.

Our results have various practical implications.
First they can be readily used in the popular context of conditionally Gaussian models, 
where there exist unknown variances at various levels of the hierarchy \citep{gelman2006data}.
In that case our results describe, for example, the optimal updating strategies for the hierarchical mean structure conditional on the variances, allowing to optimize the mean parametrization on the fly (Section \ref{sec:cond_optimal}), or the computationally optimal way of imposing statistically identifiability (Sections \ref{sec:identif}), and provide
theoretically grounded indication of which 
parameters to monitor in the convergence diagnostic process (Section \ref{sec:motivating_example}).
Also, our results can be used as a building block to derive computational complexity statements about the Gibbs Sampler in the context of multilevel linear model (see e.g.\ \cite{papaspiliopoulos2018scalable} for work in that direction).
Note that in the context of conditionally Gaussian models the entire Gaussian mean component could be updated in a single block, thus avoiding convergence issues related to single-site updates.
However these block updates can in principle be computationally expensive
(up to $O(n^3)$ cost in the dimension ($n$) of the Gaussian to be updated), while single-site updating schemes with provably bounded convergence rate can offer a more scalable alternative.
For some class of models, sparse linear algebra methods can reduce the cost of the block update by exploiting sparsity in the posterior precision matrix, but the resulting computational cost depends on the model structure and can still be super-linear (see e.g.\ Section \ref{sec:crossed} for models leading to dense precision matrices and \cite{papaspiliopoulos2018scalable} for more discussion).

While impressive results are being obtained with black-box software implementation of Hamiltonian Monte Carlo (HMC) such as STAN \citep{carpenter2017stan}, our results suggest that Gibbs Sampling schemes built on our methodological guidance can be substantially cheaper than gradient-based ones in the context of hierarchical models, leading to improved performances (Section \ref{sec:gamma_poisson}).
Moreover, our simulations show that the methodological results we develop in this paper are also helpful when fitting multilevel models with gradient-based schemes (Section \ref{sec:HMC}) and 
allow to obtain drastic improvements in efficiency also when using generic software, such as STAN.

Throughout the paper, we shall couch all our results in terms of $L^2$ rates of convergence. Specifically,
let $(\bbeta(s))_{s=1,2,\dots}$ be a Markov chain with stationary distribution $\pi$ and transition operator defined by $P^s f(\bbeta(0))=\E[f(\bbeta(s))|\bbeta(0)]$.
The \emph{rate of convergence} $\rho(\bbeta(s))$ associated to $(\bbeta(s))_{s=1,2,\dots}$ is defined as the smallest number $\rho$ such that for all $r>\rho$
\begin{equation}\label{eq:rates_of_conv}
\lim_{s\rightarrow\infty}
\frac{\|P^s f-\E_\pi[f]\|_{L^2(\pi)}}{r^s}
=0
\qquad
\forall 
f\in L^2(\pi)\,,
\end{equation}
where $L^2(\pi)$ denotes the space of square $\pi$-integrable functions, $\|\cdot\|_{L^2(\pi)}$ is its associated $L^2$-norm and $\E_\pi[f]=\int f\, d\pi$ is the expectation of $f$ with respect to $\pi$.
The rate of convergence $\rho(\bbeta(s))$ characterizes the speed at which $(\bbeta(s))_{s=1,2,\dots}$ converges to its stationary distribution $\pi$, with a simple argument giving that if
$$
T = \min \{s;\ 
\|P^sf-\E_\pi[f]\|_{L^2(\pi)} \le \epsilon \}
 $$ 
 then $T = \mathcal{O}\left(\frac{1}{-\log(\rho)}\right)$.

\subsection{Paper overview and structure}
Section \ref{sec:motiv} carefully introduces the 3-level hierarchical models we shall consider, and  provides motivating simulations. Then
in Section \ref{sec:multigrid} we shall give a  complete analysis for 3-level symmetric models (i.e.\ homogeneous variances and symmetric data structure). At the heart of the analysis is a multigrid decomposition of the Gibbs sampler into completely independent Markov chains describing different levels of hierarchical granularity, Theorem \ref{thm:factorization_3}.
Such multigrid decomposition simultaneously applies to every Gibbs sampler induced by all 
centred/non-centred parametrizations 
 and is fundamentally a statistical property of the hierarchical models under consideration.
Although multigrid ideas have already been used in methodological contexts to design improved MCMC schemes \citep{GoodmanSokal1989,LiuSabatti2000}, to our knowledge they had never been used in theoretical contexts to study convergence rates.
We demonstrate that the slowest of these independent chains is always that corresponding to the coarsest level, regardless of the value of the variance components and on the number of branches in the hierarchy, and thus derive explicit expressions for the rates of convergence in symmetric contexts.

In Section \ref{sec:crossed} we focus on crossed effect models, using again a multigrid decomposition approach to derive explicit convergence rates.
The results show that in the context of crossed models, centred/non-centred reparametrizations are not sufficient to guarantee fast convergence of the resulting Gibbs Sampler.
On the other hand, we show that the latter can be achieved by imposing stronger statistical identifiability through additional linear constraints and our theory provides indications on which constraints lead to faster convergence. 
Finally, a simulation study reported in Section \ref{sec:gamma_poisson} suggests 
that the analysis of the Gaussian case leads to useful guidance also in the case of non-Gaussian models for both the Gibbs Sampler and Hamiltonian Monte Carlo algorithms \citep{neal2011mcmc}.

Section \ref{sec:bespoke} considers 3-level non-symmetric hierarchical models, providing bounds on convergence rates based on comparisons with related symmetric models and discussing the use use of \emph{bespoke parametrizations}, where the choice of centred or non-centred parametrization in each branch of the hierarchy depends on the branch-specific parameter.

Section \ref{sec:k_levels} considers hierarchical models with arbitrary depth ($\ge 4$).
Using an appropriate auxiliary random walk, whose evolution through the hierarchical tree is governed by the parameters' squared partial correlations, we are able to extend the multigrid analysis to general tree structures and some non-symmetric cases.
We again demonstrate a fundamental multigrid decomposition in 
Theorem \ref{thm:factorization_general} where the coarsest level chain
converges the slowest, and we give explicit formulae for optimal partial non-centering strategies.

\section{Three level hierarchical linear models}\label{sec:motiv}

The theoretical innovation in this paper is centred around an important case in which we can obtain explicit Gibbs sampler rates of convergence, and as a result study explicitly the effects of particular models, parametrization schemes and blocking strategies. 
Therefore we shall begin with a detailed study of the following three-level Gaussian linear model, giving a fairly complete understanding of the interaction between model structure and parametrization and the 
Gibbs Sampler convergence behaviour.
\begin{taggedmodel}{S3}[Symmetric 3-levels hierarchical model]\label{model:S3}
Suppose
\begin{equation}\label{eq:NCP_3}
y_{ijk}=\mu+a_i+b_{ij}+\epsilon_{ijk},
\end{equation}
where $i$, $j$ and $k$ run from 1 to $I$, $J$ and $K$ respectively and $\epsilon_{ijk}$ are iid normal random variables with mean 0 and variance $\var{e}$.
We employ the standard Bayesian model specification assuming 
$a_i\sim N(0,\var{a})$, $b_{ij}\sim N(0,\var{b})$ and a flat prior on $\mu$.
\end{taggedmodel}
For the theoretical analysis, we will consider the variance terms $\var{a}$, $\var{b}$ and $\var{e}$ to be known, while in the simulations we will assume them to be unknown and give them a prior distribution.
Defining $\avec=(a_i)_{i}$, $\bvec=(b_{ij})_{i,j}$ and $\yvec=(y_{ijk})_{i,j,k}$, the Gibbs Sampler explores the posterior distribution $(\mu,\avec,\bvec)|\yvec$ by iteratively sampling from the full conditional distributions of $\mu$, $\avec$ and $\bvec$ as follows (see below for motivation of denoting such sampler as GS($1,1$)).
\begin{taggedsampler}{GS($1,1$)}\label{sampler:GSNN}
Initialize $\mu(0)$, $\avec(0)$ and $\bvec(0)$ and then iterate
\begin{enumerate}[noitemsep,nolistsep]
\item Sample $\mu(s+1)$ from $p(\mu|\avec(s),\bvec(s),\yvec)$;
\item Sample $a_i(s+1)$ from $p(a_i|\mu(s+1),\bvec(s),\yvec)$ for all $i$; 
\item Sample $b_{ij}(s+1)$ from $p(b_{ij}|\mu(s+1),\avec(s+1),\yvec)$ for all $i$ and $j$,
\end{enumerate}
where $p(\mu|\avec,\bvec,\yvec)$, $p(a_i|\mu,\bvec,\yvec)$ and $p(b_{ij}|\mu,\avec,\yvec)$ are the full conditionals of Model \ref{model:S3} (see supplementary material for explicit expressions).
\end{taggedsampler}

Given the conditional independence structure of the model, Sampler GS(1,1) is equivalent to a blocked Gibbs sampler with components $\mu$, $\avec$ and $\bvec$, i.e.\ a scheme performing consecutive updates of $\mu|\avec,\bvec$, $\avec|\mu,\bvec$ and $\bvec|\mu,\avec$ at each iteration.

\begin{figure}[h!]
\centering
\begin{subfigure}[b]{0.49\textwidth}
\centering
  \tikz{
     \node[obs] (y) {$y_{ijk}$};%
     \node[latent,above=of y,yshift=-0.5cm] (eta) {$\eta_{ij}$}; %
     \node[latent,above=of eta,yshift=-0.5cm] (gamma) {$\gamma_i$}; %
     \node[latent,above=of gamma,xshift=0cm,yshift=-0.5cm] (mu) {$\mu$}; %
     \plate [inner sep=0.2cm,xshift=0cm,yshift=0.1cm] {plate1} {(y)} {$k=1,\dots,K$}; %
     \plate [inner sep=0.2cm,xshift=-0cm,yshift=0.1cm] {plate2} {(plate1) (eta)} {$j=1,\dots,J$};
     \plate [inner sep=0.2cm,xshift=-0cm,yshift=0.1cm] {plate3} {(plate2) (gamma)} {$i=1,\dots,I $}
     \edge {eta} {y}
     \edge {gamma} {eta}
     \edge {mu} {gamma}
}   
        \caption{Fully centred parametrization}
        \label{fig:CC}
\end{subfigure}
\hfill
\begin{subfigure}[b]{0.49\textwidth}
\centering
  \tikz{
     \node[obs] (y) {$y_{ijk}$};%
     \node[latent,above=of y,xshift=0cm] (b) {$b_{ij}$}; %
     \node[latent,left=of b,xshift=-0.4cm] (a) {$a_i$}; %
     \node[latent,left=of a,xshift=0cm] (mu) {$\mu$}; %
     \plate [inner sep=0.2cm,xshift=0.1cm,yshift=0.1cm] {plate1} {(y)} {$k=1,...,K$}; %
     \plate [inner sep=0.2cm,xshift=0cm,yshift=0.1cm] {plate2} {(plate1) (b)} {$j=1,...,J$};
     \plate [inner sep=0.2cm,xshift=0cm,yshift=0.1cm] {plate3} {(plate2) (a)} {$i=1,...,I $}
     \edge {b} {y}
     \edge {a} {y}
     \edge {mu} {y}
}   
        \caption{Fully non-centred parametrization}
        \label{fig:NN}
\end{subfigure}
\begin{subfigure}[b]{0.49\textwidth}
\centering
  \tikz{
     \node[obs] (y) {$y_{ijk}$};%
     \node[latent,above=of y,xshift=0cm,yshift=-0.4cm] (eta) {$\eta_{ij}$}; %
     \node[latent,above=of eta,yshift=-0.4cm] (a) {$a_i$}; %
     \node[latent,left=of a,xshift=-0.7cm] (mu) {$\mu$}; %
     \plate [inner sep=0.2cm,xshift=0cm,yshift=0.1cm] {plate1} {(y)} {$k=1,...,K$}; %
     \plate [inner sep=0.2cm,xshift=-0cm,yshift=0.1cm] {plate2} {(plate1) (eta)} {$j=1,...,J$};
     \plate [inner sep=0.2cm,xshift=-0cm,yshift=0.1cm] {plate3} {(plate2) (a)} {$i=1,...,I $}
     \edge {eta} {y}
     \edge {a} {eta}
     \edge {mu} {eta}
}   
        \caption{Mixed parametrization: $(\mu,\avec,\etavec)$}
        \label{fig:NC}
\end{subfigure}
\hfill
\begin{subfigure}[b]{0.49\textwidth}
\centering
  \tikz{
     \node[obs] (y) {$y_{ijk}$};%
     \node[latent,above=of y,xshift=0cm,yshift=-0.4cm] (b) {$b_{ij}$}; %
     \node[latent,left=of b,xshift=-0.4cm] (gamma) {$\gamma_i$}; %
     \node[latent,above=of gamma,xshift=0cm,yshift=-0.2cm] (mu) {$\mu$}; %
     \plate [inner sep=0.2cm,xshift=0.1cm,yshift=0.1cm] {plate1} {(y)} {$k=1,...,K$}; %
     \plate [inner sep=0.2cm,xshift=-0cm,yshift=0.1cm] {plate2} {(plate1) (eta)} {$j=1,...,J$};
     \plate [inner sep=0.2cm,xshift=-0cm,yshift=0.1cm] {plate3} {(plate2) (gamma)} {$i=1,...,I $}
     \edge {b} {y}
     \edge {gamma} {y}
     \edge {mu} {gamma}
}   
        \caption{Mixed parametrization: $(\mu,\gavec,\bvec)$}
        \label{fig:CN}
\end{subfigure}
    \caption{Graphical representations of 3-levels hierarchical linear models under different parametrizations.}
    \label{fig:4parametrizations}
\end{figure}

The parametrization $(\mu,\avec,\bvec)$ induced by \eqref{eq:NCP_3} is often referred to as \emph{non-centred parametrization} (NCP) and it is contrasted with the \emph{centred para\-me\-tri\-za\-tion} (CP) obtained by replacing $a_i$ and $b_{ij}$ with $\gamma_i=\mu+a_i$ and $\eta_{ij}=\gamma_i+b_{ij}$ respectively.
Under the centred parametrization $(\mu,\gavec,\etavec)$ the model formulation becomes
\begin{equation}\label{eq:CP_3}
y_{ijk}\sim N(\eta_{ij},\var{e}),\qquad
\eta_{ij}\sim N(\gamma_i,\var{b}),\qquad
\gamma_i\sim N(\mu,\var{a}),\qquad
p(\mu)\propto1\,.
\end{equation}
Figures \ref{fig:NN} and \ref{fig:CC} provides a graphical representation of the two parametrizations.
In the $(\mu,\avec,\bvec)$ case $(1,1)$ refers to the fact that both levels 1 and 2 use a non-centred parametrization, while in the $(\mu,\gavec,\etavec)$ case $(0,0)$ indicates that both levels use a centred parametrization.
The resulting Gibbs sampler for the centred parametrization is as follows.
\begin{taggedsampler}{GS($0,0$)}\label{sampler:GSCC}
Initialize $\mu(0)$, $\gavec(0)$ and $\etavec(0)$ and then iterate
\begin{enumerate}[noitemsep,nolistsep]
\item Sample $\mu(s+1)$ from $p(\mu|\gavec(s),\etavec(s),\yvec)$;
\item Sample $\gamma_i(s+1)$ from $p(\gamma_i|\mu(s+1),\etavec(s),\yvec)$ for all $i$; 
\item Sample $\eta_{ij}(s+1)$ from $p(\eta_{ij}|\mu(s+1),\gavec(s+1),\yvec)$ for all $i$ and $j$,
\end{enumerate}
where $p(\mu|\gavec,\etavec,\yvec)$, $p(\gamma_i|\mu,\etavec,\yvec)$ and $p(\eta_{ij}|\mu,\gavec,\yvec)$ are the full conditionals induced by \eqref{eq:CP_3} (see supplementary material for explicit expressions).
\end{taggedsampler}
Together with the fully non-centred parametrization $(\mu,\avec,\bvec)$ and the fully centred parametrization $(\mu,\gavec,\etavec)$, one can also consider the mixed parametrizations given by $(\mu,\gavec,\bvec)$ and $(\mu,\avec,\etavec)$ and the corresponding Gibbs Sampler schemes $GS{(0,1)}$ and $GS(1,0)$.
See Figures \ref{fig:NC} and \ref{fig:CN} for graphical representations.

\subsection{Illustrative example}\label{sec:motivating_example}
As an illustrative example, we simulated data from Model \ref{model:S3} with $I=J=100$, $K=5$, $\mu=0$, $\sigma_{a}=\sigma_{e}=10$ and $\sigma_b=10^{-0.5}$.
This correspond to a scenario of high level of noise in the measurements.
We fit model \ref{model:S3} assuming the standard deviations $(\sigma_{a},\sigma_{b},\sigma_e)$ to be unknown and placing weakly informative priors, namely $\frac{1}{\var{a}}$, $\frac{1}{\var{b}}$ and $\frac{1}{\var{e}}$ a priori distributed according to an Inverse Gamma distribution with shape and rate parameters equal to $0.01$.
We compare the efficiency of the Gibbs sampling schemes corresponding to the four different parametrizations, denoting them by $GS(1,1)$, $GS(0,0)$, $GS(0,1)$ and $GS(1,0)$.
For this simple example we initialized the chains at true values of the parameters $(\mu,\avec,\bvec)$ and $(\sigma_a,\sigma_b,\sigma_e)$, which we know because we are in a simulated dataset example.
The more realistic case of starting the chains from randomly chosen states led to the same conclusions of this illustrative examples. 
Note that all the four schemes have the same starting states (modulo re-parametrization) to have a fair comparison.

Figure \ref{fig:level0_mixing} shows the mixing behaviour of the global mean $\mu$ and displays the potentially dramatic difference among mixing properties of the Gibbs Sampler under different parametrizations.
\begin{figure}[h!]
\centering
\includegraphics[width=\linewidth]{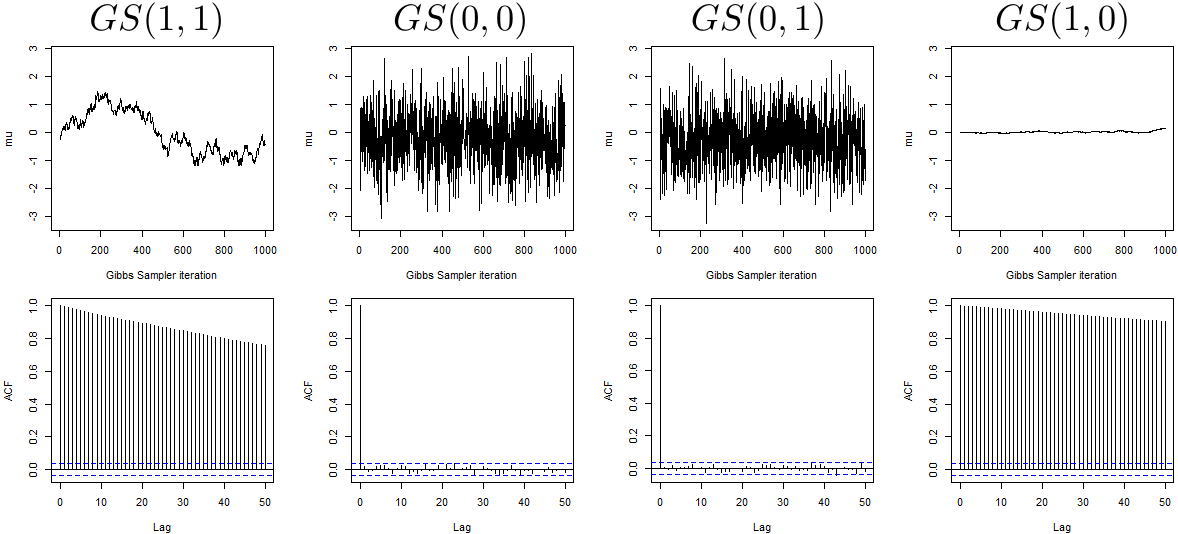}
\caption{
Mixing behaviour at level $0$ (i.e. $\mu$) under four different parametrizations.
The ranges of the y-axes are constant across parametrizations.
}
\label{fig:level0_mixing}
\end{figure}
Based on Figure \ref{fig:level0_mixing} one would certainly exclude using $GS(1,1)$
and
$GS(1,0)$ to fit the model under consideration and may be tempted to deduce that both
$GS(0,0)$ and 
$GS(0,1)$
lead to good mixing properties of the resulting chain.
However, as an additional check, a cautious practitioner may also explore the mixing of the parameters at the first level, namely $\avec$ and $\gavec$. 
Figure \ref{fig:level1_mixing} displays the behaviour of the global averages of such parameters, namely $\amean=\frac{\sum_{i} a_i}{I}$ and $\gamean=\frac{\sum_i \gamma_i}{I}$, in the first 1000 iterations.
\begin{figure}[h!]
\centering
\includegraphics[width=\linewidth]{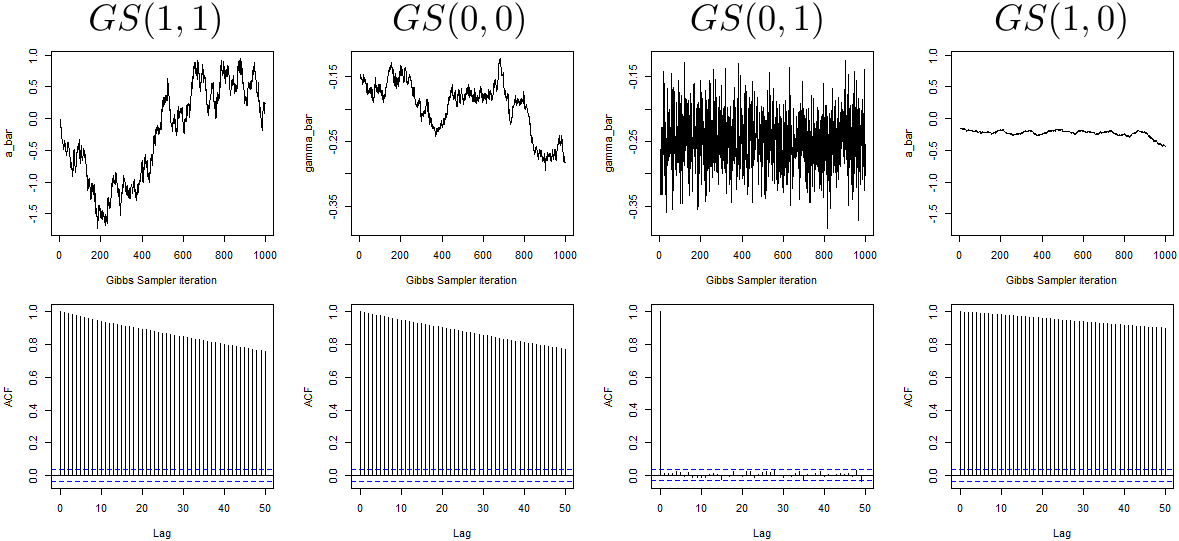}
\caption{Mixing behaviour at level 1 (i.e.\ $\amean$ and $\gamean$) under four different parametrizations. 
The ranges of the y-axes are constant across parametrizations sharing the same parameters at level 1.}
\label{fig:level1_mixing}
\end{figure}
Again, we see a dramatic difference induced by different parametrizations and, somehow surprisingly, we see that, despite having good mixing behaviour at level 0 (i.e.\ $\mu$), $GS(0,0)$ displays very poor mixing behaviour at level 1 (i.e.\ $\gavec$).
It is then natural to explores also the mixing behaviour at level 2 and Figure \ref{fig:level2_mixing} does so again by plotting the global averages $\bmean=\frac{\sum_{ij}\beta_{ij}}{IJ}$ and $\etamean=\frac{\sum_{ij}\eta_{ij}}{IJ}$.
\begin{figure}[h!]
\centering
\includegraphics[width=\linewidth]{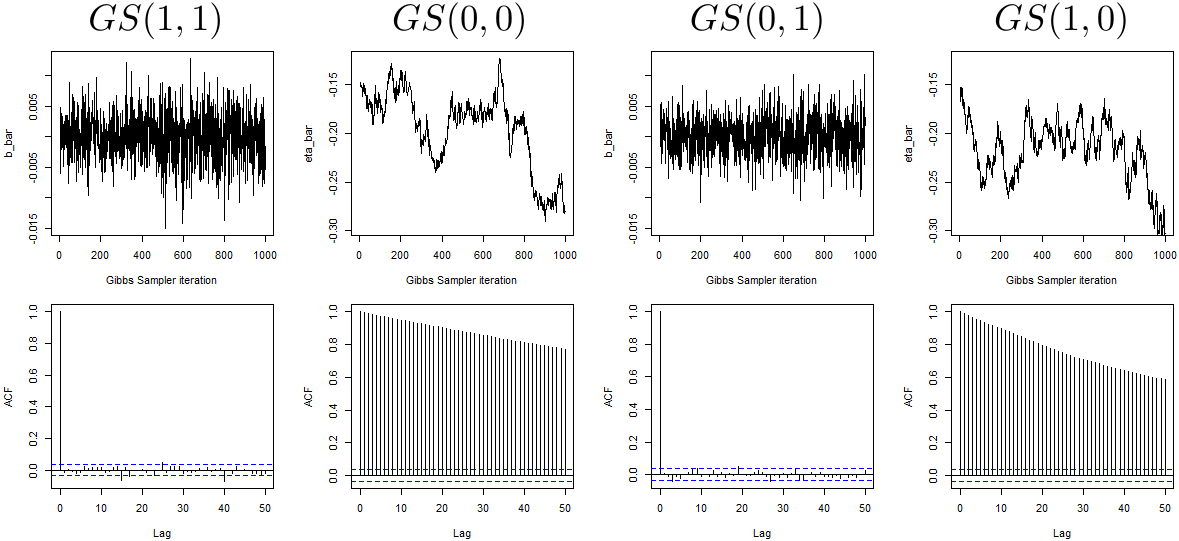}
\caption{
Mixing behaviour at level 2 (i.e.\ $\bmean$ and $\etamean$) under four different parametrizations. 
The ranges of the y-axes are constant across parametrizations sharing the same parameters at level 2.
}
\label{fig:level2_mixing}
\end{figure}
In this case $GS(1,1)$ and $GS(0,1)$ are the only one achieving good mixing.
Based on Figures \ref{fig:level0_mixing}, \ref{fig:level1_mixing} and \ref{fig:level2_mixing} it is natural to choose to fit the model using the sampler $GS(0,1)$ corresponding to the mixed parametrization $(\mu,\gavec,\bvec)$, as it is the only one providing a good mixing across all three levels.

This simple example shows many typical issues arising when fitting Bayesian multi-level models and raises many questions. 
For example, one would like to know what are good parameters to use to diagnose convergence, in order to avoid misleading conclusions like the one suggested by Figure \ref{fig:level0_mixing}.
In fact, while in two level model good mixing of the global hyperparameters such as $\mu$ typically indicates good global mixing, this is not true in other multi-level models. 
Indeed, it is legitimate to wonder whether diagnoses based only on the global means, like in Figures 
\ref{fig:level0_mixing}-\ref{fig:level2_mixing}, are enough to deduce good mixing of the whole Markov chain, which in our example has more than $10^4$ dimensions ($1+I+IJ$ mean components and $3$ precision components).
Below we will show that for Model \ref{model:S3}, mixing of the global means ensures mixing of the whole $(1+I+IJ)$-dimensional mean components of the chain given the variances (see e.g.\ Corollary \ref{thm:rate_eq_skeleton}).
Therefore it is enough to monitor the three global means and the three variances to ensure a reliable check of the chain mixing properties.

Even more crucially, it is desirable to have simple and theoretically grounded guidance in choosing a computationally efficient parametrization, given the huge impact it can have on computational performances.
The theoretical analysis developed in the next section will provide useful guidance in this respect.

\section{Multigrid decomposition for the three level hierarchical model}\label{sec:multigrid}
The basic ingredient of our analysis is the following multigrid decomposition.
Consider the four possible parametrization of Model \ref{model:S3}: 
$(\mu,\avec,\bvec)$,
$(\mu,\gavec,\etavec)$
and the mixed parametrizations $(\mu,\gavec,\bvec)$ and $(\mu,\avec,\etavec)$.
In order to provide a unified treatment, regardless of the chosen parametrization, we denote the parameters used by $(\bbeta^{(0)},\bbeta^{(1)},\bbeta^{(2)})$ and the resulting Gibbs Sampler by $GS(\bbeta)$.
For example, in the NCP case $\bbeta^{(0)}=\mu$, $\bbeta^{(1)}=\avec$, $\bbeta^{(2)}=\bvec$ and $GS(\bbeta)$ coincides with GS(1,1).
First consider the map $\delta$ sending $\bbeta=(\bbeta^{(0)},\bbeta^{(1)},\bbeta^{(2)})$ to 
\begin{equation}\label{eq:multigrid_3}
\delta(\bbeta)
=
\left( \begin{array}{c}
\delta^{(0)}\bbeta \\
\delta^{(1)}\bbeta \\
\delta^{(2)}\bbeta
\end{array}\right)
=
\left( \begin{array}{c}
\delta^{(0)}\bbeta^{(0)}\,,\,\delta^{(0)}\bbeta^{(1)}\,,\,\delta^{(0)}\bbeta^{(2)} \\
\delta^{(1)}\bbeta^{(1)}\,,\,\delta^{(1)}\bbeta^{(2)}\\
\delta^{(2)}\bbeta^{(2)}
\end{array}\right)\,,
\end{equation}
where, loosely speaking, $\delta^{(i)}\bbeta$ represent the increments of $\bbeta$ at the $i$-th coarseness level.
More precisely 
\begin{gather*}
\delta^{(0)}\beta^{(0)}
=
\beta^{(0)}\,,\qquad
\delta^{(0)}\bbeta^{(1)}
=
\beta^{(1)}_{\cdot}\,,\qquad
\delta^{(0)}\bbeta^{(2)}
=
\beta^{(2)}_{\cdot\cdot}\,,
\\
\delta^{(1)}\bbeta^{(1)}
=
\left(\beta^{(1)}_1-\beta^{(1)}_{\cdot},\dots,\beta^{(1)}_I-\beta^{(1)}_{\cdot}\right),\,
\delta^{(1)}\bbeta^{(2)}
=
\left(
\beta^{(2)}_{1\cdot}-\beta^{(2)}_{\cdot\cdot},\dots,
\beta^{(2)}_{I\cdot}-\beta^{(2)}_{\cdot\cdot}\right)\,,
\\
\delta^{(2)}\bbeta^{(2)}
=
\left(
\beta^{(2)}_{11}-\beta^{(2)}_{1\cdot},
\beta^{(2)}_{12}-\beta^{(2)}_{1\cdot},\dots,
\beta^{(2)}_{I(J-1)}-\beta^{(2)}_{I\cdot},
\beta^{(2)}_{IJ}-\beta^{(2)}_{I\cdot}\right)\,,
\end{gather*}
where
\begin{align*}
\beta^{(1)}_{\cdot}=\frac{\sum_{i}\beta^{(1)}_i}{I}
\,,\qquad
\beta^{(2)}_{\cdot\cdot}
=
\frac{\sum_{i,j}\beta^{(2)}_{ij}}{IJ}
\,,\qquad
\beta^{(2)}_{i\cdot}=\frac{\sum_{j}\beta^{(2)}_{ij}}{J}\,.
\end{align*}
It is easy to see that the map $\delta$ is a bijection between $\R^d$ and 
$\R^3\times (\R^{I})^*\times (\R^{I})^*\times_{i=1}^I(\R^{J})^*$, where $(\R^p)^*=\{(v_1,\dots,v_p)\in\R^p\,:\,\sum_{i=1}^pv_i=0\}$.
The dimensionality of $\delta\bbeta$ equals the one of $\bbeta$, which is $1+I+IJ$, because $\delta\bbeta$ has $3+2I+IJ$ parameters and $2+I$ constraints.
The following theorem shows that the Markov chain induced by $GS(\bbeta)$ factorizes under the transformation $\delta$.
\begin{theorem}[Multigrid Decomposition]\label{thm:factorization_3} 
Let $(\bbeta(s))_{s=1}^\infty$ be a Markov chain on $\R^d$ evolving according to $GS(\bbeta)$. 
Then the timewise transformations $(\delta^{(0)}\bbeta(s))_{s=1}^\infty$, $(\delta^{(1)}\bbeta(s))_{s=1}^\infty$ and $(\delta^{(2)}\bbeta(s))_{s=1}^\infty$ are each a Markov chain and evolve independently.
\end{theorem}
\tikzstyle{state}=[shape=ellipse,draw=blue!50,fill=blue!20]
\tikzstyle{observation}=[shape=rectangle,draw=orange!50,fill=orange!20]
\tikzstyle{lightedge}=[<-,dotted]
\tikzstyle{mainstate}=[state,thick]
\tikzstyle{mainedge}=[<-,solid]
\definecolor{Apricot}{rgb}{0.98, 0.81, 0.69}
\newlength{\cola}
\setlength{\cola}{0cm}
\newlength{\colb}
\setlength{\colb}{2cm}
\newlength{\cole}
\setlength{\cole}{5cm}
\newlength{\colf}
\setlength{\colf}{7cm}%
\tikzstyle{state}=[shape=ellipse,draw=Apricot!50,fill=Apricot!20]
\tikzstyle{beta}=[state]
\tikzstyle{delta0}=[state,draw=-red!50!green!25,fill=-red!85!green!70]
\tikzstyle{delta1}=[state,draw=red!50,fill=red!25]
\tikzstyle{delta2}=[state,draw=green!40,fill=green!15]
\begin{figure}[t!]
\begin{center}
\begin{tikzpicture}[]
\node               at (\cola,6+2.7) {$s$};
\node[beta,inner sep=0.2pt] (s1_1) at (\cola,5.4+2.7) {$\bbeta^{(0)}$};
\node[beta,inner sep=0.2pt] (s2_1) at (\cola,4.6+2.7) {$\bbeta^{(1)}$};
\node[beta,inner sep=0.2pt] (s3_1) at (\cola,3.8+2.7) {$\bbeta^{(2)}$};
\node               at (\colb,6+2.7) {$s+1$};
\node[beta,inner sep=0.2pt] (s1_2) at (\colb,5.4+2.7) {$\bbeta^{(0)}$}
    edge[mainedge,>=latex] (s2_1)
    edge[mainedge,>=latex] (s3_1)
;
\node[beta,inner sep=0.2pt] (s2_3) at (\colb,4.6+2.7) {$\bbeta^{(1)}$}
    edge[mainedge,>=latex] (s1_2)
    edge[mainedge,>=latex] (s3_1)
;
\node[beta,inner sep=0.2pt] (s3_4) at (\colb,3.8+2.7) {$\bbeta^{(2)}$}
    edge[mainedge,>=latex,bend right=55] (s1_2)
    edge[mainedge,>=latex] (s2_3)
;
\node               at (\cole,8.5+0.2) {$s$};
\node[delta0,inner sep=0.2pt] (s1_3) at (\cole,7.9+0.2) {$\delta^{(0)}\bbeta$};
\node[delta1,inner sep=0.2pt] (s2_3) at (\cole,7.1+0.2) {$\delta^{(1)}\bbeta$};
\node[delta2,inner sep=0.2pt] (s3_3) at (\cole,6.3+0.2) {$\delta^{(2)}\bbeta$};
\node               at (\colf,8.5+0.2) {$s+1$};
\node[delta0,inner sep=0.2pt] (s1_4) at (\colf,7.9+0.2) {$\delta^{(0)}\bbeta$}
    edge[mainedge,>=latex] (s1_3)
;
\node[delta1,inner sep=0.2pt] (s2_4) at (\colf,7.1+0.2) {$\delta^{(1)}\bbeta$}
    edge[mainedge,>=latex] (s2_3)
;
\node[delta2,inner sep=0.2pt] (s3_4) at (\colf,6.3+0.2) {$\delta^{(2)}\bbeta$}
    edge[mainedge,>=latex] (s3_3)
;
\end{tikzpicture}
\end{center}
\caption{
Illustration of Theorem \ref{thm:factorization_3}.
Left: the transition from $\bbeta(s)$ to $\bbeta(s+1)$ in Sampler $GS(\bbeta)$ follows the structure of a Gibbs Sampler with $3$ components.
Right: the transition from $\delta\bbeta(s)$ to $\delta\bbeta(s+1)$ in Sampler $GS(\bbeta)$ follows the structure of three independent Markov chains.
}
\end{figure}

While the posterior independence of $\delta^{(0)}\bbeta$, $\delta^{(1)}\bbeta$ and $\delta^{(2)}\bbeta$ is well-known, the remarkable fact following from Theorem \ref{thm:factorization_3} is that also the Markov chains induced by the Gibbs Sampler are independent (note that the independence of the random vector under the target measure is a necessary but not sufficient condition for the independence of a corresponding MCMC scheme).
\begin{rmk}
It is worth noting that the three subspaces of $\R^d$ spanned by the vectors
$\delta^{(0)}\bbeta$,  
$\delta^{(1)}\bbeta$ and  
$\delta^{(2)}\bbeta$,
respectively, do not depend on the choice of parametrization $\bbeta$.
Thus the multigrid decomposition is intrinsic to the model, and not dependent on the particular parametrization being considered.
\end{rmk}

Theorem \ref{thm:factorization_3} provides a useful tool to analyze the Markov chain of interest, $\bbeta(s)$. 
In fact the factorization into independent Markov chains implies that the rate of convergence of $\bbeta(s)$ is simply given by the worst rate of convergence among $\delta^{(0)}\bbeta(s)$, $\delta^{(1)}\bbeta(s)$ and $\delta^{(2)}\bbeta(s)$.
Interestingly, the slowest chain is always the chain at the highest level $\delta^{(0)}\bbeta(s)$, regardless of the choice of parametrization and the values of $(I,J,K,\sigma_{a},\sigma_b,\sigma_{e})$.
\begin{theorem}[Hierarchical ordering of convergence rates]\label{thm:rates_ordering_3} 
Let $\delta^{(0)}\bbeta(s)$, $\delta^{(1)}\bbeta(s)$ and $\delta^{(2)}\bbeta(s)$ be the Markov chains defined in Theorem \ref{thm:factorization_3}.
Then the associated convergence rates satisfy
$$
\rho(\delta^{(0)}\bbeta(s))\geq
\rho(\delta^{(1)}\bbeta(s))\geq
\rho(\delta^{(2)}\bbeta(s))=0\,.
$$
\end{theorem}

Theorems \ref{thm:factorization_3} and \ref{thm:rates_ordering_3} imply that the rate of convergence of the global chain $\bbeta(s)$ coincides with the one of the sub-chain $\delta^{(0)}\bbeta(s)$ sampling the global means $(\beta^{(0)},\beta^{(1)}_{\cdot},\beta^{(2)}_{\cdot\cdot})$. 
\begin{coroll}\label{thm:rate_eq_skeleton}(Rate of convergence of $GS(\bbeta) $)
Given the notation of Theorem \ref{thm:factorization_3},
$$
\rho(\bbeta(s))=
\rho(\delta^{(0)}\bbeta(s))\,.
$$
\end{coroll}

\subsection{Explicit rates of convergence under different parametrizations}\label{sec:explicit_rates}
The multigrid decomposition developed in Section \ref{sec:multigrid} allows to perform a direct analysis on the convergence properties of the Markov chain of interest $\bbeta(s)$.
The latter is a Gibbs Sampler targeting a multivariate Gaussian distributions and thus, in principle, could be analyzed using, for example, the tools developed in \cite{.amit:1996:properties+convergence+perturbations,RobertsSahu1997,khare2009rates}.
However, these results require to have a full characterization of the spectrum of a $d\times d$ matrix, where $d$ is the number of dimensions in the Markov chain under consideration. 
Given the high-dimensionality of $\bbeta(s)$, which has $1+I+IJ$ parameters, it is hard to apply directly such results and in fact the convergence properties of $\bbeta(s)$ have been studied heuristically or numerically in the literature (see e.g.\ \cite[Sec.4]{GelfandSahuCarlin1995} and \cite[Sec.4.2]{RobertsSahu1997}). 
Corollary \ref{thm:rate_eq_skeleton}, however implies that it suffices to study the skeleton chain $\delta^{(0)}\bbeta(s)$, which is a low-dimensional chain (namely 3-dimensional) amenable to direct analysis.
Therefore, using Corollary \ref{thm:rate_eq_skeleton}, we can derive analytic expressions for the rates of convergence for the Gibbs Sampler under different parametrizations.
\begin{theorem}\label{thm:rates_3}
Given an instance of Model \ref{model:S3}, the rate of convergence of the four Gibbs Sampler schemes $GS(0,0)$, $GS(1,1)$, $GS(0,1)$ and $GS(1,0)$ are given by
\begin{align*}
&\rho_{00}=
1-\frac{\tvar{a}}{\tvar{a}+\tvar{b}}\,\frac{\tvar{b}}{\tvar{b}+\tvar{e}}\,,
&\rho_{10}=
\max\left\{\frac{\tvar{a}}{\tvar{a}+\tvar{b}},\frac{\tvar{e}}{\tvar{b}+\tvar{e}}\right\}\,,
\\
&\rho_{01}=
1-\frac{\tvar{a}}{\tvar{a}+\tvar{e}}\,\frac{\tvar{e}}{\tvar{b}+\tvar{e}}\,,
&\rho_{11}=
\max\left\{\frac{\tvar{a}}{\tvar{a}+\tvar{e}},\frac{\tvar{b}}{\tvar{b}+\tvar{e}}\right\}\,,
\end{align*}
where $\tvar{a}=\frac{\var{a}}{I}$, $\tvar{b}=\frac{\var{b}}{IJ}$ and $\tvar{e}=\frac{\var{e}}{IJK}$.
\end{theorem}
Theorem \ref{thm:rates_3} provides explicit and informative formulas regarding the interaction between choice of parametrization and resulting efficiency of the Gibbs Sampler for Model \ref{model:S3}.
\begin{figure}[h!]
\centering
\includegraphics[width=\linewidth]{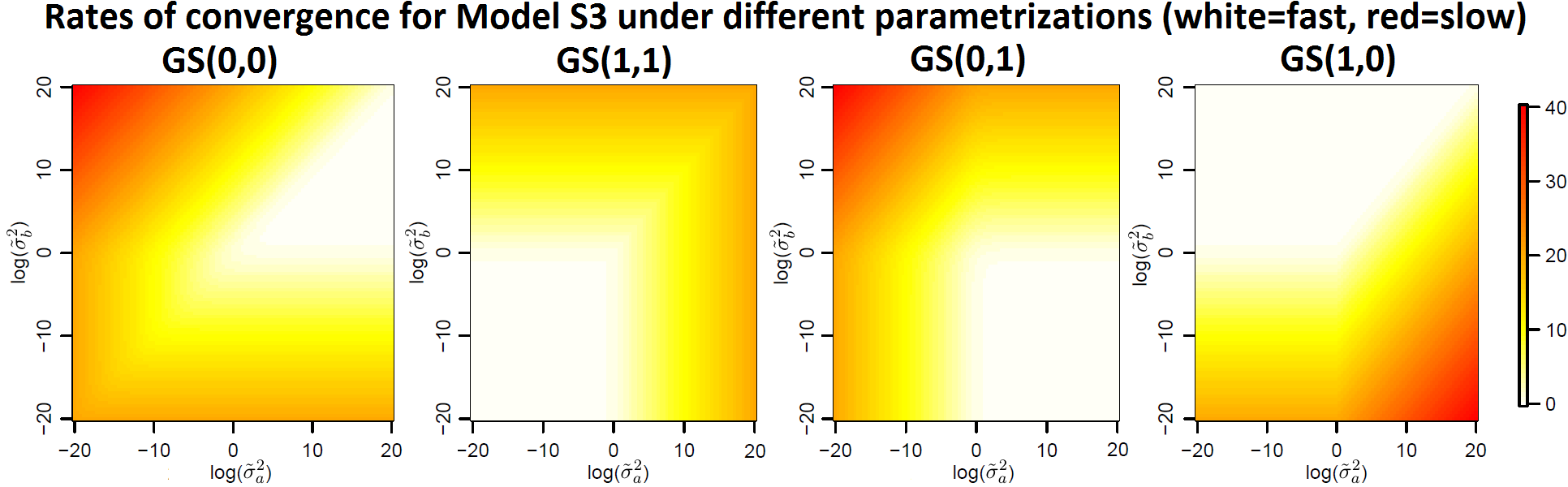}
\caption{Plot of rates of convergence for three-levels Gaussian hierarchical models under different parametrizations.
 Color levels correspond to values of $\log(1-\rho)$, where $\rho$ is the rate of convergence, as a function of $\log(\tvar{a})$ and $\log(\tvar{b})$ for fixed $\log(\tvar{e})=0$.
}
\label{fig:rates_3}
\end{figure}
Figure \ref{fig:rates_3} summarizes graphically the dependence of the converge rates of different parametrizations from the values of the variances of various levels.
Roughly speaking, the figure suggests that there is a partition of the hyperparameter space (corresponding to the white regions in each plot) such that in each region one and only one of the four parametrizations performs well.

Consider for example the illustrative example of Section \ref{sec:motivating_example}. 
Applying Theorem \ref{thm:rates_3} to such context we obtain that the $L^2$ rates of convergence (up to the third decimal digit) of the various Gibbs Samplers under consideration given $(I,J,K,\sigma_{a},\sigma_{b},\sigma_e)=(100,100,5,10,10^{-0.5},10)$ are
$$
(\rho_{00},\rho_{11},\rho_{01},\rho_{10})
=
(0.995,0.998,0.007,0.999)\,.
$$
Recall that values of $\rho$ close to 1 mean slow convergence, see \eqref{eq:rates_of_conv} and discussion thereof.
These numbers provide a quantitative and theoretically grounded description of the behaviour heuristically observed in Section \ref{sec:motivating_example} and can be easily used to optimize performances (see e.g. Section \ref{sec:cond_optimal} below).

\subsection{Conditionally optimal parametrization}\label{sec:cond_optimal}
A natural and practically relevant question is what is the optimal parametrization (among the four possible choices $(\mu,\avec,\bvec)$, $(\mu,\gavec,\bvec)$,$(\mu,\avec,\etavec)$ and $(\mu,\gavec,\etavec)$) as a function of the normalized variance components $(\tvar{a},\tvar{b},\tvar{e})$.
Using the formulas of Theorem \ref{thm:rates_3} we can obtain the following explicit answers.
\begin{coroll}[Optimal parametrization for Model \ref{model:S3}]\label{coroll:optimal_param_3levels}
The rate of convergence of the Gibbs Sampler targeting Model \ref{model:S3} is minimized by the following choice of parametrization:
\begin{itemize}
\item use a centred parametrization $\etavec$ at the lowest level if and only if $\tvar{b}\geq\tvar{e}$,
\item use a centred parametrization $\gavec$ at the middle level if and only if $\tvar{a}\geq\tvar{b}+\tvar{e}$.
\end{itemize}
The resulting Gibbs Sampler has a rate of convergence $\rho$ upper bounded by $\frac{2}{3}$, with the equality $\rho=\frac{2}{3}$ holding if and only if $\tvar{a}=\tvar{b}+\tvar{e}$ and $\tvar{b}=\tvar{e}$ (in which case all parametrizations are equivalent).
\end{coroll}
Table \ref{table:optimal_param_3levels} provides a graphical representation of the decision process.
This simple rule guarantees that the resulting Gibbs Sampler has a rate of converges smaller than $\frac{2}{3}$, thus guaranteeing a high sampling efficiency for fixed variances.
\begin{table}[h!]
\begin{center}
\begin{tabular}{ l | c | c }
  & $\tvar{a}\geq \tvar{b}+\tvar{e}$ & $\tvar{a}< \tvar{b}+\tvar{e}$ \\
  \hline			
 $\tvar{b}\geq \tvar{e}$ & $(\mu,\gavec,\etavec)$ & $(\mu,\avec,\etavec)$ \\
  \hline			
 $\tvar{b}< \tvar{e}$  & $(\mu,\gavec,\bvec)$ & $(\mu,\avec,\bvec)$ \\
  \hline  
\end{tabular}\caption{Optimal parametrization for 3-levels hierarchical models as a function of the normalized variance components.}\label{table:optimal_param_3levels}
\end{center}
\end{table}
Table \ref{table:optimal_param_3levels} implies that the choice of parametrization of a given level (i.e.\ whether it is computationally convenient to use a centred or non-centred para\-me\-tri\-zation) depends on the ratio between the normalized variance at the level under consideration and the sum of the normalized variances of the levels below.
This results extend previous intuition for the two-level case (e.g.\ \cite{RobertsPapaspiliopoulosSkold2003}) to deeper hierarchical levels (in this case three levels).

Corollary \ref{coroll:optimal_param_3levels} allows for simple and effective strategies to ensure high sampling efficiency in practical implementations of Gibbs Sampling for Model \ref{model:S3} in the case of unknown variances.
Common implementations choose a parametrization $\bbeta=(\bbeta^{(0)},\bbeta^{(1)},\bbeta^{(2)})$ of the Gaussian component (for example the fully centred parametrization $\bbeta=(\mu,\gavec,\etavec)$) and alternate updating $\bbeta|(\sigma_a,\sigma_b,\sigma_e)$ with $GS(\bbeta)$ and $(\sigma_a,\sigma_b,\sigma_e)|\bbeta$ with direct sampling (which is straightforward using the conditional independence of $\sigma_a$, $\sigma_b$ and $\sigma_e$ given $\bbeta$). 
Given Corollary \ref{coroll:optimal_param_3levels}, instead, one can choose the optimal parametrization $\bbeta$ given $(\sigma_a,\sigma_b,\sigma_e)$ on-the-fly according to Table \ref{table:optimal_param_3levels}. 
This ensures that the sampling step $\bbeta|(\sigma_a,\sigma_b,\sigma_e)$ will have a high efficiency, regardless of the values of  $(I,J,K,\sigma_{a},\sigma_b,\sigma_{e})$.
Note that the additional computational cost  required by choosing the optimal parametrization according to Table \ref{table:optimal_param_3levels} at each step is negligible compared to the cost of a Gibbs Sampling iteration.

Figure \ref{fig:ACFs} compares the resulting autocorrelation functions in the context of the illustrative example of Section \ref{sec:motivating_example}, with unknown variances $(\sigma_a,\sigma_b,\sigma_e)$.
\begin{figure}[t!]
\centering
\includegraphics[width=\linewidth]{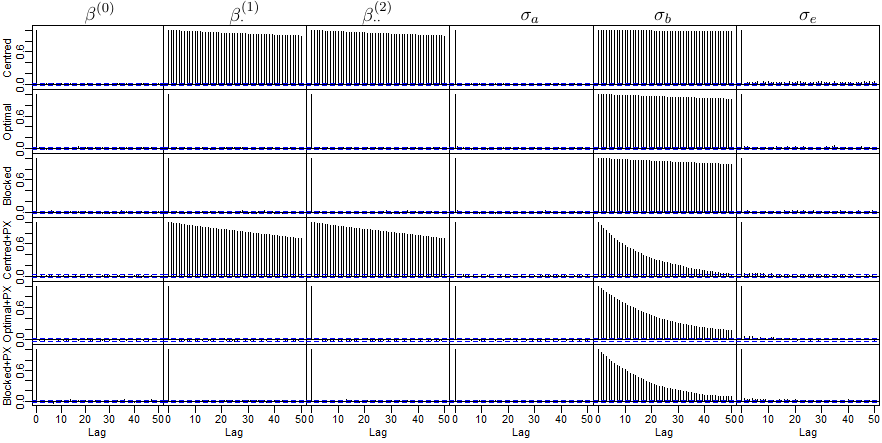}
\caption{Autocorrelation functions for the global means and the standard deviations, i.e.\
 $(\beta^{(0)},
\beta^{(1)}_{\cdot},
\beta^{(2)}_{\cdot\cdot},
\sigma_a,\sigma_b,\sigma_e)$,
under the following updating schemes, one per row:
(Centred) Gibbs Sampler with fully centred parametrization $\bbeta=(\mu,\gavec,\etavec)$;
(Optimal) Gibbs Sampler with optimal parametrization $\bbeta$ chosen on-the-fly according to Table \ref{table:optimal_param_3levels};
(Blocked) Gibbs sampler updating $\bbeta|(\sigma_a,\sigma_b,\sigma_e)$ and $(\sigma_a,\sigma_b,\sigma_e)|\bbeta$ exactly;
(+PX) combination with parameter expansion methodology.
More details in Section \ref{sec:cond_optimal}.
}
\label{fig:ACFs}
\end{figure}
We compare the Gibbs Sampler with optimal parametrization (updating $\bbeta|(\sigma_a,\sigma_b,\sigma_e)$ with $GS(\bbeta)$, where $\bbeta$ is the optimal parametrization chosen according to Table \ref{table:optimal_param_3levels}, and $(\sigma_a,\sigma_b,\sigma_e)|\bbeta$ exactly) with the centred Gibbs Sampler (updating $(\mu,\gavec,\etavec)|(\sigma_a,\sigma_b,\sigma_e)$ with $GS(0,0)$ and $(\sigma_a,\sigma_b,\sigma_e)|(\mu,\gavec,\etavec)$ exactly) and a blocked Gibbs Sampler (updating $\bbeta|(\sigma_a,\sigma_b,\sigma_e)$ and $(\sigma_a,\sigma_b,\sigma_e)|\bbeta$ exactly), which can be implemented because the distribution of $\bbeta|(\sigma_a,\sigma_b,\sigma_e)$ is multivariate Gaussian.
All schemes are then combined with the parameter expansion methodology of \citet{meng1999seeking,liu1999parameter}.
The results in Figure \ref{fig:ACFs} show that using the optimal parametrization reduces significantly the autocorrelation compared to, e.g., the fully centred one, and achieves a mixing that is basically equivalent to the one obtained by the exact blocked Gibbs Sampler.
In all cases parameter expansion helps to reduce the autocorrelation in the samples of the standard deviations $(\sigma_a,\sigma_b,\sigma_e)$.

The similarity of performances between the Gibbs Sampler with optimal para\-me\-tri\-za\-tion and the blocked one is interesting because the Gibbs update of $\bbeta|(\sigma_a,\sigma_b,\sigma_e)$ only requires univariate updates and has a potentially lower computational cost compared to a full multivariate block update of $\bbeta|(\sigma_a,\sigma_b,\sigma_e)$, which requires large matrix operations. 
While these matrix operations can be performed efficiently in the context of nested linear models (see e.g. \citealp{papaspiliopoulos2017note}), their cost becomes significantly larger for example in the context of crossed random effect models (see Section \ref{sec:crossed} below and \citealp{papaspiliopoulos2018scalable}).
Note that such a similarity of performances
is not surprising given our theoretical results above. 
In fact Corollary \ref{coroll:optimal_param_3levels} guarantees that the sampler $GS(\bbeta)$ used in the Gibbs update have a rate of convergence upper bounded by $2/3$, which is well separated from 1. 
When such updates are nested within a larger sampler (e.g the one updating $\bbeta|(\sigma_a,\sigma_b,\sigma_e)$ and $(\sigma_a,\sigma_b,\sigma_e)|\bbeta$) the difference between and exact update of $\bbeta$ and a Gibbs one with good rate of convergence can easily become negligible.

\section{Multigrid decomposition for crossed effect models}\label{sec:crossed}
Interestingly, the multigrid decomposition can be used to analyze non-nested models.
In this section we focus on the following crossed effect model.
\begin{taggedmodel}{Ck}[k-factors crossed-effects model]\label{model:k_factors}
\begin{align}
y_{i_1\dots i_k}=&
\mu+a^{(1)}_{i_1}+\dots+a^{(k)}_{i_k}+\epsilon_{i_1\dots i_k}
&i_s=1,\dots,n_s, \quad s=\,\dots,k\,,
\end{align}
with $a^{(s)}_{i_s}\stackrel{iid}\sim N(0,1/\tau_s)$ for $s\in \{1,\dots,k\}$, $\epsilon_{i_1\dots i_k}\stackrel{iid}\sim N(0,1/\tau_e)$ and $p(\mu)\propto 1$.
We denote the number of observed datapoints by $N=\prod_{s=1}^kn_s$.
\end{taggedmodel}
Similarly to Sections \ref{sec:motiv} and \ref{sec:multigrid}, we use bold letters to denote the following vectors: $\by=(y_{i_1\dots i_k})_{i_1,\dots ,i_k}$, $\ba^{(s)}=(a^{(s)}_{i_s})_{i_s}$, $\ba=(\ba^{(1)},\dots,\ba^{(k)})$ and $\ba^{(-s)}=(\ba^{(1)},\dots,\ba^{(s-1)},\ba^{(s+1)},\dots,\ba^{(k)})$.
The standard Gibbs Sampler to sample from the posterior distribution $\mathcal{L}(\mu,\ba|\by)$ of Model \ref{model:k_factors} is defined as follows.
\begin{taggedsampler}{GS-crossed}\label{sampler:k_factors_DSGS}
At each iteration
\begin{enumerate}[noitemsep,nolistsep]
\item sample $\mu$ from $\mathcal{L}(\mu|\ba,\by)$,
\item sample $\ba^{(s)}$ from $\mathcal{L}\left(\ba^{(s)}|\mu,\ba^{(-s)},\by\right)$ with $s$ going from $1$ to $k$.
\end{enumerate}
\end{taggedsampler}
Model \ref{model:k_factors} and Sampler \ref{sampler:k_factors_DSGS} have recently been analysed in \cite{papaspiliopoulos2018scalable} using the multigrid decomposition approach developed in Section \ref{sec:multigrid} of this paper to derive expressions for the convergence rate of Sampler \ref{sampler:k_factors_DSGS}.
In particular, \cite{papaspiliopoulos2018scalable} considered
the following linear functions of $\ba$ 
\begin{align}
\ma^{(s)}&=\frac{1}{n_s}\sum_{i=1}^{n_s}a^{(s)}_{i}
\quad\hbox{and}\quad
\delta\ba^{(s)}=(\ba^{(s)}-\ma^{(s)})\,.\label{eq:residuals_crossed}
\end{align}
for each $s\in\{1,\dots,k\}$ and proved the following result.
\begin{theorem}[\cite{papaspiliopoulos2018scalable}]\label{thm:multigrid_decomposition}
Let 
$$
\left((\mu,\ba)(t)\right)_{t=1}^\infty
=
\left(\mu(t),\ba^{(1)}(t),\dots,\ba^{(k)}(t)\right)_{t=1}^\infty
$$
be the Markov chain generated by Sampler \ref{sampler:k_factors_DSGS}.
Then the time-wise transformations
$
\left((\mu,\ma^{(1)},\dots,\ma^{(k)})(t)\right)_{t=1}^\infty$
and
$\left(\delta\ba^{(1)}(t)\right)_{t=1}^\infty$, \dots, $\left(\delta\ba^{(k)}(t)\right)_{t=1}^\infty$
are $(k+1)$ independent Markov chains.
Moreover, the rate of convergence of $
\left((\mu,\ba)(t)\right)_{t=1}^\infty$
is
\begin{equation}\label{eq:crossed_rate}
\rho=
\max_{s\in\{1,\dots,k\}}\frac{N\precision{e}}{N\precision{e}+n_s\precision{s}}
\,.
\end{equation}
\end{theorem}
Theorem \ref{thm:multigrid_decomposition} implies that the convergence properties of Sampler \ref{sampler:k_factors_DSGS} deteriorate as $N$ increases because $\max_{s\in\{1,\dots,k\}}(N\precision{e})^{-1}(N\precision{e}+n_s\precision{s})$ goes to 1 as $N\to\infty$. 
Motivated by this consideration, \cite{papaspiliopoulos2018scalable} propose a collapsed Gibbs Sampler that avoids such slowdown for increasing data size while preserving the same computational cost per iteration of Sampler \ref{sampler:k_factors_DSGS}. 
In the following two sections we extend the analysis of Model \ref{model:k_factors} performed in \cite{papaspiliopoulos2018scalable}, 
focusing on the role of, respectively, reparametrizations and statistical identifiability.

\subsection{Reparametrizations and crossed effects models}
In the context of nested models, reparametrization techniques based on hierarchical centering offers a way to make the Gibbs Sampler robust to large datasets (see e.g.\ Corollary \ref{coroll:optimal_param_3levels}). 
We now show that this is not the case in the crossed effects context of Model \ref{model:k_factors}.
In this section we focus on the case $k=2$, which is a case often studied theoretically in the literature (see e.g \citet{GaoOwen2017,brown2018empirical} for recent examples).
In this case, hierarchical centering leads to four possible parametrizations defined as
\begin{equation}\label{eq:centering}
(\mu,\bbeta^{(1)},\bbeta^{(2)})=
(\mu,\ba^{(1)}+(1-\lambda_1) \mu,\ba^{(2)}+(1-\lambda_2) \mu)\,,
\; \hbox{ for }(\lambda_1,\lambda_2)\in\{0,1\}^{2}\,.
\end{equation}
Each parametrization corresponds to a different Gibbs Sampler, which at each iteration updates $\mu$ from $\mathcal{L}(\mu|\bbeta^{(1)},\bbeta^{(2)},\by)$, $\bbeta^{(1)}$ from $\mathcal{L}\big(\bbeta^{(1)}|\mu,\bbeta^{(2)},\by\big)$, and $\bbeta^{(2)}$ from $\mathcal{L}\big(\bbeta^{(2)}|\mu,\bbeta^{(1)},\by\big)$.
The following result characterizes the rate of convergence $\rho_{\lambda_1\lambda_2}$ of such Gibbs Samplers for all combinations $(\lambda_1,\lambda_2)\in\{0,1\}^{2}$.
\begin{theorem}\label{thm:centering_two_factors}
Let $r_1=
\frac{N\precision{e}}{N\precision{e}+n_1\precision{1}}$ and $r_2=
\frac{N\precision{e}}{N\precision{e}+n_2\precision{2}}$. Then we have
\begin{equation}\label{eq:rates_crossed}
\begin{aligned}
\rho_{11}
&=
\max\{r_1,r_2\}
\,,\quad
&&\rho_{10}
=
1-r_2(1-r_1)\,,
\\
\rho_{01}
&=
1-r_1(1-r_2)
\,,\quad
&&\rho_{00}
\geq
1+r_1r_2-\min\{r_1,r_2\}
\}
\,.
\end{aligned}
\end{equation}
\end{theorem}

Figure \ref{fig:rates_crossed} summarizes graphically the results of Theorem \ref{thm:centering_two_factors}, showing
the dependence of the converge rates on the choice of parametrization.
The rate displayed in Figure \ref{fig:rates_crossed} for the fully centred parametrization is 
the lower bound given in \eqref{eq:rates_crossed}.
\begin{figure}[h!]
\centering
\includegraphics[width=\linewidth]{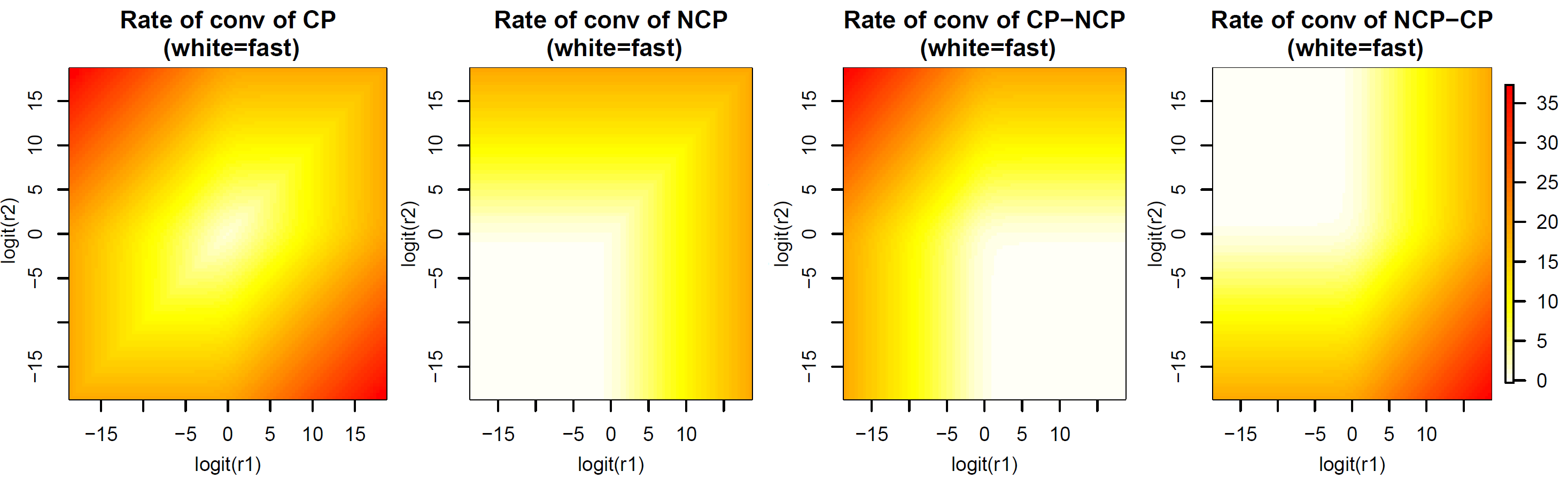}
\caption{Plot of the rates of convergence in \eqref{eq:rates_crossed}.
 Color levels correspond to values of $\log(1-\rho)$, where $\rho$ is the rate of convergence (or its lower bound for $\rho_{11}$), as a function of $\log(r_1/(1-r_1))$ and $\log(r_2/(1-r_2))$.
}
\label{fig:rates_crossed}
\end{figure}

Theorem \ref{thm:centering_two_factors} implies that centering both factors (i.e.\ setting $\lambda_1=\lambda_2=0$) is always computationally worse than any of the other parametrizations because $\rho_{00}\geq \max\{\rho_{11},\rho_{01},\rho_{10}\}$.
On the other hand, the optimal choice of $(\lambda_1,\lambda_2)$ among $(1,1)$, $(0,1)$ and $(1,0)$ depends on the specific values of $r_1$ and $r_2$. 
More precisely, the expressions in \eqref{eq:rates_crossed} imply that the convergence rate is minimized by centering the first factor (i.e. setting $\lambda_1=0$) if and only if $r_1\geq (2-r_2)^{-1}$ and centering the second factor (i.e. setting $\lambda_2=0$) if and only if $r_2\geq (2-r_1)^{-1}$.
These results are in agreement with, for example, the empirical results obtained in \citet[Sec.6]{GelfandSahuCarlin1996} and \citet{browne2004illustration}.

More crucially, Theorem \ref{thm:centering_two_factors} implies that $\min\{\rho_{00},\rho_{01},\rho_{10},\rho_{11}\}\rightarrow 1$ as
$n_1,n_2\rightarrow \infty$.
Therefore, regardless of the parametrizations chosen, the convergence of 
Gibbs Samplers targeting Model \ref{model:k_factors} deteriorate as the number of factors $n_1$ and $n_2$ increase.
This is in contrast with the nested case analysed in Section \ref{sec:multigrid}, where reparametrization techniques are successful in providing samplers with good convergence properties for all choices of hyperparameter values.
In the next section we show that a more effective way to achieve good convergence properties is to impose stronger identifiability constraints.

\subsection{Connections to statistical identifiability}\label{sec:identif}
The parameters $(\mu,\ba^{(1)},\dots,\ba^{(k)})$ in Model \ref{model:k_factors} are not identifiable, in the sense that the mapping $(\mu,\ba^{(1)},\dots,\ba^{(k)})\rightarrow \mathcal{L}(\by|\mu,\ba^{(1)},\dots,\ba^{(k)})$ in not injective.
While this is not strictly speaking an issue for Bayesian inferences, one may wonder whether imposing identifiability on model parameters results in avoiding the degradation of mixing described in previous sections (see e.g.\ \citet{vines1996fitting,gelfand1999identifiability,xie2006measures,kaufman2010bayesian,vallejos2015basics} for related discussion and some examples in applications).
We consider imposing identifiability by conditioning on some linear constraints, such as the commonly used choices of $a^{(s)}_1=0$ or $\ma^{(s)}=0$.
More generally, one can obtain identifiability for Model \ref{model:k_factors} by imposing a linear constraint $c_s=0$ for each $s$ from $1$ to $k$, where
  $c_s=\sum_{j=1}^{n_s}w_{j}^{(s)}a^{(s)}_j$ is a linear combination of $(a_{1}^{(s)},\dots, a_{n_s}^{(s)})$ weighted by some non-negative terms $(w_{1}^{(s)},\dots, w_{n_s}^{(s)})$ satisfying $\sum_{j=1}^{n_s}w_{j}^{(s)}>0$.
Interestingly, one can exploit the multigrid decomposition to derive the convergence rates of the resulting Gibbs Samplers for all choices of weights $(w_{1}^{(s)},\dots, w_{n_s}^{(s)})$.
\begin{theorem}\label{thm:rates_identifiable_general}
The rate of convergence of Sampler \ref{sampler:k_factors_DSGS} conditioned on $c_s=0$ for $s=1,\dots,k$ is given by
\begin{equation}\label{eq:modified_ratios}
\rho=
\max_{s\in\{1,\dots,k\}}\left(\frac{N\precision{e}(1-q_s)}{N\precision{e}+n_s\precision{s}}
\right)
\,,
\end{equation}
where 
$q_s=(\sum_{j=1}^{n_s}w_{j}^{(s)})^2/(n_s\sum_{j=1}^{n_s}(w_{j}^{(s)})^2)$.
\end{theorem}
Comparing \eqref{eq:modified_ratios} with \eqref{eq:crossed_rate} we can see that, since $(1-q_s)\in[0,1)$, the rate of convergence always decreases after imposing the identifiability constraints $c_s=0$ for $s=1,\dots,k$.
Thus, Theorem \ref{thm:rates_identifiable_general} implies that, in this context, imposing identifiability always improves the convergence properties of the Gibbs Sampler. To our knowledge, this is the first rigorous result of this kind in the Bayesian computation literature.
On the other hand, the result also shows that imposing identifiability per se does not guarantee fast convergence. 
For example, Theorem \ref{thm:rates_identifiable_general} implies that
the rate of convergence of Sampler \ref{sampler:k_factors_DSGS} conditioned on $a^{(s)}_1=0$ for each $s\in\{1,\dots,k\}$ is given by
\begin{equation*}
\rho=
\max_{s\in\{1,\dots,k\}}\left(\frac{N\precision{e}}{N\precision{e}+n_s\precision{s}}
\frac{n_s-1}{n_s}\right)
\,,
\end{equation*}
while the rate of convergence of Sampler \ref{sampler:k_factors_DSGS} conditioned on $\ma^{(s)}=0$ for each $s\in\{1,\dots,k\}$ equals 0, i.e. the sampler produces i.i.d.\ draws from the posterior distribution $\mathcal{L}(\mu,\ba|\by, \ma^{(1)}=\dots=\ma^{(k)}=0)$.
While in both cases we observe an improvement over the original Gibbs Sampler in terms of convergence rates, the result shows that conditioning on $a^{(s)}_1=0$ for each $s\in\{1,\dots,k\}$ leads to a convergence rate that can still go to 1 as the datasize increase.
Interestingly \eqref{eq:modified_ratios} implies that the rate of convergence is minimized when $q_s$ is maximized, which happens when the weights in the linear constraints are constant, for example $w_{j}^{(s)}=n_s^{-1}$ for all $s=1,\dots,k$ and $j=1,\dots,n_s$.

\section{Beyond the Gaussian case: a Poisson example}\label{sec:gamma_poisson}
The results of Section \ref{sec:identif} provide guidance on the choice of which linear constraint to use to impose identifiability for models also beyond the Gaussian case. As an example, we consider the following crossed random effect model with Poisson likelihood, which is the simplest analogue of Model \ref{model:k_factors} in the context of count data.
\begin{taggedmodel}{CkP}[Poisson crossed-effects model]\label{model:poisson_crossed}
\begin{align}
y_{i_1\dots i_k}\sim&
\hbox{Poisson}(\mu\, a^{(1)}_{i_1}\cdots a^{(k)}_{i_k})
&i_s=1,\dots,n_s\hbox{ for }s\in \{1,\dots,k\}
\end{align}
with $a^{(s)}_{i_s}\stackrel{iid}\sim Gamma(\alpha_s,\beta_s)$ for $s=1,\dots,k$ and 
$\mu\sim Gamma(\alpha_\mu,\beta_\mu)$.
\end{taggedmodel}
We consider sampling from the posterior distribution $\mathcal{L}(\mu,\ba|\by)$ of Model \ref{model:poisson_crossed} using
the standard Gibbs Sampler that, similarly to  
Sampler \ref{sampler:k_factors_DSGS}, 
at each iteration updates $\mu$ from $\mathcal{L}(\mu|\ba,\by)$ and then $\ba^{(s)}$ from $\mathcal{L}\left(\ba^{(s)}|\mu,\ba^{(-s)},\by\right)$ for $s=1,\dots,k$.
Here $\by$, $\ba$ and $\ba^{(-s)}$ are defined as in the beginning of Section \ref{sec:crossed}. 

We explore the extent to which the conclusions drawn from Theorem \ref{thm:rates_identifiable_general} apply also to Model \ref{model:poisson_crossed} by means of simulations. 
We consider the case $k=2$ with three different combinations of values of $n_1$ and $n_2$, with data $(y_{i_1\dots i_k})_{i_1\dots i_k}$ simulated from the model. For the prior hyperparameters we use $\alpha_1=\alpha_2=\alpha_\mu=2$ and $\beta_1=\beta_2=\beta_\mu=0.1$.
We compare the standard Gibbs Sampler with no constraints, with the versions obtained by imposing the linear constraint $a^{(1)}_1=a^{(2)}_1=1$ and 
$\ma^{(1)}=\ma^{(2)}=1$, respectively, where $\ma^{(1)}$ and $\ma^{(2)}$ are defined as in \eqref{eq:residuals_crossed}.
Table \ref{table:GP} reports the resulting  effective same sizes for various parameters (namely $\mu$, the average effective sample size of $a^{(1)}_{i}$ for $i=2,\dots,n_1$ and similarly for $a^{(2)}_{j}$ over $j=2,\dots,n_2$).
\begin{table}[ht]
\centering
\begin{tabular}{c|ccc|ccc|ccc}
\hline
  & \multicolumn{3}{c|}{$n_1=5$, $n_2=5$} & \multicolumn{3}{c|}{$n_1=5$, $n_2=100$} & \multicolumn{3}{c}{$n_1=100$, $n_2=100$} \\ 
 & $\mu$ & $a^{(1)}_{i_1}$ & $a^{(2)}_{i_2}$ & $\mu$ & $a^{(1)}_{i_1}$ & $a^{(2)}_{i_2}$ & $\mu$ & $a^{(1)}_{i_1}$ & $a^{(2)}_{i_2}$ \\ 
  \hline
Unconstrained & 273 & 287 & 268 & 273 & 314 & 280 & 287 & 351 & 255 \\ 
  $a^{(1)}_1=a^{(2)}_1=1$ & 1483 & 2741 & 1656 & 506 & 33614 & 507 & 372 & 553 & 300 \\ 
$\ma^{(1)}=\ma^{(2)}=1$
 & 52405 & 52537 & 49847 & 46404 & 49634 & 50134 & 43869 & 50682 & 51259 \\ 
   \hline
\end{tabular}
\caption{Effective sample sizes for the standard unconstrained Gibbs Sampler for Model \ref{model:poisson_crossed}, and for the two version where identifiability is obtained by imposing the constraints $a^{(1)}_1=a^{(2)}_1=1$ and $\ma^{(1)}=\ma^{(2)}=1$, respectively.
Effective sample sizes 
 correspond to $10^5$ iterations of each algorithm, with the first half of the samples discarded as burn-in.
}
\label{table:GP}
\end{table}
The results in Table \ref{table:GP} are very much coherent with the theoretical results obtained in the Gaussian case in Theorem \ref{thm:rates_identifiable_general}. In particular, imposing identifiability always improves mixing of the Gibbs Sampler, and imposing constraints on $\ma^{(1)}$ and $\ma^{(2)}$ leads to a sampler with faster convergence compared to imposing constraints on $a^{(1)}_1$ and $a^{(2)}_1$.
Moreover, the difference in resulting efficiency between the two set of linear constraints increases with $n_1$ and $n_2 $, as suggested by Theorem \ref{thm:rates_identifiable_general}. 

\subsection{Comparison with Hamiltonian Monte Carlo}\label{sec:HMC}
Finally, we also explore whether the results in Theorem \ref{thm:rates_identifiable_general} are useful also to guide the implementation of other MCMC schemes targeting Model \ref{model:poisson_crossed}, such as Hamiltonian Monte Carlo (HMC) \citep{neal2011mcmc} and the No-U-Turn Sampler (NUTS) \citep{hoffman2014no} implemented in the widely used software STAN \citep{carpenter2017stan}.
We consider the same setting of the rightmost column of Table \ref{table:GP} where $n_1=n_2=100$.
\begin{table}[ht]
\centering
\begin{tabular}{c|ccc|c|ccc}
\hline
  & \multicolumn{3}{c|}{Effective Sample Size (ESS)} & \multicolumn{1}{c|}{Runtime} & \multicolumn{3}{c}{ESS/time [1/s]}\\ 
 & $\mu$ & $a^{(1)}_{i_1}$ & $a^{(2)}_{i_2}$ &  & $\mu$ & $a^{(1)}_{i_1}$ & $a^{(2)}_{i_2}$ \\ 
  \hline
HMC (unconstrained)& 		530 & 545 & 518 & 					5697s & 0.09 & 0.10 & 0.09 \\ 
HMC ($a^{(1)}_1=a^{(2)}_1=1$) & 	1661 & 1611 & 3907 & 		1737s & 0.96 & 0.93 & 2.2 \\ 
HMC ($\ma^{(1)}=\ma^{(2)}=1$) & 1369 & 1263 & 1897 & 1459s & 0.94 & 0.87 & 1.3 \\ 
\hline
  NUTS (unconstrained)& 	30 & 18 & 28 & 					906s & 0.03 & 0.02 & 0.03 \\ 
  NUTS ($a^{(1)}_1=a^{(2)}_1=1$) & 	108 & 176 & 696 & 		312s & 0.35 & 0.57 & 2.2 \\ 
  NUTS ($\ma^{(1)}=\ma^{(2)}=1$) & 	8000 & 19914 & 20062 & 	134s & 60 & 149 & 150 \\ 
   \hline
Gibbs (unconstrained) & 	
21 & 27 & 4.3 & 	0.91s & 
24 & 30 & 4.7\\ 
Gibbs ($a^{(1)}_1=a^{(2)}_1=1$) & 	
   27 & 27 &  5.6 &
    	0.92s & 
29 & 30 & 6.1\\ 
   Gibbs ($\ma^{(1)}=\ma^{(2)}=1$) & 	
8000 & 7992 & 8037 & 	1.00s & 
8000 & 7992 & 8037\\ 
   \hline
\end{tabular}
\caption{Comparison of HMC, NUTS and the Gibbs Sampler for Model \ref{model:poisson_crossed} with and without linear constraints. 
The values of the effective sample sizes and the runtimes refer to a single run of each algorithm  
for $10^4$ iterations, with the first 2000 iterations discarded as burn-in. 
}
\label{table:GP_stan}
\end{table}
Table \ref{table:GP_stan} reports effective sample sizes (ESS), runtime and ESS per unit of computation time for HMC and NUTS, and compare those to the Gibbs Sampler one (see also the supplementary material for traceplots and autocorrelation functions).
Table \ref{table:GP_stan} suggests that imposing identifiability through linear constraints as suggested by Theorem \ref{thm:rates_identifiable_general} helps significantly also gradient-based sampler such as HMC and NUTS, both by
speeding up convergence (higher ESS) and by reducing the cost per iteration (lower runtime).
We expect the reduction in runtime for HMC and NUTS to arise from the adaptation of the tuning parameters performed by STAN (e.g.\ step-size and number of leapfrog steps within each iteration) so that for better identified and less correlated targets (as the one with the linear constraints) the number of required leapfrog steps per iteration is lower, leading to a reduction in runtime.
Overall, imposing identifiability results in a significantly higher sampling efficiency: for example, when the constraints $\ma^{(1)}=\ma^{(2)}=1$ is imposed, NUTS is over three orders of magnitude more efficient than in the unconstrained version.
Also, especially for NUTS, the constraints $\ma^{(1)}=\ma^{(2)}=1$ lead to a more efficient sampler than the ones with $a^{(1)}_1=a^{(2)}_1=1$, which is analogous to the results of 
 Theorem \ref{thm:rates_identifiable_general}. 
In general, Table \ref{table:GP_stan} supports the fact that the results of Theorem \ref{thm:rates_identifiable_general} can also provide useful guidance to derive significantly more efficient implementations of gradient-based MCMC algorithms.
Finally, note that the runtime of the Gibbs Sampler is orders of magnitude lower than the one HMC and NUTS, suggesting that for random effect models such as Model \ref{model:poisson_crossed} gradient-based schemes can be much more costly than Gibbs-type schemes, which exploit more directly the conditional independence among random effects.
We leave a more detailed investigation of these aspects in the context of more general and complex models to future work.

All simulations reported in Tables  \ref{table:GP} and \ref{table:GP_stan} were performed on the same desktop computer with 16GB of RAM and an 
Intel core i7-7700 @ 3.60 GHz processor, using the R programming language \citep{R}.
Effective sample sizes are estimated using the \emph{mcmcse} R package \citep{mcmcse}.
The supplementary material provides the R code used to implement the Gibbs Samplers and the Stan code used to specify the models.
\begin{rmk}\label{rmk:gamma_poisson}
Interestingly, the multigrid decomposition can be applied also to Model \ref{model:poisson_crossed}, with the appropriate modifications. 
In this case the Markov chain $\left((\mu,\ba)(t)\right)_{t=1}^\infty$
induced by the Gibbs Sampler can be decomposed into $(k+1)$ independent Markov chains 
$
\left((\mu,\tilde{a}^{(1)},\dots,\tilde{a}^{(k)})(t)\right)_{t=1}^\infty$
and
$\big(\tilde{\delta}\ba^{(1)}(t)\big)_{t=1}^\infty$, \dots, $\big(\tilde{\delta}\ba^{(k)}(t)\big)_{t=1}^\infty$,
where
$\tilde{a}^{(s)}=\sum_{i_s}a^{(s)}_{i_s}$ and
$\tilde{\delta} a^{(s)}_{i_s}=a^{(s)}_{i_s}/\tilde{a}^{(s)}$.
In this case the rate of convergence of the original chain coincides with the one of $\left((\mu,\tilde{a}^{(1)},\dots,\tilde{a}^{(k)})(t)\right)_{t=1}^\infty$, which evolves according to a $(k+1)$-dimensional Gibbs Sampler with full conditionals given by:
\begin{equation}\label{eq:GS_poisson_crossed}
\begin{aligned}
\mu|\by,\tilde{a}
&\sim
Gamma(\alpha_\mu+y_{\cdot},\beta_\mu+\prod_{s=1}^k \tilde{a}^{(s)} )\,,
\\
\tilde{a}^{(s)}|\by,\mu,\tilde{a}^{(-s)}
&\sim
Gamma(I\alpha_s+y_{\cdot},\beta_s+\mu \prod_{\ell\neq s}^k \tilde{a}^{(\ell)})
\qquad\hbox{for } s\in\{1,\dots,k\}\,,
\end{aligned}
\end{equation}
where $y_{\cdot}=\sum_{i_1,\dots,i_k}y_{i_1\dots i_k}$.
We expect such a $(k+1)$-dimensional Gibbs Sampler to be potentially amenable to analysis using the framework of iterated random functions \citep{diaconis1999iterated},
in order to obtain an upper bound on convergence rates (see e.g.\ \citealp[Theorem 2.1.(b)]{alsmeyer2001limit}).
We leave these extensions to future works and mention it in Section \ref{sec:conclusion} as a possible avenue for future research directions.
\end{rmk}

\section{Non-symmetric hierarchical models}\label{sec:bespoke}
Section \ref{sec:multigrid} describes how to choose the optimal parametrization as a function of $(I,J,K,\sigma_{a},\sigma_b,\sigma_{e})$ for Model \ref{model:S3}.
In general, both the variance terms $\var{b}$ and $\var{e}$, and the number of branches $J$ and $K$ in the hierarchy could depend on $i$ and $j$.
In this section we consider such non-symmetric cases for two and three level hierarchical models.
In these non-sym\-me\-tric cases the computationally optimal strategy will involve centering some branches of the hierarchy and non-centering others: we will refer to these strategies as \emph{bespoke parametrizations}.

Consider the following non-symmetric 2-levels model (which we describe in terms of precisions rather than variances for notational convenience).
\begin{taggedmodel}{NS2}[Non-symmetric 2-levels hierarchical model]\label{model:2_nonsymmetric}
Consider the following 2-levels model with centred parametrization
\begin{align*}
\qquad p(\mu)&\propto 1&\\
\qquad \gamma_i&\sim N(\mu,1/\precision{a}) &&i=1,\dots,I,\\
\qquad y_{ij}&\sim N(\gamma_{i},1/\precision{e,i}) &&j=1,\dots,J_{i},
\end{align*}
where the precision components $(\precision{a},(\precision{e,i})_i)$ are assumed to be known.
\end{taggedmodel}
\cite{RobertsPapaspiliopoulosSkold2003}
studied the symmetric version of Model \ref{model:2_nonsymmetric},
where $J_i=J$ and 
$\precision{e,i}=\precision{e}$ for all 
$i$ and some fixed $J$ and $\precision{e}$.
They showed that the rates of convergence induced by the centred and non-centred parametrizations are given respectively by 
\begin{equation}\label{eq:rate_2_symmetric}
\rho_{CP}=\frac{\precision{a}}{\precision{a}+\tprec{e}}
\quad\hbox{ and }\quad
\rho_{NCP}=\frac{\tprec{e}}{\precision{a}+\tprec{e}}
\,,
\end{equation}
where $\tprec{e}=J\precision{e}$.
The following Theorem provides an extension to the general non-symmetric case.
We consider Model \ref{model:2_nonsymmetric} with a bespoke parametrization 
$(\mu,\beta_1,\dots,\beta_I)$ defined by $I$ indicators $(\lambda_1,\dots,\lambda_I)\in\{0,1\}^I$ as $\beta_i=\gamma_i-\lambda_i\mu$, meaning that $\lambda_i$ equals 0 if component $i$ is centred and 1 if it is non-centred.
\begin{theorem}\label{thm:bespoke_parametrizations_2}
The rate of convergence of the Gibbs Sampler targeting Model \ref{model:2_nonsymmetric} with bespoke parametrization given by $(\lambda_1,\dots,\lambda_I)\in\{0,1\}^I$ is
\begin{equation}\label{eq:rate_2_heterogenous}
\rho_{\lambda_1\dots \lambda_I}
=
\frac{
\sum_{i\,:\,\lambda_i=1}\tprec{i}\frac{\tprec{i}}{\tprec{i}+\precision{a}}
+
\sum_{i\,:\,\lambda_i=0}\precision{a}\frac{\precision{a}}{\tprec{i}+\precision{a}}
}{
\sum_{i\,:\,\lambda_i=1}\tprec{i}+\sum_{i\,:\,\lambda_i=0}\precision{a}
}\,,
\end{equation}
where $\tprec{i}=J_i\precision{e,i}$.
\end{theorem}

Equation \eqref{eq:rate_2_heterogenous} shows that in the non-symmetric case, the GS rate of convergence is given by a weighted average of the precision ratios $\frac{\precision{a}}{\tprec{i}+\precision{a}}$ and $\frac{\tprec{i}}{\tprec{i}+\precision{a}}$ depending on whether each component is centred or not.
This has clear analogies with the symmetric case in \eqref{eq:rate_2_symmetric}.
The weights in the average of \eqref{eq:rate_2_heterogenous} are themselves a function of $(\lambda_1,\dots,\lambda_I)$, thus introducing dependence across components in terms of centering and the overall convergence rate.
Nonetheless, the following corollary shows that even in the context of Model \ref{model:2_nonsymmetric},
the choice of optimal parametrization in each branch of the three can be carried out independently following the same intuition of the symmetric case: 
for each $i$ in $\{1,\dots,I\}$ use centred parametrization $\gamma_i$ if and only if $\precision{a}\leq J_i\precision{e,i}$, otherwise use a non-centred parametrization $a_i=\gamma_i-\mu$.
Note that the optimal choice on each branch of the hierarchy can be taken independently of other branches, which make the implementation easy (compared to a scenario where the optimal decision on each branch was influenced by other branches).
\begin{coroll}\label{coroll:optimal_param_2_levels}
Let $\bar{\lambda}_i=\1(\precision{a}>\tprec{i})$ for all $i$ from 1 to $I$. Then 
$$
\rho_{\bar{\lambda}_1\dots \bar{\lambda}_{I}}
\leq
\rho_{\lambda_1\dots \lambda_I}
\qquad\hbox{for any }(\lambda_1\dots \lambda_I)\in\{0,1\}^I\,.
$$
\end{coroll}
By \eqref{eq:rate_2_heterogenous}, the strategy described in Corollary \ref{coroll:optimal_param_2_levels} ensures that $
\rho_{\bar{\lambda}_1\dots \bar{\lambda}_{I}}
\leq 1/2$.
This is the same upper bound one can obtain in the symmetric case (see \eqref{eq:rate_2_symmetric} and \cite{RobertsPapaspiliopoulosSkold2003}), meaning that in this case bespoke parametrizations are successful in dealing with the lack of symmetry.

Consider now the three-level non-symmetric case.
\begin{taggedmodel}{NS3}[Non-symmetric 3-levels hierarchical model]\label{model:3_nonsymmetric}
Consider a more general 3-levels model with centred parametrization
\begin{align*}
\qquad p(\mu)&\propto 1&\\
\qquad \gamma_i&\sim N(\mu,\var{a}) &&i=1,\dots,I,\\
\qquad \eta_{ij}&\sim N(\gamma_i,\var{b,i}) &&j=1,\dots,J_i,\\
\qquad y_{ijk}&\sim N(\eta_{ij},\var{e,ij}) &&k=1,\dots,K_{i,j},
\end{align*}
where variance components are assumed to be known.
\end{taggedmodel}
In this case the multigrid factorization of Theorem \ref{thm:factorization_3} does not apply directly to Model \ref{model:3_nonsymmetric}, but nonetheless it can still be used to obtain upper bounds on the rates of convergence.
\begin{theorem}\label{thm:rates_3_uneven_CP}
Given an instance of Model \ref{model:3_nonsymmetric} we define
$$
\ratio{a}{b}^{(i)}
=
\frac{\var{a}}{\var{a}+J_i^{-1}\var{b}},
\quad\hbox{and}\quad
\ratio{e}{b}^{(i)}=\frac{1}{J_i}\sum_{j=1}^{J_i}
\frac{K_{ij}^{-1}\var{e,ij}}{\var{b,i}+K_{ij}^{-1}\var{e,ij}}\,.
$$
If $\ratio{a}{b}^{(i)}\geq\ratio{a}{b}^{(i')}\ratio{e}{b}^{(i')}$ for every $i,i'\in\{1,\dots,I\}$,
 then the rate of convergence of the Gibbs Sampler with centred parametrization $(\mu,\gavec,\etavec)$ satisfies
$$
\rho
\;\leq\;
1-
\frac{1}{I}
\sum_{i=1}^{I}
\ratio{a}{b}^{(i)}
+
\max_{i=1,\dots,I}
\ratio{a}{b}^{(i)}
\ratio{e}{b}^{(i)}\,.
$$
\end{theorem}
The results of Theorem \ref{thm:rates_3_uneven_CP} suggest that as the number of datapoints increase the efficiency of the Gibbs sampler with centred parametrization increases.
In fact, as $K_{ij}$ increases the assumptions of Theorem \ref{thm:rates_3_uneven_CP} are eventually satisfied and the bound on the convergence rate goes to $0$ as $J_i$ and $K_{ij}$ increase.
Theorem \ref{thm:rates_3_uneven_CP} provides rigorous theoretical support and characterization of the well known fact that the centred parametrization is to be preferred in contexts of large and informative datasets \citep{GelfandSahuCarlin1995,RobertsPapaspiliopoulosSkold2003}.
We note that the convergence rate for the Gibbs Sampler targeting Model \ref{model:3_nonsymmetric} is not easily tractable, and that deriving analytic expressions for the optimal bespoke parametrization in this context is still an open problem.
\section{Hierarchical linear models with arbitrary number of levels}\label{sec:k_levels}
In this section we consider Gaussian hierarchical models with an arbitrary number of levels, namely $k$ levels.
We refer to the highest level of the hierarchy (i.e. the one furthest away from the data) as level 0, down to level $k-1$ being the lowest level (i.e. closest to the data).
The 3 level model of Section \ref{sec:multigrid} is a special case of the theory developed here where $k=3$.

\subsection{Model formulation}
In order to allow for more generality and keep the notation concise, in this section we will use a graphical models notation.
In particular $T$ will denote a finite tree with $k$ levels and root $t_0\in T$. 
For each node $t\in T$ we will denote by $pa(t)$ the unique parent of $t$ and by $ch(t)$ the collection of children of $t$. 
Moreover we write $s\preceq t$ and $s\succeq t$ if $s$ is respectively an ancestor or a descendant of $t$ (with $s$ and $t$ possibly being equal) while $s\prec t$ and $s\succ t$ denote the same notions with the additional condition of $s\neq t$.
For each node $t\in T$ we denote by $\ell(t)$ the level of node $t$ (i.e. its distance from $t_0$).
For each $d\in\{0,\dots,k-1\}$ we denote by $T_d=\{t\in T\,:\,\ell(t)=d\}$ the collection of nodes at level $d$. For example we have $T_0=\{t_0\}$ and $T=\cup_{d=0}^{k-1} T_d$.
Noisy observations will occur only at leaf nodes.
The collection of leaf nodes is denoted as $T_L=\{t\in T\,:\,ch(t)=\emptyset\}$.
For simplicity we assume that all leaf nodes are at level $k-1$, i.e. $T_L=T_{k-1}$, although this assumption could be easily relaxed allowing some branches to be longer than others.

\begin{taggedmodel}{NSk}[$k$-levels hierarchical model]\label{model:NSk} 
Suppose that we have a hierarchy described by a tree $T$ with observations occurring at leaf nodes 
$T_L$.
We assume the following hierarchical model
\begin{align}
y^{(i)}_{t}&\sim N(\gamma_t,
1/\precision{t}^{(e)}
)
&t\in T_L\,,\\
\gamma_{t}&\sim N(\gamma_{pa(t)},
1/\precision{t}
)
&t\in T\backslash{t_0},
\end{align}
where $i\in\{1,\dots,n_t\}$ with $n_t$ being the number of observed data at leaf node $t$, $(\precision{t})_{t\in T\backslash t_0}$ and $(\precision{t}^{(e)})_{t\in T_L}$ are known precision components and all normal random variables are sampled independently.
Following the standard Bayesian model specification we assume a flat prior on $\gamma_{t_0}$.
\end{taggedmodel}
We are interested in sampling from the posterior distribution of $\bgamma_T=(\gamma_t)_{t\in T}$ given some observations $\textbf{y}=(y_t)_{t\in T_L}$.
The centred parametrization $\bgamma_T$ of Model \ref{model:NSk} leads to the following Gibbs Sampler.
\begin{taggedsampler}{GS($\bgamma_T$)}\label{sampler:GSkCP}
Initialize $\bgamma_T(0)$ and then iterate the following kernel:\\
For $d=0,\dots,k-1$, 
sample $\gamma_t(s+1)$ from 
$p(\gamma_t|\bgamma_{T_{d-1}}(s+1),\bgamma_{T_{d+1}}(s),\yvec)$
for all $t\in T_d$,
where $p(\gamma_t|\bgamma_{T_{d-1}},\bgamma_{T_{d+1}},\yvec)=p(\gamma_t|\bgamma_{T\backslash t},\yvec)$ is the full conditional distribution of $\gamma_t$ given by Model \ref{model:NSk}.
When $d$ equals 0 or $k-1$ the levels $\bgamma_{T_{d-1}}$ and $\bgamma_{T_{d+1}}$ have to be replaced by empty sets in the conditioning. 
\end{taggedsampler}

\subsection{Non centering and hierararchical reparametrizations}\label{sec:PNCP_k_levels}
Model \ref{model:NSk}  expresses Gaussian hierarchical models using a centred pa\-ra\-me\-tri\-za\-tion.
The corresponding non-centred version is given by the following example.
\begin{ex}[Fully non-centred parametrization]\label{ex:NCPk}
Under the same setting as Model \ref{model:NSk}, define
\begin{align*}
y^{(i)}_{t}&\sim N\Big(\sum_{r\preceq t}\alpha_r,
1/\precision{t}^{(e)}
\Big)
&t\in T_L,\\
\alpha_{t}&\sim N(0,
1/\precision{t}
)
&t\in T\backslash{t_0},
\end{align*}
and assume a flat prior on $\alpha_{t_0}$.
\end{ex}
The non-centred parametrization $\balpha_T$ can be obtained as a linear transformation of the centred version $\bgamma_T$ of Model \ref{model:NSk}.
More generally, we will consider the class of parametrizations that can be obtained by reparametrizing  $\bgamma_T$ as follows.
\begin{defi}[Hierarchical reparametrizations]\label{defi:hierarchical_reparametrization}
Let $\bgamma_T=(\gamma_t)_{t\in T}$ be a random vector with elements indexed by a tree $T$.
We say that $\bbeta_T=(\beta_t)_{t\in T}$ is a hierarchical (linear) reparametrization of $\bgamma_T$ if
\begin{equation}\label{eq:hierarchical_reparametrization}
\beta_t=\sum_{r\preceq t}\lambda_{tr } \gamma_r\,\qquad t\in T,
\end{equation}
for some real-valued coefficients $\Lambda=(\lambda_{tr})_{r\preceq t,t\in T}$ satisfying $\lambda_{tt}\neq 0$ for all $t\in T$.
We denote \eqref{eq:hierarchical_reparametrization} by $\bbeta_T=\Lambda\bgamma_T$.
\end{defi}
Using terminology from \citet{RobertsPapaspiliopoulosSkold2003}, 
we refer to the family of hierarchical reparametrizations of $\gavec_T=(\gamma_t)_{t\in T}$ as \emph{partially non-centred parametrizations} (PNCP) of Model \ref{model:NSk}.
Note that \eqref{eq:hierarchical_reparametrization} does not span the space of all linear transformations of $\bgamma_T$.
In fact $\Lambda=(\lambda_{tr})_{r\preceq t,t\in T}$ can be thought as a $|T|\times|T|$ matrix $\Lambda=(\lambda_{tr})_{r,t\in T}$ inducing a linear transformation of $\bgamma_T$ with the additional sparsity constraint of being zero on all elements $\lambda_{tr}$ such that $r\npreceq t$.
The following Lemma shows that the definition of the class of PNCP does not depend on the starting parametrization used to formulate Model \ref{model:NSk}.
For example, one could equivalently define the class of PNCP of Model \ref{model:NSk} as the collection of hierarchical reparametrization of the non-centred parametrization $\balpha_T$ of Example \ref{ex:NCPk}.
\begin{lemma}\label{lemma:hierarchical_reparametrization_invertible}
If $\bbeta_T$ is a hierarchical reparame\-tri\-za\-tion of $\bgamma_T$, then also $\bgamma_T$ is a hierarchical reparametrization of $\bbeta_T$.
\end{lemma}

As for the 3-levels case we are interested in assessing the computational efficiency of the different Gibbs Sampling schemes arising from different PNCP's.
For each PNCP $\bbeta_T$ the corresponding Gibbs Sampler scheme $GS(\bbeta_T)$ is defined analogously to $GS(\bgamma_T)$.
\begin{taggedsampler}{GS($\bbeta_T$)}\label{sampler:GSk}
Initialize $\bbeta_T(0)$ and then iterate the following kernel:\\
For $d=0,\dots,k-1$, sample $\beta_t(s+1)$ from 
$
p(\beta_t|(\bbeta_{T_p}(s+1))_{0\leq p<d},(\bbeta_{T_p}(s))_{d<p\leq k-1},\yvec)
$
for all $t\in T_d$,
where $p(\beta_t|(\bbeta_{T_p})_{0\leq p<d},(\bbeta_{T_p})_{d<p\leq k-1},\yvec)=p(\beta_t|\bbeta_{T\backslash t},\yvec)$ is the full conditional distribution of $\beta_t$ given by Model \ref{model:NSk}.
\end{taggedsampler}

Sampler \ref{sampler:GSk} is easy to implement because the family of PNCP preserves the hierarchical structure of Model \ref{model:NSk}.
In fact, any PNCP of Model \ref{model:NSk} exhibits the following conditional independence structure:
\begin{equation}\tag{H}\label{eq:hierarchical_structure}
\beta_r\bot\beta_t|\bbeta_{T\backslash\{r,t\}}
\hbox{ unless }
r\preceq t
\hbox{ or }
t\preceq r\,.
\end{equation}
Note that this is a weaker condition than the one holding for the centred parametrization $\bgamma_T$. In the latter case, the conditional independence graph corresponds exactly to the tree $T$, meaning that if $r\neq t$
\begin{equation}\tag{T}\label{eq:tree_cond_indep}
\gamma_r\bot\gamma_t|\bgamma_{T\backslash\{r,t\}} 
\hbox{ unless }
r=pa(t)
\hbox{ or }
t=pa(r)\,.
\end{equation}
Despite being weaker than \eqref{eq:tree_cond_indep}, condition \eqref{eq:hierarchical_structure} still guarantees that all parameters at the same level are conditionally independent (thus simplifying the update of $\bbeta_{T_d}|\bbeta_{T\backslash T_d}$) and that the full conditional distribution of each $\beta_t$ depends only on the descendants or ancestors of $t$.
The following Lemma and Corollary provide a simple way to check that any PNCP of Model \ref{model:NSk} satisfies \eqref{eq:hierarchical_structure}.
\begin{lemma}\label{lemma:H_closed}
Property \eqref{eq:hierarchical_structure} is closed under hierarchical re-parametrizations, meaning that if $\bbeta_T$ satisfies \eqref{eq:hierarchical_structure} then any hierarchical re-parametrization of $\bbeta_T$ satisfies \eqref{eq:hierarchical_structure} too.
\end{lemma}

\begin{coroll}\label{coroll:PNCP}
Any PNCP $\bbeta_T$ of Model \ref{model:NSk} satisfies \eqref{eq:hierarchical_structure}.
\end{coroll}

\subsection{Symmetry assumption}\label{sec:symmetry}
To provide a full analysis of the convergence properties of Sampler \ref{sampler:GSk} we need a symmetry assumption that we now define.
Let $\rho_{tr}$ denote the partial correlation $Corr\left(\beta_t,\beta_r\Big | \bbeta_{T\backslash\{t,r\}}\right)$, namely
\begin{align*}
\rho_{tr}&=-\frac{Q_{tr}}{\sqrt{Q_{tt}Q_{rr}}}&t\neq r\,,
\end{align*}
and $\rho_{tt}=1$ for all $t$. Here $Q$ is the precision matrix of $\bbeta_T$.
Let $\textbf{X}=(X_\ell)_{\ell=0}^{k-1}$ be a random walk going through $T$ from root to leaves as follows: $X_0=t_0$ almost surely and then, for $\ell\in\{0,\dots,k-2\}$
\begin{equation}\label{eq:non_symm_markov}
P(X_{\ell+1}=t\,|\,X_{\ell}=r)=\frac{\rho^2_{tr}}{\sum_{t'\in ch(r)}\rho^2_{t'r}}\1(t\in ch(r))\,.
\end{equation}
Equation \eqref{eq:non_symm_markov} implies that at each step the Markov chain $\textbf{X}$ jumps from the current state $r$ to one of its children $t\in ch(r)$ choosing $t$ proportionally to the squared partial correlation between $\beta_r$ and $\beta_t$.
Since $\ell(X_d)=d$ almost surely for all $d\in\{0,\dots,k-1\}$ we can use the following simplified notation: for any $t$ and $r$ in $T$ we use $P(t)$, $P(t|r)$ and $P(t\cap r)$ to denote respectively  $P(X_{\ell(t)}=t)$, $P(X_{\ell(t)}=t\,|\,X_{\ell(r)}=r)$ and $P(X_{\ell(t)}=t\,\cap\,X_{\ell(r)}=r)$.

Given the above definitions, we define the following symmetry condition: there exist a $k\times k$ symmetric matrix $C=(c_{dp})_{d,p=0}^{k-1}$ such that
\begin{align}\tag{S}\label{eq:k_symmetry}
\rho_{tr}
=
&
c_{\ell(r)\ell(t)}\sqrt{P(t|r)}
& r\preceq t\,,
\end{align}
and $\rho_{tr}=0$ if $r\npreceq t$ and $t\npreceq r$.
Note that $\rho_{tr}$ is invariant to coordinate-wise rescaling of $\bbeta_T$ and therefore both property \eqref{eq:k_symmetry} and the transition kernel of $\textbf{X}$ are invariant to rescalings.
Therefore we can consider, without loss of generality, the following rescaled version of $\bbeta_T$ defined by 
\begin{align}\label{eq:rescaling}
\tbeta_t&=\beta_t\sqrt{\frac{Q_{tt}}{P(t)}}&t\in T\,.
\end{align}
Given \eqref{eq:rescaling}, condition \eqref{eq:k_symmetry} can be written, in terms of the precision matrix of $\tbbeta_T=(\tbeta_t)_{t\in T}$ as
\begin{align}\tag{$\widetilde{\hbox{S}}$}\label{eq:k_symmetry_rescaled}
\tilde{Q}_{tt}=P(t)
\quad\hbox{and}\quad
-\tilde{Q}_{tr}
=&
c_{\ell(t)\ell(r)}P(t \cap r)
\quad \hbox{for }t\neq r
\,.
\end{align}
The rescaled version $\tbbeta_T$ will be helpful later to design an appropriate multigrid decomposition of $\bbeta_T$.
Also, it can be seen that property \eqref{eq:k_symmetry_rescaled} is closed under symmetric hierarchical parametrizations.
\begin{defi}[Symmetric hierarchical reparametrizations]\label{defi:symmetric_reparametrization}
We say that $\bbeta_T=\Lambda\balpha_T$ is a symmetric hierarchical reparametrization of $\balpha_T$ if the coefficients of $\Lambda=(\lambda_{tr})_{r\preceq t,t\in T}$  depend only on the levels of $r$ and $t$ in the hierarchy $T$.
\end{defi}
\begin{lemma}\label{lemma:S_closed} 
Property \eqref{eq:k_symmetry_rescaled} is closed under symmetric hierarchical repa\-ra\-me\-tri\-zations, meaning that if $\tbbeta_T$ satisfies \eqref{eq:k_symmetry_rescaled} then any symmetric hierarchical repara\-me\-trization of $\tbbeta_T$ satisfies \eqref{eq:k_symmetry_rescaled} too.
\end{lemma}
Various special cases of Model \ref{model:NSk} satisfy assumption \eqref{eq:k_symmetry}.
For example, we now consider three cases: a fully symmetric case (both the tree $T$ and the variances $(\tau_t)_{t\in T}$ are symmetric), a weakly symmetric case (non-symmetric tree and symmetric variances) and a non-symmetric case (both tree and variances non-symmetric).

\begin{taggedmodel}{Sk}[Symmetric $k$-levels hierarchical model]\label{model:Sk} 
Consider the $k$-level Gaussian Hierarchical model where the observed data are generated from
\begin{align*}
y_{i_1,\dots,i_{k-1},j}&\sim N(\gamma^{(k-1)}_{i_1,\dots,i_{k-1}},1/\tau_e)
&&(i_1,\dots,i_{k-1},j)\in[I_1]\times\dots\times[I_{k-1}]\times[J]\,,
\end{align*}
where $[N]=\{1,\dots,N\}$ for any positive integer $N$.
The parameters have the following hierarchical structure: for each level $d$ from $1$ to $k-1$
\begin{align*}
\gamma^{(d)}_{i_1,\dots,i_d}&\sim N(\gamma^{(d-1)}_{i_1,\dots,i_{d-1}},1/\tau_d) 
&&(i_1,\dots,i_d)\in[I_1]\times\dots\times[I_d]\,,
\end{align*}
Here $(\precision{1},\dots,\precision{k-1},\precision{e})$ are known precisions and the root parameter $\gamma^{(0)}$ is given a flat prior $p(\gamma^{(0)})\propto 1$.
For each $d\in\{1,\dots,k-1\}$ the positive integer $I_d$ represents the number of branches from level $d-1$ to level $d$.
\end{taggedmodel}
It is easy to see that the posterior distribution of Model \ref{model:Sk},
conditioned on the observed data $\yvec=(y_{i_1,\dots,i_{k-1},j})_{i_1,\dots,i_{k-1},j}$, 
satisfies \eqref{eq:k_symmetry}.
In this case the random walk $\textbf{X}$ defined by \eqref{eq:non_symm_markov} coincides with the natural random walk going through $T$.

\begin{ex}[Weakly symmetric case]\label{ex:weakly_symm}
Another special case of Model \ref{model:NSk} satisfying \eqref{eq:k_symmetry} is given by the case of a general tree $T$ and precision terms defined as $\tau_t=\frac{\tau_{\ell(t)}}{\prod_{s\prec t}|ch(s)|}$ for all $t\in T$ and  $\tau^{(e)}_t=\frac{\tau_e}{n_t\prod_{s\prec t}|ch(s)|}$ , where $(\tau_1,\dots,\tau_{k},\tau_e)\in\R_+^{k+1}$ are level-specific precision terms.
This is an extension of Model \ref{model:Sk} where the lack of symmetry of $T$ is compensated by appropriate variance terms.
Condition \eqref{eq:k_symmetry} can be checked by evaluating the partial correlations $(\rho_{tr})_{t,r\in T}$ of the resulting vector $\bgamma_T$.
\end{ex}

\begin{ex}[Non-symmetric cases]\label{ex:non_symm}
In both cases previously considered (Model \ref{model:Sk} and Example \ref{ex:weakly_symm}) the auxiliary Markov chain $\textbf{X}$ defined in \eqref{eq:non_symm_markov} follows a natural random walk, in the sense that at each time the chain chooses the next state uniformly at random among children nodes.
However, condition \eqref{eq:k_symmetry} is also satisfied by non-symmetric cases where $\textbf{X}$ is not a natural random walk.
In particular any instance of Model \ref{model:NSk} such that 
\begin{align}\tag{S*}\label{eq:k_symmetry_appendix}
\sum_{r\in ch(t)}\rho_{tr}^2
\;=\;
c_{\ell(t)}
\qquad\hbox{for all }
t\in T\backslash T_L\,,
\end{align}
for some $(k-1)$-dimensional vector $(c_0,\dots,c_{k-2})$ induces a posterior distribution satisfying \eqref{eq:k_symmetry}.
In fact, in the context of Model \ref{model:NSk} conditions \eqref{eq:k_symmetry_appendix} and \eqref{eq:k_symmetry} are equivalent (this can be derived from \eqref{eq:tree_cond_indep} and \eqref{eq:non_symm_markov}).
\end{ex}

The cases previously considered are expressed in terms of centred pa\-ra\-me\-tri\-za\-tion, meaning that as all the  instances of Model \ref{model:NSk} they satisfy \eqref{eq:tree_cond_indep}.
Nevertheless Lemma \ref{lemma:S_closed} shows that any symmetric hierarchical reparametrization of a vector satisfying \eqref{eq:k_symmetry_rescaled} still satisfies \eqref{eq:k_symmetry_rescaled}.
This implies, for example, that the fully non-centred version of Model \ref{model:Sk} and any mixed strategy where some level is centred and some is not centred, still satisfies \eqref{eq:k_symmetry_rescaled} (after rescaling).

Moreover, note that the exact analysis we will now provide for the Gibbs sampler on models satisfying \eqref{eq:k_symmetry_rescaled} can be used to provide bound on general cases that do not satisfy \eqref{eq:k_symmetry_rescaled} (see for example Theorem \ref{thm:rates_3_uneven_CP}).

\subsection{Multigrid decomposition}\label{sec:multigrid_k}
We now show how to use the multigrid decomposition to analyze the Gibbs Sampler for random vectors $\bbeta_T$ satisfying \eqref{eq:hierarchical_structure} and \eqref{eq:k_symmetry}.
Our aim is to provide a transformation of $\bbeta_T$
 that factorizes the Gibbs Sampler Markov Chain into independent and more tractable sub-chains.
Similarly to Section \ref{sec:multigrid} in the following we will often denote $\bbeta_{T_d}=(\beta_t)_{t\in T_d}$ by $\bbeta^{(d)}$. 
We proceed in two steps, first introducing the averaging operators $\phi^{(p)}$ and then the residual operators $\delta^{(p)}$. 
For any $p\leq d$ the averaging operator $\phi^{(p)}:\R^{T_d}\rightarrow\R^{T_p}$ is defined as
\begin{align}\label{eq:phi_defi}
\phi^{(p)}_r\bbeta^{(d)}
=&
\E[\beta_{X_d}|\bbeta_T,X_p=r]&r\in T_p
\nonumber\\
=&
\sum_{t\in T_d}
\beta_{t}P(t|r)
&
\end{align}
where $\textbf{X}=(X_\ell)_{\ell=0}^{k-1}$
is the Markov chain defined by \eqref{eq:non_symm_markov}.
Loosely speaking $\phi^{(p)}\bbeta^{(d)}=\E[\beta_{X_d}|\bbeta_T,X_p]$
 can be interpreted as the averages of $\bbeta^{(d)}$ at the coarseness corresponding to the $p$-th level of the hierarchy.
In particular $\phi^{(d)}\bbeta^{(d)}=\bbeta^{(d)}$ and $\phi^{(0)}_{t_0}\bbeta^{(d)}=\E[\beta_{X_d}|\bbeta_T]$.
\begin{ex}[Averaging operators in the symmetric case]
Let $\bbeta_T=\bgamma_T$ be given by Model \ref{model:Sk}. Then 
\begin{align*}
\phi^{(p)}_r\bbeta^{(d)}
=&
\frac{1}{\prod_{\ell=p+1}^d I_\ell}
\left(\sum_{t\in T_d\,:\,t\succeq r}
\beta_{t}\right)
&r\in T_p\,.
\end{align*}
\end{ex}

Given $\phi$, we define the residual operators
$
\delta^{(p)}:\R^{T_d}
\rightarrow
 \R^{T_p}
$
as
$\delta^{(p)}=(\delta^{(p)}_r)_{r\in T_p}$
with $\delta^{(p)}_r:\R^{T_d}
\rightarrow\R$ defined as
\begin{align}\label{eq:delta_defi}
\delta^{(p)}_r\bbeta^{(d)}=&
\phi^{(p)}_r\bbeta^{(d)}-\phi^{(p-1)}_{pa(r)}\bbeta^{(d)}
&r\in T_p
\end{align}
for $1\leq p\leq d\leq k-1$
and
$
\delta^{(0)}\bbeta^{(d)}=\phi^{(0)}\bbeta^{(d)}
$
for $0=p\leq d\leq k-1$.
Analogously to the 3-level case of Section \ref{sec:multigrid}, under suitable assumptions the residual operators $\delta^{(p)}$ decompose the Markov chain generated by the Gibbs Sampler into independent sub-chains.
To prove the result we first need the following lemma.
\begin{lemma}[$p$-residuals interact only with $p$-residuals]\label{lemma:delta_expansion}
Let $\bbeta_T$ be a Gaussian random vector satisfying \eqref{eq:hierarchical_structure} and \eqref{eq:k_symmetry_rescaled}.
Then for any $p$ and $d$ with $0\leq p\leq d\leq k-1$, for all $r\in T_p$ we have
the identity
\begin{align*}
\E[\delta^{(p)}_r\bbeta^{(d)}|\bbeta\backslash\bbeta^{(d)}]
-\E[\delta^{(p)}_r\bbeta^{(d)}]
&=
\sum_{\ell\in\{p,\dots,k-1\}\backslash d}
c_{d\ell}
\left(\delta^{(p)}_r\bbeta^{(\ell)}-\E[\delta^{(p)}_r\bbeta^{(\ell)}]\right)
\,.
\end{align*}
\end{lemma}
Given Lemma \ref{lemma:delta_expansion} we can prove the following multigrid decomposition for hierarchical linear models.
\begin{theorem}[Multigrid decomposition for $k$ levels]\label{thm:factorization_general}
Let $(\bbeta(s))_{s\in\N}$ be a Markov chain evolving according to \ref{sampler:GSk} with $\bbeta_T$ satisfying \eqref{eq:hierarchical_structure} and \eqref{eq:k_symmetry_rescaled}.
Then $(\delta^{(0)}\bbeta(s))_s$, \dots, $(\delta^{(k-1)}\bbeta(s))_s$ are $k$ independent Markov chains. 
Moreover, each chain $\delta^{(p)}\bbeta(s)=(\delta^{(p)}\bbeta^{(d)}(s))_{d=p}^{k-1}$ evolves according to the following blocked Gibbs sampler scheme with $(k-p)$ components:
for $d$ going from $p$ to $k-1$ sample
\begin{align}\label{eq:blocked_gibbs}
\delta^{(p)}\bbeta^{(d)}(s+1)
\;\sim\;
\mathcal{L}\left(\delta^{(p)}\bbeta^{(d)}|
(\delta^{(p)}\bbeta^{(\ell)}(s+1))_{p\leq \ell<d},(\delta^{(p)}\bbeta^{(\ell)}(s))_{d<\ell\leq k-1}
\right)\,,
\end{align}
where $\mathcal{L}(X|Y)$ denotes the conditional distribution of $X$ given $Y$.
\end{theorem}
Theorem \ref{thm:factorization_general} implies that running a Gibbs sampler $(\bbeta(s))_s$ targeting distributions satisfying \eqref{eq:hierarchical_structure} is equivalent to running $k$ independent blocked Gibbs Samplers, one for each level of coarseness from $(\delta^{(0)}\bbeta(s))_s$ to $(\delta^{(k-1)}\bbeta(s))_s$.
\begin{coroll}\label{coroll:rate_k_levels}
Let $\bbeta_T$ satisfies \eqref{eq:hierarchical_structure} and \eqref{eq:k_symmetry_rescaled}.
Then the rate of convergence of \ref{sampler:GSk} is given by $\max\{\rho_0,\dots,\rho_{k-1}\}$ where for each $p\in\{0,\dots,{k-1}\}$, $\rho_p$ is the rate of convergence of $(\delta^{(p)}\bbeta(s))_s$.
\end{coroll}

\subsection{Hierarchical ordering of rates}\label{sec:ordering_k}
Combining the results in \citet[Sec.2.2]{RobertsSahu1997} with the multigrid decomposition, we can characterize the rates of convergence of the $k$ independent Markov chains described above as follows.
\begin{theorem}\label{thm:kronecker}
The rate of convergence of $(\delta^{(p)}\bbeta(s))_s$ is given by the largest modulus eigenvalue of $(\I_{k-p}-L)^{-1}U$. Here $\I_{k-p}$ is the $(k-p)$ dimensional identity matrix, while $L$ and $U$ are, respectively, the strictly lower and strictly upper triangular part of $(c_{d\ell})_{d,\ell=p}^{k-1}$, with $c_{d\ell}$ given by \eqref{eq:k_symmetry_rescaled}.
\end{theorem}
Theorem \ref{thm:kronecker} implies that the convergence properties of the $k$ independent Markov chains are closely related.
In particular, from the rates of convergence point of view, the $k$ Markov chains updating $\delta^{(p)}\bbeta$ for $p=0,\dots,k-1$ behave as Gibbs samplers targeting a decreasing number of dimensions (from $k$ down to 1) of a single $k$-dimensional Gaussian distribution with precision matrix given by $-C$, where $C=(c_{d\ell})_{d,\ell=p}^{k-1}$ is given by \eqref{eq:k_symmetry_rescaled}.
This suggests that the convergence properties of the sub-chains will typically improve from that of $(\delta^{(0)}\bbeta(s))_s$ to that $(\delta^{(k-1)}\bbeta(s))_s$ and that the rate of convergence of $(\delta^{(0)}\bbeta(s))_s$ will typically determine the rate of the whole sampler \ref{sampler:GSk}.
In particular, in the centred parametrization case we can use the well-known Cauchy interlacing theorem (see e.g. \citealp{bhatia2013}) to
show that the rate of convergence is monotonically non-increasing from $(\delta^{(0)}\bbeta(s))_s$ to $(\delta^{(k-1)}\bbeta(s))_s$.
\begin{theorem}\label{thm:ordering_CP}\emph{(Ordering of rates for centred parametrization)}
Let $\bgamma$ be a Gaussian vector satisfying \eqref{eq:tree_cond_indep} and \eqref{eq:k_symmetry_rescaled} and 
let $(\bgamma(s))_{s\in\N}$ be the corresponding Markov chain evolving according to \ref{sampler:GSkCP}.
Then the convergence rates of the $k$ independent Markov chains $(\delta^{(0)}\bgamma(s))_s$, \dots, $(\delta^{(k-1)}\bgamma(s))_s$ satisfy
\begin{equation}\label{eq:ordering_k}
\rho(\delta^{(0)}\bgamma(s))\geq
\rho(\delta^{(1)}\bgamma(s))\geq
\cdots\geq
\rho(\delta^{(k-1)}\bgamma(s))=0\,.
\end{equation}
\end{theorem}
In Theorem \ref{thm:ordering_CP} we needed the additional assumption \eqref{eq:tree_cond_indep} to prove \eqref{eq:ordering_k}.
The reason is that, while in most cases the convergence rates of a deterministic-scan Gibbs Sampler targeting a $n$-th dimensional Gaussian distribution improves if one of the coordinates is conditioned to a fixed value and the sampler targets only the remaining $(n-1)$ coordinates, this is not true in general.
Example 2 of \cite{RobertsSahu1997} provides a counter-example (see also \citealp[page 319]{Whittaker1990}).
In  \cite{RobertsSahu1997}, this example was used a counter-example regarding blocking strategies, it also works in the present context.
We note that, if one were to consider a random scan version of the Gibbs Sampler, the reversibility of the induced Markov chains would allow to prove the ordering result in Theorem \ref{thm:ordering_CP} with no need to assume \eqref{eq:tree_cond_indep}. 
We leave this as a direction of future research and briefly mention it in Section \ref{sec:conclusion}.


Theorem \ref{thm:ordering_CP} implies the following corollary.
\begin{coroll}\label{coroll:rate_k_levels_CP}
Let $\bgamma$ be a Gaussian vector satisfying \eqref{eq:tree_cond_indep} and \eqref{eq:k_symmetry_rescaled}.
Then the rate of convergence of \ref{sampler:GSkCP}
is given by the largest squared eigenvalue of the $k$-dimensional matrix $C-\I_{k}$, where $C=(c_{d\ell})_{d,\ell=0}^{k-1}$ is defined by \eqref{eq:k_symmetry_rescaled} and $\I_{k}$ is the $k$-dimensional matrix.
\end{coroll}
In particular, considering the special case of Model \ref{model:Sk} it is easy to deduce the following result.
\begin{coroll}\label{coroll:rate_k_levels_CP_symm}
The rate of convergence of $GS(\bgamma_T)$ targeting Model \ref{model:Sk} is given by the largest squared eigenvalue of the $k$-dimensional matrix
\begin{equation*}
\begin{pmatrix}
0 & r_{1}\\
(1-r_{2}) & 0 & r_{2}\\
 & \dots & \dots & \dots\\
 &  & (1-r_{k-2}) & 0 & r_{k-2}\\
 &  &  & (1-r_{k-1}) & 0
\end{pmatrix}
\end{equation*}
where 
$
r_\ell=
\frac{I_\ell\tau_{\ell}}{\tau_{\ell-1}+I_\ell\tau_{\ell}}
$
with $(\tau_1,\dots,\tau_{k-1})$ and $(I_1,\dots,I_{k-1})$ given by Model \ref{model:Sk}, $\tau_0=0$, $\tau_{k}=\tau_e$ and $I_{k}=J$.
\end{coroll}

\subsection{Example: rates of convergence for 4-level models}\label{sec:4_levels}
The results developed in Sections \ref{sec:multigrid_k} and \ref{sec:ordering_k} allow to analyze hierarchical models with an arbitrary number of levels.
For example we could consider 4-level extensions of Model \ref{model:S3}.
\begin{taggedmodel}{S4}\label{model:S4}(Symmetric 4-levels hierarchical model) 
Suppose
\begin{equation}\label{eq:NCP_4}
y_{ijk\ell}=\mu+a_i+b_{ij}+c_{ijk}+\epsilon_{ijk\ell},
\end{equation}
where $i$, $j$, $k$ and $\ell$ run from 1 to $I$, $J$, $K$ and $L$ respectively and $\epsilon_{ijk\ell}$ are iid normal random variables with mean 0 and variance $\var{e}$.
We employ a standard Bayesian model specification assuming 
$a_i\sim N(0,\var{a})$, $b_{ij}\sim N(0,\var{b})$,  $c_{ijk}\sim N(0,\var{c})$ and a flat prior on $\mu$.
\end{taggedmodel}
In order to fit Model \ref{model:S4} with a Gibbs Sampler like \ref{sampler:GSk}, one could consider centering or non-centering each of the three levels $(a_i)_i$, $(b_{ij})_{ij}$ and $(c_{ijk})_{ijk}$.
Let $(\lambda_1,\lambda_2,\lambda_3)\in\{0,1\}^3$ be the non-centering indicators associated to the resulting in $8=2^3$ combinations. Here $\lambda_d=1$ indicates that the $d$-th level is non-centred while $\lambda_d=0$ indicates that it is centred.
The corresponding rates of convergence $\rho_{(\lambda_1,\lambda_2,\lambda_3)}$ can then be expressed in terms of the following normalized variance ratios
\begin{align*}
\ratio{i}{j}
&=
\frac{\tilde{\sigma}_i^2}{\tilde{\sigma}_i^2+\tilde{\sigma}_j^2}
&i,j\in\{1,2,3,4\}
\,,
\end{align*}
where $\tvar{1}=\frac{\var{a}}{I}$, $\tvar{2}=\frac{\var{b}}{IJ}$, $\tvar{3}=\frac{\var{c}}{IJK}$ and $\tvar{4}=\frac{\var{e}}{IJKL}$.
If $\lambda_1=1$
 (i.e. using the non-centred parametrization $(a_i)_i$ at level $1$) the rates are 
\begin{align*}
\rho_{111}
&=
\max\{\ratio{1}{4},\ratio{2}{4},\ratio{3}{4}\}
\\
\rho_{110}
&=
\max\{\ratio{1}{3},\ratio{2}{3},\ratio{4}{3}\}
\\
\rho_{100}
&=
\max\{\ratio{1}{2},1-\ratio{2}{3}\ratio{3}{4}\}
\\
\rho_{101}
&=
\max\{\ratio{1}{2},1-\ratio{2}{4}\ratio{4}{3}\}
\end{align*}
When $\lambda_1=0$ the expressions for the convergence rates are still explicit, but slightly more involved and are reported in Section 3.1 of the supplementary material.
These rates can be derived from Corollary \ref{coroll:rate_k_levels} and Theorem \ref{thm:kronecker}.
It is worth noting that also in this 4-level case the skeleton chain $\delta^{(0)}\bbeta$ is always the slowest chain for all centred and non-centred parametrizations (which can be checked by computing the rates of convergence of $\delta^{(1)}\bbeta$, $\delta^{(2)}\bbeta$ and $\delta^{(3)}\bbeta$ using Theorem \ref{thm:kronecker} and comparing those to the ones of $\delta^{(0)}\bbeta$), even if for the general $k$-level case we were able to prove this fact only for the fully-centred parametrization (Theorem \ref{thm:ordering_CP}).
The expressions given here can be easily used to derive conditionally optimal parametrizations for Model \ref{model:S4} given the rescaled variance components $(\tvar{i})_{i=1}^4$. For example, choosing whether to center or not each level by comparing the level-specific rescaled variances with the sum of the rescaled variances of the lower levels like in Section \ref{sec:cond_optimal} leads to rates of convergence upper bounded by $\frac{3}{4}$.

\section{Conclusions and future work}\label{sec:conclusion}
In this work we studied the convergence properties of the Gibbs Sampler algorithm in the context of Gaussian multilevel models.
To do so we developed a novel analytic approach based on multigrid decompositions that allows to factorize the Markov chain of interest into independent and easier to analyze sub-chains.
This decomposition enables us to evaluate explicitly the $L^2$-rate of convergence in various models of interest.
The results offer a detailed and valuable insight into the interaction between multilevel structures (e.g.\ nested and crossed) and the Gibbs Sampler and provide guidance on the choice of the computationally optimal parametrizations or linear constraints, which can potentially be relevant also beyond the Gaussian case (see e.g.\ Section \ref{sec:gamma_poisson}), and indication of which parameters to monitor in the convergence diagnostic process (see Theorem \ref{thm:rates_ordering_3} and discussion at the end of Section \ref{sec:motivating_example}).
Since the first preprint version of this paper, the multigrid decomposition developed in this paper has already found other practical applications.
In particular \citet{papaspiliopoulos2018scalable} have successfully exploited it to analyze the computational complexity of the Gibbs Sampler in the context of crossed random effect models (see also \citealp{GaoOwen2017}) and to design an algorithmic modification with linear computational complexity.

Together with explicit formulas for $L^2$-rates of convergences, the multigrid decomposition we developed in this paper provides a simple and intuitive theoretical characterizations of practical behaviors commonly observed in practice when fitting hierarchical models with MCMC, such as slower mixing for hyper-parameters at higher levels (see Theorems \ref{thm:rates_ordering_3} and \ref{thm:ordering_CP}), algorithmic scalability with width of the hierarchy but not with height (e.g. Theorem \ref{thm:rates_3} and Corollary \ref{coroll:rate_k_levels_CP_symm}) and good performances of centred parametrization in data-rich contexts (Theorem \ref{thm:rates_3_uneven_CP}).
We hope that the results presented in this work will provide a first step towards providing quantitative understanding of the behavior of MCMC algorithms (even beyond the Gibbs Sampler) in the extremely popular context of Bayesian hierarchical and multilevel models.

The present work could be extended in many directions. 
For example, it would be interesting to extend the results for non-symmetric cases, either by generalizing the bounds of Theorem \ref{thm:rates_3_uneven_CP} or by weakening the symmetry assumption in \eqref{eq:k_symmetry}.
In terms of classes of models considered, a natural and important extension would be to consider the multivariate case (where each parameter $\gamma_t$ is a multivariate random vector) and the regression case. 
We expect many results developed in this work to extend to the multivariate and regression case, even if in that context the role played by non-symmetric cases will be more crucial. 
Another important class of models that would be worth approaching with methodologies analogous to the ones developed here are models based on Gaussian processes commonly used, for example, in spatial statistics (see e.g.\ \citealp{BassSahu2016}).

An important and ambitious aim 
would be to extend the results to other tractable distributions within the exponential family beyond the Gaussian case.
A starting point for this could be to analyze Model \ref{model:poisson_crossed} as mentioned in Remark \ref{rmk:gamma_poisson}.
Also, many non-Gaussian hierarchical models can be well-approximated by Gaussian ones for sufficiently large data sets, so that it is reasonable to conjecture that the qualitative conclusions (at least) of our study might remain valid when extrapolated to non-Gaussian models, rather like the analysis given in \cite{sahu1999convergence}. A detailed study of this question is left for future work.

We have concentrated in this paper on deterministic samplers. However, explicit rates of convergence of random scan samplers are also available in the Gaussian case as described in  \cite{.amit:1996:properties+convergence+perturbations} and extended in \cite{RobertsSahu1997}.   
deterministic and random scan samplers can sometimes differ substantially in their convergence properties, see for example \cite{roberts2015surprising}, although no general theory for this phenomenon is well-understood, so that the insights of this work could be particularly useful in this direction.
Also, in the random scan case the reversibility of the induced Markov chains would allow us to apply the Cauchy interlacing theorem under weaker assumptions than Theorem \ref{thm:ordering_CP} and thus prove orderings results for general hierarchical parametrizations $\bbeta_T=(\beta_t)_{t\in T}$. 

While this work is focused on $L^2$-rates of convergence, the same approach could be used to derive bounds on the distance (e.g.\ total variation or Wasserstein) between the distribution of the Markov chain at a given iteration and the target distribution (see e.g.\ \citet{.amit:1996:properties+convergence+perturbations}, \citet[eq.(15)]{RobertsSahu1997} and \citealp[Sec.4.4]{khare2009rates}).
Such a formulation would be interesting to extend the recent growth in literature on providing rigorous characterizations of the computational complexity of Bayesian hierarchical linear models, see for example  \cite{RajaratnamSparks2015,roberts2016complexity,johndrow2015approximations}.
In order to provide full characterizations, however, the case of unknown variances should be considered (see e.g.\ \cite{jones2004sufficient} for the two level case).


\subsection*{Acknowledgments}
The authors are grateful for stimulating discussions with Omiros Papaspiliopoulos
and Art Owen.
GZ supported in part by an EPSRC Doctoral Prize fellowship, by the European Research Council (ERC) through StG ``N-BNP" 306406 and by MIUR through the PRIN Project 2015SNS29B.
GOR acknowledges support from EPSRC through grants EP/K014463/1 (i-Like) and EP/D002060/1 (CRiSM).

%

\bibliographystyle{plainnat}
\bibliography{Multigrid_bibliography}

\clearpage
\begin{center}
\textbf{\huge Supplementary material for ``Multilevel linear models, Gibbs samplers and multigrid decompositions''}\\
\end{center}

\numberwithin{equation}{section}
\numberwithin{theorem}{section}

\setcounter{equation}{0}
\setcounter{figure}{0}
\setcounter{table}{0}
\setcounter{page}{1}
\setcounter{section}{0}
\makeatletter


\section{Remarks on the results' proofs}
We list proofs following the mathematical chronology. 
This is different from the order of appearance in the paper because here we start from the results for $k$-level models, namely the results from Lemma \ref{lemma:hierarchical_reparametrization_invertible} to 
Corollary \ref{coroll:rate_k_levels_CP_symm},
and then move to the ones for $3$-level models and crossed models, namely the results from Theorem \ref{thm:factorization_3} to Theorem \ref{thm:rates_3_uneven_CP}.

The results developed here rely on classical theory for Gaussian Gibbs Samplers and their autoregression representations.
The following theorem summarizes the most relevant facts for our purposes.
See Lemma 1 and Theorem 1 of \citet{RobertsSahu1997} for proofs and more detailed statements. 
\begin{theorem}\label{thm:gaussian_AR}
The Markov chain $(\bbeta(s))_{s=1,2,\dots}$ generated by a deterministic-scan Gibbs Sampler targeting a $n$-dimensional Gaussian distribution with covariance matrix $\Sigma$ evolves as a multivariate Gaussian autoregressive process with transition kernel given by
$$
\mathcal{L}(\bbeta(s+1)|\bbeta(s))= N(B\bbeta(s)+b,\Sigma-B\Sigma B^T)\,,
$$
for some fixed $n$-dimensional square matrix $B$ and $n$-dimensional vector $b$.
The rate of convergence of $(\bbeta(s))_{s=1,2,\dots}$ equals the largest modulus eigenvalue of $B$.
\end{theorem}
The autoregressive matrix $B$, which we will sometimes refer to as $B$-matrix of the Gibbs Sampler under consideration, can be explicitly derived from the target covariance $\Sigma$, following the recipe in Section 2.2 of \citet{RobertsSahu1997}.
The latter involves computing an $n$-dimensional square matrix $A$ and then exploiting the identity $B=(\I_n-L)^{-1}U$, where $\I_n$ is the $n$-dimensional identity matrix, $L$ is the lower triangular part of $A$ and $U=A-L$.
The specific form of $A$, which we will sometimes refer to as $A$-matrix, depends on the updating strategy of the Gibbs Sampler under consideration.

In principle, Theorem \ref{thm:gaussian_AR} provides a constructive way to compute the rates of convergence of Gaussian Gibbs Samplers.
In realistic high-dimensional statistical scenarios, however, it is typically very challenging to compute analytically the largest modulus eigenvalue of $B$, as the latter is a matrix with dimensionality equal to the number of parameters in the problem. 
In the following proofs, we will 
 exploit the probabilistic structure of hierarchical models, and the multigrid decomposition in particular, to reduce the problem to studying low-dimensional Gaussian autoregressions, namely the skeleton chains $\delta^{(0)}(\bbeta(s))_{s=1,2,\dots}$, with more tractable $B$-matrices.
\section{Proofs of Lemmas}

\subsection*{Proofs of Lemmas \ref{lemma:hierarchical_reparametrization_invertible} and  \ref{lemma:H_closed}}
In order to prove Lemmas \ref{lemma:hierarchical_reparametrization_invertible} and  \ref{lemma:H_closed} we need some preliminary results on matrices $M=(M_{tr})_{t,r\in T}$ indexed by elements of a tree $T$.
\begin{lemma}[Triangular matrices on trees]\label{lemma:triangular}
Suppose that a matrix $L=(L_{tr})_{t,r\in T}$ satisfies the following lower-triangularity condition
\begin{align}
\tag{L$_m$}\label{eq:L}
L_{tr}=&0
\hbox{ unless }
t\succeq r
\hbox{ and }
L_{tt}\neq 0 \quad\forall t\in T
\,.
\end{align}
Then $L$ is invertible and its inverse satisfies \eqref{eq:L}.
Similarly, if $U=(U_{tr})_{t,r\in T}$ satisfies the following upper-triangularity condition
\begin{align}
\tag{U$_m$}\label{eq:U}
U_{tr}=&0
\hbox{ unless }
t\preceq r
\hbox{ and }
U_{tt}\neq 0 \quad\forall t\in T
\,,
\end{align}
then $U$ is invertible and its inverse still satisfies \eqref{eq:U}.
\end{lemma}
\begin{proof}
Suppose that $L$ satisfies \eqref{eq:L}.
Also without loss of generality suppose $L_{tt}=1$ for all $t\in T$ by rescaling. 
Then write $L=(\mathbb{I}_T+N)$ where $\mathbb{I}_T$ is the $|T|\times |T|$ identity matrix and $N$ satisfies the following strict lower-triangularity condition
\begin{align}
\tag{L$^*_m$}\label{eq:LL}
N_{tr}=&0
\hbox{ unless }
t\succ r
\,.
\end{align}
Consider 
$N^2_{tr}=\sum_{s\in T}N_{ts}N_{sr}=\sum_{s\in T\,:\,r\prec s\prec t}N_{ts}N_{sr}$.
From the last expression it follows that $N^2_{tr}\neq 0$ implies
$r\prec t$ and $|\ell(r)-\ell(t)|\geq 2$, where $\ell(t)$ denote the level of $t$ in the tree $T$ .
Iterating the same argument we have that $N^p_{tr}\neq 0$ implies $r\prec t$ and $|\ell(r)-\ell(t)|\geq p$.
It follows that $N^p$ satisfies \eqref{eq:LL} for all $p\geq1$ and that $N^p=\boldsymbol{0}_T$ for all  $p\geq k$ where $\boldsymbol{0}_T$ is the $|T|\times |T|$ zero matrix.
Here $k$ indicates the number of levels of $T$, as in Section \ref{sec:k_levels} of the paper.
From $N^p=\boldsymbol{0}_T$ for $p\geq k$ it follows that 
$$
(\mathbb{I}_T+N)(\mathbb{I}_T+\sum_{p=1}^{k-1} (-1)^pN^p)
=
\mathbb{I}_T+(-1)^{k-1}N^{k}
=
\mathbb{I}_T
$$ 
and therefore $L^{-1}=(\mathbb{I}_T+N)^{-1}=\mathbb{I}_T+\sum_{p=1}^{k-1} (-1)^pN^p$.
Since $N^p$ satisfies \eqref{eq:LL} for all $p\geq1$ it follows that 
$L^{-1}=\mathbb{I}_T+\sum_{p=1}^{k-1} (-1)^pN^p$ satisfies \eqref{eq:L}.

The analogous statement for \eqref{eq:U} can be deduced by observing that $U$ satisfies 
\eqref{eq:U} if and only if its transpose satisfies \eqref{eq:L}.
%
\end{proof}

Lemma \ref{lemma:triangular} can be used to deduce Lemma \ref{lemma:hierarchical_reparametrization_invertible} in a straightforward way.
\begin{proof}[Proof of Lemma \ref{lemma:hierarchical_reparametrization_invertible}]
Suppose that $\bbeta_T=\Lambda\bgamma_T$ is a hierarchical reparametrization of $\bgamma_T$.
This is equivalent to say that $\Lambda=(\Lambda_{tr})_{t,r,\in T}$ is a matrix satisfying \eqref{eq:L} and that for all $t\in T$, $\bbeta_t=\sum_{s\in T}\Lambda_{ts}\gamma_s$. 
Lemma \ref{lemma:triangular} implies that $\Lambda$ is invertible and that $Z=\Lambda^{-1}$ satisfies \eqref{eq:L}.
Therefore $\bgamma_T=Z\bbeta_T$ is a hierarchical reparametrization of $\bbeta_T$.
\end{proof}

To prove Lemma \ref{lemma:H_closed} we need an additional preliminary result.

\begin{lemma}[Closure of \ref{eq:H}]\label{lemma:closure_Hm}
Suppose that $M=(M_{tr})_{t,r\in T}$ satisfies the following condition
\begin{align}\tag{H$_m$}\label{eq:H}
M_{tr}=&0
\hbox{ unless }
t\preceq r
\hbox{ or }
r\preceq t\,.
\end{align}
Then if we multiply $M$ from the right with a matrix $L$ satisfying \eqref{eq:L}, the product $ML$ still satisfy \eqref{eq:H}.
Similarly if we multiply $M$ from the left  with a matrix $U$ satisfying \eqref{eq:U}, the product $UM$ still satisfy \eqref{eq:H}.
\end{lemma}
\begin{proof}
Consider
$(UM)_{tr}=\sum_{s\in T}U_{ts}M_{sr}$. From \eqref{eq:H} and \eqref{eq:U}, the elements $U_{ts}M_{sr}$ in the latter sum are non-zero only when $s$ belongs to the intersection of $\{s\in T\,:\,s\succeq t\}$ and $\{s\in T\,:\,s\succeq r\hbox{ or }s\preceq r\}$.
It is easy to see that the intersection of these two sets is non-empty only if $t\succeq r$ or $t\preceq r$.
Therefore $UM$ satisfies \eqref{eq:H}.
The argument to show that $ML$ satisfies \eqref{eq:H} is analogous.
\end{proof}
We now combine Lemmas \ref{lemma:triangular} and \ref{lemma:closure_Hm} to prove Lemma \ref{lemma:H_closed}.
\begin{proof}[Proof of Lemma \ref{lemma:H_closed}]
Suppose that $\bbeta_T$ satisfies \eqref{eq:hierarchical_structure}.
This is equivalent to saying that its precision matrix $Q^{(\bbeta)}=(Q^{(\bbeta)}_{tr})_{t,r\in T}$ satisfies \eqref{eq:H}.
Consider a hierarchical reparametrization of $\bbeta_T$ denoted by $\Lambda\bbeta_T$.
Then its precision matrix is given by 
$Q^{(\Lambda\bbeta)}=(\Lambda^T)^{-1}Q^{(\bbeta)}\Lambda^{-1}$ where $\Lambda^T$ denote the transpose of  $\Lambda$.
By definition of hierarchical reparametrizations, $\Lambda$ satisfies \eqref{eq:L}.
Therefore Lemma \ref{lemma:triangular} implies that $\Lambda^{-1}$ satisfies \eqref{eq:L} and, consequently, Lemma \ref{lemma:closure_Hm} implies that $Q^{(\bbeta)}\Lambda^{-1}$ satisfies \eqref{eq:H}.
Since $\Lambda$ satisfies \eqref{eq:L}, then $\Lambda^T$ satisfies \eqref{eq:U} and thus Lemma \ref{lemma:triangular} implies that $(\Lambda^T)^{-1}$ satisfies \eqref{eq:U}.
We can then apply Lemma \ref{lemma:closure_Hm} to $(\Lambda^T)^{-1}$ and $Q^{(\bbeta)}\Lambda^{-1}$ to deduce that $(\Lambda^T)^{-1}Q^{(\bbeta)}\Lambda^{-1}$ satisfies \eqref{eq:H} and thus $\Lambda\bbeta_T$ satisfies \eqref{eq:hierarchical_structure}.
\end{proof}
Corollary \ref{coroll:PNCP} follows easily from equation \eqref{eq:tree_cond_indep} and Lemma \ref{lemma:H_closed}.

\subsection*{Proof of Lemma \ref{lemma:S_closed}}
The strategy to prove Lemma \ref{lemma:S_closed} is similar to the one used to prove Lemmas \ref{lemma:hierarchical_reparametrization_invertible} and  \ref{lemma:H_closed} above, with the difference that we have to check that the symmetry condition is preserved under the operations considered. To do so we first prove two auxiliary lemmas.

\begin{lemma}[Symmetric triangular matrices on trees]\label{lemma:symmetric_triangular}
Suppose that a matrix $L=(L_{tr})_{t,r\in T}$ satisfies the following symmetric lower-triangularity condition
\begin{align}
\tag{SL$_m$}\label{eq:SL}
L_{tr}=&l_{\ell(t)\ell(r)}\1(t\succeq r)
\hbox{ and }
l_{dd}\neq 0
\hbox{ for all }
 d\in\{0,\dots,k-1\}
\,,
\end{align}
where $(l_{pd})_{p,d=0}^{k-1}$ is a $k\times k$ real valued matrix.
Then $L$ is invertible and its inverse satisfies \eqref{eq:SL}.
Similarly, if $U=(U_{tr})_{t,r\in T}$ satisfies the following symmetric upper-triangularity condition
\begin{align}
\tag{SU$_m$}\label{eq:SU}
U_{tr}=&u_{\ell(t)\ell(r)}\1(t\preceq r)
\hbox{ and }
u_{dd}\neq 0
\hbox{ for all }
 d\in\{0,\dots,k-1\}
\,,
\end{align}
with $(u_{pd})_{p,d=0}^{k-1}$ being a $k\times k$ real valued matrix,
then $U$ is invertible and its inverse still satisfies \eqref{eq:SU}.
\end{lemma}
\begin{proof}
Suppose that $L$ satisfies \eqref{eq:SL}.
Also without loss of generality suppose $L_{tt}=1$ for all $t\in T$ by rescaling.
Since $L$ also satisfies \eqref{eq:L}, arguing as in the proof of Lemma \ref{lemma:triangular} we can write 
$L^{-1}=\mathbb{I}_T+\sum_{p=1}^{k-1} (-1)^pN^p$ 
where
$N=L-\mathbb{I}_T$ 
and $N$ satisfies the following symmetric strict lower-triangularity condition
\begin{align}
\tag{SL$^*_m$}\label{eq:SSL}
N_{tr}=&l_{\ell(t)\ell(r)}\1(t\succ r)
\hbox{ and }
l_{dd}\neq 0
\hbox{ for all }
 d\in\{0,\dots,k-1\}
\,.
\end{align}
From 
$N^2_{tr}=
\sum_{s\in T}N_{ts}N_{sr}=
\sum_{s\in T\,:\,r\prec s\prec t}N_{ts}N_{sr}=
\sum_{\ell'\,:\,\ell(r)\prec \ell' \prec \ell(t)}l_{\ell(t)\ell'}l_{\ell'\ell(r)}$
we deduce that $N^2$ still satisfies \eqref{eq:SSL} for some different $(l_{pd})_{p,d=0}^{k-1}$.
Iterating the same argument we have that $N^p$ satisfies \eqref{eq:SSL} for all $p\geq 1$.
Since $L^{-1}=\mathbb{I}_T+\sum_{p=1}^{k-1} (-1)^pN^p$ it follows that $L^{-1}$ satisfies \eqref{eq:SL}.

The analogous statement for \eqref{eq:SU} can be deduced by observing that $U$ satisfies \eqref{eq:SU} if and only if its transpose satisfies \eqref{eq:SL}.
\end{proof}
\begin{lemma}\label{lemma:closure_Sm}
Let $M=(M_{tr})_{t,r\in T}$ and $L=(L_{tr})_{t,r\in T}$ satisfy respectively \eqref{eq:k_symmetry_rescaled} and \eqref{eq:SL}.
Then the product $ML$ satisfies \eqref{eq:k_symmetry_rescaled} for some different matrix $(c_{pd})_{p,d=1}^{k-1}$.
Similarly, if $U=(U_{tr})_{t,r\in T}$ satisfies \eqref{eq:SU} then the product $UM$ satisfy \eqref{eq:k_symmetry_rescaled}.
\end{lemma}
\begin{proof}
First consider the product $ML$.
Lemma \ref{lemma:closure_Hm} implies that $ML$ satisfies \eqref{eq:H} and therefore $(ML)_{tr}=0$ unless $t\preceq r$ or $t\succeq r$.
Given $t\preceq r$ or $t\succeq r$ and using \eqref{eq:k_symmetry_rescaled} and \eqref{eq:SL} we have
\begin{align*}
&(ML)_{tr}
=
\sum_{s\in T}M_{ts}L_{sr}
=
\sum_{s\in T\,:\,s\succeq r}M_{ts}L_{sr}
\\&=
\sum_{s\in T\,:\,s\succeq (t\vee r) }M_{ts}L_{sr}
+
\sum_{s\in T\,:\,r\preceq s\prec t}M_{ts}L_{sr}
\\&=
\sum_{s\in T\,:\,s\succeq (t\vee r)}
m_{\ell(t)\ell(s)}l_{\ell(s)\ell(r)}P(s)
+
\sum_{s\in T\,:\,r\preceq s\prec t}
m_{\ell(t)\ell(s)}l_{\ell(s)\ell(r)}P(t)
\\&=
P(t\vee r)
\sum_{\ell'=\ell(t\vee r)}^{k-1}m_{\ell(t)\ell'}l_{\ell'\ell(r)}
+
P(t)
\sum_{\ell'=\ell(r)}^{\ell(t)-1}m_{\ell(t)\ell'}l_{\ell'\ell(r)}
\\&=
P(t\vee r)
\sum_{\ell'=\ell(r)}^{k-1}m_{\ell(t)\ell'}l_{\ell'\ell(r)}\,,
\end{align*}
where the sum from $\ell(r)$ to $\ell(t)$ equals 0 if $\ell(r)\geq \ell(t)$.
The latter equation implies that $ML$ satisfies \eqref{eq:k_symmetry_rescaled}.

To prove that the product $UM$ satisfies \eqref{eq:k_symmetry_rescaled} note that $M$ and $U$ satisfying \eqref{eq:k_symmetry_rescaled} and \eqref{eq:SU} respectively is equivalent to $M^T$ and $U^T$ satisfying \eqref{eq:k_symmetry_rescaled} and \eqref{eq:SL} respectively.
Therefore, by the first part of this Lemma, $(UM)^T=M^TU^T$ satisfies \eqref{eq:SL} and thus $((UM)^T)^T=UM$ satisfies \eqref{eq:k_symmetry_rescaled}.
\end{proof}
We now combine Lemmas \ref{lemma:symmetric_triangular} and \ref{lemma:closure_Sm} to prove Lemma \ref{lemma:S_closed}.
\begin{proof}[Proof of Lemma \ref{lemma:S_closed}]
Suppose that the precision matrix $Q^{(\tbbeta)}=(Q^{(\tbbeta)}_{tr})_{t,r\in T}$ of $\tbbeta_T$ satisfies \eqref{eq:k_symmetry_rescaled}.
Consider a symmetric hierarchical reparametrization of $\tbbeta_T$ denoted by $\Lambda\tbbeta_T$.
Then its precision matrix is given by 
$Q^{(\Lambda\tbbeta)}=(\Lambda^T)^{-1}Q^{(\tbbeta)}\Lambda^{-1}$ where $\Lambda^T$ denote the transpose of  $\Lambda$.
By definition of symmetric hierarchical reparametrizations, $\Lambda$ satisfies \eqref{eq:SL}.
Therefore Lemma \ref{lemma:symmetric_triangular} implies that $\Lambda^{-1}$ satisfies \eqref{eq:SL} and, consequently, Lemma \ref{lemma:closure_Sm} implies that $Q^{(\tbbeta)}\Lambda^{-1}$ satisfies \eqref{eq:k_symmetry_rescaled}.
Since $\Lambda$ satisfies \eqref{eq:SL}, then $\Lambda^T$ satisfies \eqref{eq:SU} and thus Lemma \ref{lemma:triangular} implies that $(\Lambda^T)^{-1}$ satisfies \eqref{eq:SU}.
We can then apply Lemma \ref{lemma:closure_Sm} to $(\Lambda^T)^{-1}$ and $Q^{(\tbbeta)}\Lambda^{-1}$ to deduce that $Q^{(\Lambda\tbbeta)}=(\Lambda^T)^{-1}Q^{(\tbbeta)}\Lambda^{-1}$ satisfies \eqref{eq:k_symmetry_rescaled}.
\end{proof}

\subsection*{Proof of Lemma \ref{lemma:delta_expansion}}\label{proof:delta_expansion}
\begin{proof}[Proof of Lemma \ref{lemma:delta_expansion}]
Suppose $\bbeta_T$ has zero mean (otherwise replace $\bbeta_T$ by $\bbeta_T-\mathbb{E}[\bbeta]$).
As in Section \ref{sec:multigrid_k} of the paper, given any $r$ and $t$ in $T$ we denote $P(X_{\ell(t)}=t|X_{\ell(r)}=r)$ by $P(t|r)$.
Using \eqref{eq:hierarchical_structure}, for any $d\in\{0,\dots,k-1\}$ and $t\in T_d$, we can write the full conditional expectation of $\beta_t$ as
\begin{align*}
\E[\beta_t|\bbeta_{T\backslash t}]
&=
\sum_{r\prec t}
A_{tr}
\beta_r
+
\sum_{r\succ t}
A_{tr}
\beta_r\,,
\end{align*}
where $A_{tr}=-\frac{Q_{tr}}{Q_{tt}}$ for any $r\neq t$.
Note that \eqref{eq:k_symmetry_rescaled} implies
\begin{equation}\label{eq:symm_A}
A_{tr}=
\frac{c_{\ell(t)\ell(r)}P(r\cap t)}{P(t)}
=
c_{\ell(t)\ell(r)}P(r|t)\,.
\end{equation}
It follows 
\begin{align*}
\E[\beta_t|\bbeta_{T\backslash t}]
&=
\sum_{r\prec t}
c_{d\ell(r)}P(r|t)
\beta_r
+
\sum_{r\succ t}
c_{d\ell(r)}P(r|t)
\beta_r
\\
&=
\sum_{\ell\neq d}
c_{d\ell}
\E[\beta_{X_\ell}|\bbeta_T,X_d=t]\,.
\end{align*}
Since the last equation does not depend on $\bbeta^{(d)}$ we have
$
\E[\beta_t|\bbeta_{T\backslash t}]=
\E[\beta_t|\bbeta\backslash\bbeta^{(d)}].$
For any $r\in T_p$, by definition of $\phi^{(p)}_r\bbeta^{(d)}$, we have
\begin{align*}
\E[\phi^{(p)}_r\bbeta^{(d)}|\bbeta\backslash\bbeta^{(d)}]
&=
\sum_{t\in T_d}
\Pr(t|r)
\E[\beta_{t}|\bbeta\backslash\bbeta^{(d)}]
\\
&=
\sum_{t\in T_d}
\Pr(t|r)
\sum_{\ell\neq d}
c_{d\ell}
\E[\beta_{X_\ell}|\bbeta_T,X_d=t]
\\
&=
\sum_{\ell\neq d}
c_{d\ell}
\sum_{t\in T_d}
\Pr(t|r)
\E[\beta_{X_\ell}|\bbeta_T,X_d=t]
\\
&=
\sum_{\ell\neq d}
c_{d\ell}
\E[\beta_{X_\ell}|\bbeta_T,X_p=r]
\,.
\end{align*}
From the latter equation and the definition of $\delta^{(p)}_r\bbeta^{(d)}$ it follows
\begin{align*}
\E[\delta^{(p)}_r\bbeta^{(d)}|\bbeta\backslash\bbeta^{(d)}]
&=
\E[\phi^{(p)}_r\bbeta^{(d)}|\bbeta\backslash\bbeta^{(d)}]-\E[\phi^{(p-1)}_{pa(r)}\bbeta^{(d)}|\bbeta\backslash\bbeta^{(d)}]
\\
&=
\sum_{\ell\neq d}
c_{d\ell}
\left(
\E[\beta_{X_\ell}|\bbeta_T,X_p=r]
-
\E[\beta_{X_\ell}|\bbeta_T,X_{p-1}=pa(r)]
\right)
\\
&=
\sum_{\ell\in\{p,\dots,k-1\}\backslash d}
c_{d\ell}
\delta^{(p)}_r\bbeta^{(\ell)}
\,.
\end{align*}
\end{proof}

\section{Proofs for hierarchical models with an arbitrary number of levels}

\subsection*{Proof or Theorem \ref{thm:factorization_general}}\label{proof:factorization_general}
To prove Theorem \ref{thm:factorization_general} we first need the following lemma.
\begin{lemma}\label{lemma:residual_independence}
Given $d\in\{0,\dots,k-1\}$ and $p,p'\in\{0,\dots,d\}$ with $p\neq p'$, 
\begin{align*}
\delta^{(p)}\bbeta^{(d)}
\bot
\delta^{(p')}\bbeta^{(d)}
\,|\,
\bbeta\backslash\bbeta^{(d)}
\,.
\end{align*}
\end{lemma}
\begin{proof}
Let $d\in\{0,\dots,k-1\}$, $p,p'\in\{0,\dots,d\}$ with $p< p'$ and $\bbeta\backslash\bbeta^{(d)}$ be fixed.
To make the notation more compact we denote $\E[\cdot|\bbeta\backslash\bbeta^{(d)}]$ by $\tE[\cdot]$, $\phi^{(p)}_{r}\bbeta^{(d)}$ by $\tilde{\phi}^{(p)}_{r}$ and $P(X_{p'}=r'|X_p=r)$ by $P(r'|r)$ for all $r\in T_p$ and $r'\in T_{p'}$.
By replacing $\beta_t$ with $\beta_t-\tE[\beta_t]$ we can suppose without loss of generality that $\tE[\beta_t]=0$ for all $t\in T_d$ and therefore 
$\tE[\tilde{\phi}^{(p)}_{r}]=\tE[\tilde{\phi}^{(p')}_{r'}]=0$ and 
$\tE[\delta^{(p)}_r\bbeta^{(d)}]=\tE[\delta^{(p')}_{r'}\bbeta^{(d)}]=0$
for all $r\in T_{p}$ and $r'\in T_{p'}$.
By definition $\tilde{\phi}^{(p)}_r=\sum_{s\in T_{p'}}
P(s|r)
\tilde{\phi}^{(p')}_s$ and therefore
\begin{equation}
\tE[\tilde{\phi}^{(p)}_r\tilde{\phi}^{(p')}_{r'}]
=
\sum_{s\in T_{p'}}
P(s|r)
\tE[\tilde{\phi}^{(p')}_s\tilde{\phi}^{(p')}_{r'}]
=
P(r'|r)
\tE[(\tilde{\phi}^{(p')}_{r'})^2]\,,\label{eq:phi_dot_product}
\end{equation}
where we used $\tE[\tilde{\phi}^{(p')}_s\tilde{\phi}^{(p')}_{r'}]=0$  for $r'\neq s$ and $r',s\in T_{p'}$, which follows from the conditional independence of $(\beta_t)_{t\in T_d\,,\,t\succeq r'}$ and $(\beta_t)_{t\in T_d\,,\,t\succeq s}$ given $\bbeta\backslash\bbeta^{(d)}$.
Note that $P(r'|r)$ in \eqref{eq:phi_dot_product} could be 0.
From $\tE[\beta_t]=0$ for any $t\in T_d$ and \eqref{eq:k_symmetry_rescaled} we have $\tE[\beta_t^2]=P(t)^{-1}$ and therefore
\begin{align}
\tE[(\tilde{\phi}^{(p')}_{r'})^2]
=&
\tE[(\sum_{t\in T_d}P(t|r')\beta_t)^2]
=
\sum_{t\in T_d}P(t|r')^2\tE[\beta_t^2]
\nonumber\\
=&
\sum_{t\in T_d}\frac{P(t|r')^2}{P(t)}
=
\frac{1}{P(r')}\sum_{t\in T_d}P(t|r')
=
\frac{1}{P(r')}\,.
\label{eq:phi_squared}
\end{align}
Combining \eqref{eq:phi_dot_product} and \eqref{eq:phi_squared} we have 
$\tE[\tilde{\phi}^{(p)}_r\tilde{\phi}^{(p')}_{r'}]=0$ if $r'\nsucc r$ and
\begin{align}\label{eq:zero_simplification}
\tE[\tilde{\phi}^{(p)}_r\tilde{\phi}^{(p')}_{r'}]
=&
\frac{P(r'|r)}{P(r')}
=
\frac{1}{P(r)}
&
\hbox{if }r'\succ r\,.
\end{align}
From the last equality and the definition of $\delta^{(p)}_r\bbeta^{(d)}$ in \eqref{eq:delta_defi} we have
\begin{align*}
\tE&[\delta^{(p)}_r\bbeta^{(d)}\delta^{(p')}_{r'}\bbeta^{(d)}]
=\\&
\tE[(\tilde{\phi}^{(p)}_r-\tilde{\phi}^{(p-1)}_{pa(r)})(\tilde{\phi}^{(p')}_{r'}-\tilde{\phi}^{(p'-1)}_{pa(r')})]
=\\&
\tE[
\tilde{\phi}^{(p)}_r\tilde{\phi}^{(p')}_{r'}
-\tilde{\phi}^{(p-1)}_{pa(r)}\tilde{\phi}^{(p')}_{r'}
-\tilde{\phi}^{(p)}_r\tilde{\phi}^{(p'-1)}_{pa(r')}
+\tilde{\phi}^{(p-1)}_{pa(r)}\tilde{\phi}^{(p'-1)}_{pa(r')}
]
=\\&
\frac{P(r'|r)}{P(r')}
-
\frac{P(r'|pa(r))}{P(r')}
-
\frac{P(pa(r')|r)}{P(pa(r'))}
+
\frac{P(pa(r')|pa(r))}{P(pa(r'))}
=0,
\end{align*}
where the last equality is trivial if $pa(r)\nprec r'$ and can be deduced from \eqref{eq:zero_simplification} otherwise.
The desired conditional independence follows from
$\tE[\delta^{(p)}_r\bbeta^{(d)}\delta^{(p')}_{r'}\bbeta^{(d)}]=0=\tE[\delta^{(p)}_r\bbeta^{(d)}]\tE[\delta^{(p')}_{r'}\bbeta^{(d)}]$
for all $r\in T_p$ and $r'\in T_{p'}$.
\end{proof}

\begin{proof}[Proof of Theorem \ref{thm:factorization_general}]
Theorem \ref{thm:factorization_general} follows easily from Lemmas \ref{lemma:delta_expansion} and \ref{lemma:residual_independence} as follows.
For each $d\in\{0,\dots,k-1\}$
the sampling step 
\begin{align}\label{eq:sample_original}
\bbeta^{(d)}(s+1)
\;\sim\;
\mathcal{L}\left(\bbeta^{(d)}|
\big(\bbeta^{(\ell)}(s+1)\big)_{0\leq \ell<d}\,,
\big(\bbeta^{(\ell)}(s)\big)_{d<\ell\leq k-1}
\right)
\end{align}
in Sampler \ref{sampler:GSk} is equal in distribution to sampling jointly the (d+1) residuals 
\begin{align*}
(\delta^{(p)}\bbeta^{(d)}(s+1))_{p=0}^d
\;\sim\;
\mathcal{L}\left((\delta^{(p)}\bbeta^{(d)})_{p=0}^d|
\big(\bbeta^{(\ell)}(s+1)\big)_{0\leq \ell<d}\,,
\big(\bbeta^{(\ell)}(s)\big)_{d<\ell\leq k-1}
\right)\,.
\end{align*}
From the conditional independence statement in Lemma  \ref{lemma:residual_independence} the latter is equivalent to sampling independently each residual $\delta^{(p)}\bbeta^{(d)}(s+1)$ from
$$
\mathcal{L}\left(\delta^{(p)}\bbeta^{(d)}|
\big(\bbeta^{(\ell)}(s+1)\big)_{0\leq \ell<d}\,,
\big(\bbeta^{(\ell)}(s)\big)_{d<\ell\leq k-1}
\right)\,.
$$
Moreover, from Lemma \ref{lemma:delta_expansion}
\begin{align*}
\mathcal{L}\left(\delta^{(p)}\bbeta^{(d)}|
\bbeta\backslash\bbeta^{(d)}
\right)
=
\mathcal{L}\left(\delta^{(p)}\bbeta^{(d)}|
(\delta^{(p)}\bbeta^{(\ell)})_{\ell\in\{p,\dots,k-1\}\backslash d}
\right)\,.
\end{align*}
Therefore the original sampling step in \eqref{eq:sample_original} is equivalent to sampling independently
\begin{align}\label{eq:sample_reparametrized}
\delta^{(p)}\bbeta^{(d)}(s+1)
\;\sim\;
\mathcal{L}\left(\delta^{(p)}\bbeta^{(d)}|
\big(\delta^{(p)}\bbeta^{(\ell)}(s+1)\big)_{p\leq \ell<d}\,,
\big(\delta^{(p)}\bbeta^{(\ell)}(s)\big)_{d<\ell\leq k-1}
\right)\,,
\end{align}
for $p=0,\dots,d$.
The thesis follows from the equivalence between \eqref{eq:sample_original} and \eqref{eq:sample_reparametrized}.
\end{proof}

\subsection*{Proofs of Corollary \ref{coroll:rate_k_levels}, Theorem \ref{thm:kronecker} and Theorem \ref{thm:ordering_CP}}
\begin{proof}[Proof of Corollary \ref{coroll:rate_k_levels}]
The map from $\bbeta_T$ to $(\delta^{(0)}\bbeta,\dots,\delta^{(k-1)}\bbeta)$ is an injective linear transformation. 
The injectivity holds because for any $d\in\{0,\dots,k-1\}$ and $t\in T_d$ we can reconstruct $\beta_t$ from $(\delta^{(0)}\bbeta^{(d)},\dots,\delta^{(d)}\bbeta^{(d)})$
$$
\sum_{r\preceq t}\delta_r^{(\ell(r))}\bbeta^{(d)}
=
\phi_{t_0}^{(0)}\bbeta^{(d)}+
\sum_{t_0\prec r\preceq t}\left(\phi_r^{(\ell(r))}\bbeta^{(d)}-\phi_{pa(r)}^{(\ell(r)-1)}\bbeta^{(d)}\right)
=
\phi_{t}^{(d)}\bbeta^{(d)}
=
\beta_t\,.
$$
It follows that 
$(\delta\bbeta(s))_{s\in \N}=(\delta^{(0)}\bbeta(s),\dots,\delta^{(k-1)}\bbeta(s))_{s\in \N}$ is a Markov chain with the same rate of convergence of the original chain $(\bbeta(s))_{s\in \N}$.
Then the thesis follows from Theorem \ref{thm:factorization_general} and the fact that the rate of convergence of a collection of independent Markov chains equals the supremum of the rates of convergence of the single chains.
\end{proof}

\begin{proof}[Proof of Theorem \ref{thm:kronecker}]
We are interested in the rate of convergence of 
the blocked sampler targeting 
$\delta^{(p)}\bbeta=(\delta^{(p)}\bbeta^{(p)},\dots,\delta^{(p)}\bbeta^{(k-1)})$ 
and evolving according to
\eqref{eq:blocked_gibbs}.
Consider first the case $p\in\{1,\dots,k-1\}$.
Note that $\delta^{(p)}\bbeta$ has a singular variance-covariance matrix because for each $t\in T_{p-1}$ and $d\in\{p,\dots,k-1\}$ it follows from \eqref{eq:phi_defi} and \eqref{eq:delta_defi} that
\begin{align}\label{eq:delta_sum_zero}
\sum_{r\in ch(t)}P(r|t)\delta^{(p)}_r\bbeta^{(d)}=0
\end{align}
and therefore some elements of $(\delta^{(p)}\bbeta)$ are linear combinations of the others.
In order to use standard tools it is more convenient to work with non-singular Gaussian random vectors.
To do so it is sufficient to consider a sub-vector of $\delta^{(p)}\bbeta$ obtained by removing from $T_{p}$ one children node for each parent node in $T_{p-1}$.
More formally, let $f$ be an arbitrary map from $T_{p-1}$ to $T_{p}$ such that $f(t)\in ch(t)$ for all $t\in T_{p-1}$ and then define the subset $T'_p\subseteq T_p$ as $T'_p= T_p\backslash f(T_{p-1})$.
It is then easy to see that the resulting sub-vector
$\delta_{T'_p}^{(p)}\bbeta
=(\delta_{T'_p}^{(p)}\bbeta^{(p)},\dots,\delta_{T'_p}^{(p)}\bbeta^{(k-1)})
$
with $\delta_{T'_p}^{(p)}\bbeta^{(d)}=(\delta_r^{(p)}\bbeta^{(d)})_{r\in T'_p}$ for all $d\in\{p,\dots,k-1\}$ has an invertible variance-covariance matrix.
Moreover, since each $\delta^{(p)}\bbeta^{(d)}$ is a function of the corresponding $\delta_{T'_p}^{(p)}\bbeta^{(d)}$ via \eqref{eq:delta_sum_zero} it follows that the blocked sampler targeting 
$\delta^{(p)}\bbeta$ 
and evolving according to
\eqref{eq:blocked_gibbs} is equivalent in distribution to a blocked Gibbs Sampler targeting $\delta_{T'_p}^{(p)}\bbeta$ and evolving according to
\begin{align}\label{eq:blocked_gibbs_reduced}
\delta^{(p)}_{T'_p}\bbeta^{(d)}(s+1)
\;\sim\;
\mathcal{L}\left(\delta^{(p)}_{T'_p}\bbeta^{(d)}|
(\delta^{(p)}_{T'_p}\bbeta^{(\ell)}(s+1))_{p\leq \ell<d},(\delta^{(p)}_{T'_p}\bbeta^{(\ell)}(s))_{d<\ell\leq k-1}
\right)\,,
\end{align}
for $d\in\{p,\dots,k-1\}$.
Let 
 $A_{\delta^{(p)}_{T'_p}\bbeta}=(A_{\delta_r^{(p)}\bbeta^{(d)}\,\delta_{r'}^{(p)}\bbeta^{(d')}})_{r,r'\in T'_p\,,\, d,d'\in\{p,\dots,k-1\}}$ be the $A$-matrix associated to the Gibbs sampler in \eqref{eq:blocked_gibbs_reduced}, defined by
\begin{multline}\label{eq:A_mat_defi}
\E[\delta^{(p)}_r\bbeta^{(d)}|\delta^{(p)}_{T'_p}\bbeta\backslash \delta^{(p)}_r\bbeta^{(d)}]
-\E[\delta^{(p)}_r\bbeta^{(d)}]
=\\
\sum_{\delta^{(p)}_{r'}\bbeta^{(d')}\in \delta^{(p)}_{T'_p}\bbeta\backslash \delta^{(p)}_r\bbeta^{(d)}}
A_{\delta_r^{(p)}\bbeta^{(d)}\,\delta_{r'}^{(p)}\bbeta^{(d')}}
\left(\delta^{(p)}_{r'}\bbeta^{(d')}-\E[\delta^{(p')}_{r'}\bbeta^{(d')}]\right)
\,.
\end{multline}
See the discussion of Theorem \ref{thm:gaussian_AR}
for some details and references on $A$-matrices.
Then Lemma \ref{lemma:delta_expansion} implies that 
$A_{\delta_r^{(p)}\bbeta^{(d)}\,\delta_{r'}^{(p)}\bbeta^{(d')}}=c_{dd'}$
if $r=r'$ and $d'\in\{p,\dots,k-1\}\backslash d$ and 0 otherwise.
The latter is equivalent to the equation
\begin{equation}\label{eq:A_matrix_delta}
A_{\delta^{(p)}_{T'_p}\bbeta}=\left(C^{(p)}-\I_{k-p}\right)\varotimes \I_{|T'_p|}
\end{equation}
where $C^{(p)}$ is the $(k-p)\times (k-p)$ square matrix $C^{(p)}=(c_{dd'})_{d,d'=p}^{k-1}$, $\I_n$ denotes the $n$ dimensional identity matrix, $\varotimes$ denotes the Kronecker product of matrices and $|T'_p|=|T_p|-|T_{p-1}|$ is the cardinality of $T'_p$.
Theorem \ref{thm:gaussian_AR} implies that
 the rate of convergence of the Gibbs sampler in \eqref{eq:blocked_gibbs_reduced} equals the largest modulus eigenvalue of $B=(\I_{(k-p)|T'_p|}-L)^{-1}U$, where $L$ is the lower triangular part of $A_{\delta^{(p)}_{T'_p}\bbeta}$ and $U=A_{\delta^{(p)}_{T'_p}\bbeta}-L$.
Using basic properties of the Kronecker product we can see that 
\begin{align*}
B=&
(\I_{(k-p)}\varotimes\I_{|T'_p|}-\tilde{L}\varotimes\I_{|T'_p|})^{-1}\left(\tilde{U}\varotimes \I_{|T'_p|}\right)
\\=&
\left((\I_{(k-p)}-\tilde{L})\varotimes \I_{|T'_p|}\right)^{-1}\left(\tilde{U}\varotimes \I_{|T'_p|}\right)
\\=&
\left((\I_{(k-p)}-\tilde{L})^{-1}\varotimes \I_{|T'_p|}\right)\left(\tilde{U}\varotimes \I_{|T'_p|}\right)
\\=&
\left((\I_{(k-p)}-\tilde{L})^{-1} \tilde{U}\right)\varotimes \I_{|T'_p|}
\end{align*}
where $\tilde{L}$ is the lower triangular part of $(C^{(p)}-\I_{k-p})$ and $\tilde{U}=(C^{(p)}-\I_{k-p})-\tilde{L}$.
From
$B=\tilde{B}\varotimes \I_{|T'_p|}$, where $\tilde{B}=(\I_{(k-p)}-\tilde{L})^{-1} \tilde{U}$, it follows that the unique eigenvalues of $B$ are the same as the unique eigenvalues of $\tilde{B}$ and thus the largest modulus eigenvalue of $B$ equals the one of $\tilde{B}$.
The case $p=0$ is analogous (with no need to consider a sub-vector of $\delta^{(0)}\bbeta$ and $T'_0$ being equal to $T_0$ itself) and trivial to check.
\end{proof}

\begin{proof}[Proof of Theorem \ref{thm:ordering_CP}]
Theorem \ref{thm:factorization_general} shows that $(\delta^{(p)}\bgamma(s))_s$ for $p\in\{0,\dots,k-1\}$ are $k$ independent Markov chains.
Arguing as in the proof of Theorem \ref{thm:kronecker} above for each $p$ we consider $(\delta_{T'_p}^{(p)}\bgamma(s))_s$ rather than $(\delta^{(p)}\bgamma(s))_s$ to avoid working with singular Gaussian random vectors.
For any $p$ from 0 to $k-1$, \eqref{eq:tree_cond_indep} implies that the $Q$-matrix of $\delta_{T'_p}^{(p)}\bgamma$ is tridiagonal
(with $k-p$ blocks corresponding to $\delta_{T'_p}^{(p)}\bgamma^{(p)}$ up to $\delta_{T'_p}^{(p)}\bgamma^{(k-1)}$).
It follows by Theorem 5 of \citep{RobertsSahu1997} that the largest modulus eigenvalue of the autoregressive matrix 
$B$ 
 of the Gibbs Sampler in \eqref{eq:blocked_gibbs_reduced} coincides with the square of the largest eigenvalue of the corresponding $A$-matrix.
We denote the latter by $\lambda(A_{\delta_{T'_p}^{(p)}\bgamma(s)})^2$, where $\lambda(\cdot)$ is the function mapping a symmetric matrix to its largest eigenvalue.
Then, using \eqref{eq:A_matrix_delta} from the proof of Theorem \ref{thm:kronecker}, it follows $\lambda(A_{\delta_{T'_p}^{(p)}\bgamma(s)})^2=\lambda(C^{(p)}-\I_{k-p})^2$ and thus the rate of convergence of $(\delta^{(p)}\bgamma(s))_s$ is given by $\rho(\delta^{(p)}\bgamma(s))=\lambda(C^{(p)}-\I_{k-p})^2$. 
Noting that $C^{(p+1)}-\I_{k-(p+1)}$ is obtained from $C^{(p)}-\I_{k-p}$ by removing the first row and column, the desired inequality $\rho(\delta^{(p)}\bgamma(s))=\lambda(C^{(p)}-\I_{k-p})^2 \geq \lambda(C^{(p+1)}-\I_{k-(p+1)})^2=\rho(\delta^{(p+1)}\bgamma(s))$ follows by applying the Cauchy interlacing theorem (see e.g. \citealp{bhatia2013}), which states that the eigenvalues of a principal submatrix of a symmetric matrix interlace the original eigenvalues.
\end{proof}

\subsection{Convergence rates for the example in Section 7.6 of the paper}
\label{sec:supp_rates_expressions}

If instead $\lambda_1=0$ (i.e. using $(\gamma_i)_i$ with $\gamma_i=\mu+a_i$ at level 1) we have four different parametrization with rates of convergence given by
\begin{align*}
\rho_{000}
&=
\frac{1}{2}\left(
1
+\ratio{2}{3}(\ratio{4}{3}-\ratio{1}{2})
+\sqrt{
(1+\ratio{2}{3}(\ratio{4}{3}-\ratio{1}{2}))^2
-4 \ratio{2}{1}\ratio{2}{3}\ratio{4}{3}}
\right)\,,
\\
\rho_{001}
&=
\frac{1}{2}\left(
1+\ratio{2}{4}(\ratio{3}{4}-\ratio{1}{2})
+\sqrt{
(1+\ratio{2}{4}(\ratio{3}{4}-\ratio{1}{2}))^2
-4 \ratio{2}{1}\ratio{2}{4}\ratio{3}{4}}
\right)\,,
\\
\rho_{011}
&=
\frac{1}{2}\left(
1-\ratio{1}{4}\ratio{4}{2}\ratio{4}{3}+
\ratio{2}{4}\ratio{3}{4}
+\sqrt{
(1-\ratio{1}{4}\ratio{4}{2}\ratio{4}{3}+
\ratio{2}{4}\ratio{3}{4})^2
-4 \ratio{2}{4}\ratio{3}{4}}
\right)\,,
\\
\rho_{010}
&=
\frac{1}{2}\left(
1-\ratio{1}{3}\ratio{3}{2}\ratio{3}{4}+
\ratio{2}{3}\ratio{4}{3}
+\sqrt{
(1-\ratio{1}{3}\ratio{3}{2}\ratio{3}{4}+
\ratio{2}{3}\ratio{4}{3})^2
-4 \ratio{2}{3}\ratio{4}{3}}
\right)\,.
\end{align*}

\section{Proofs for hierarchical models with three levels}

\subsection*{Proof of Theorems \ref{thm:factorization_3}, \ref{thm:rates_ordering_3} and \ref{thm:rates_3}}\label{proof:factorization_3}
Theorems \ref{thm:factorization_3}, \ref{thm:rates_ordering_3} and \ref{thm:rates_3} are substantially special cases of the analogous theorems for $k$-levels.
In particular Theorem \ref{thm:factorization_3} is a special case of Theorem \ref{thm:factorization_general} for $k=2$ and Theorems \ref{thm:rates_ordering_3} and \ref{thm:rates_3} can be directly verified as follows.
Using Corollary \ref{coroll:rate_k_levels} we can evaluate the rates of convergence of the three subchains
$
\rho(\delta^{(0)}\bbeta(s))
$, 
$\rho(\delta^{(1)}\bbeta(s))$
and
$\rho(\delta^{(2)}\bbeta(s))$
for the four parametrizations under consideration in Section \ref{sec:multigrid} and check by inspection that 
$$
\rho(\delta^{(0)}\bbeta(s))\geq
\rho(\delta^{(1)}\bbeta(s))\geq
\rho(\delta^{(2)}\bbeta(s))=0
$$
and that the rates of convergence $\rho(\delta^{(0)}\bbeta(s))$ are the ones given by Theorem \ref{thm:rates_3}.
In particular the rates of convergence of $\delta^{(0)}\bbeta(s))$, 
$\delta^{(1)}\bbeta(s)$
and
$\delta^{(2)}\bbeta(s)$
under $GS(\bbeta)$ are given by the following Table.
The inequality $\rho(\delta^{(1)}\bbeta(s))\geq
\rho(\delta^{(2)}\bbeta(s))$ is trivial, while the one $\rho(\delta^{(0)}\bbeta(s))\geq
\rho(\delta^{(1)}\bbeta(s))$ can be checked case by case using the expressions in Figure \ref{table:rates_3levels} and the fact that the ratio of variances lie between 0 and 1.
\begin{figure}[h!]
\centering
\begin{tabular}{ l |c| c | c }
& $\rho(\delta^{(0)}\bbeta(s))$ &  $\rho(\delta^{(1)}\bbeta(s))$ &  $\rho(\delta^{(2)}\bbeta(s))$ 
\\
\hline			
$(\mu,\avec,\bvec)$ & 
$\frac{\tvar{a}}{\tvar{a}+\tvar{e}}\vee\frac{\tvar{b}}{\tvar{b}+\tvar{e}}$ & 
$\frac{\tvar{a}}{\tvar{a}+\tvar{e}}\frac{\tvar{b}}{\tvar{b}+\tvar{e}}$ &
0
\\
\hline			
 $(\mu,\gavec,\etavec)$ &
$1-\frac{\tvar{a}}{\tvar{a}+\tvar{b}}\frac{\tvar{b}}{\tvar{b}+\tvar{e}}$ &
$\frac{\tvar{a}}{\tvar{a}+\tvar{b}}\left(1-\frac{\tvar{b}}{\tvar{b}+\tvar{e}}\right)$ &
0
\\
\hline
$(\mu,\gavec,\bvec)$ & 
$1-\frac{\tvar{a}}{\tvar{a}+\tvar{e}}\frac{\tvar{e}}{\tvar{b}+\tvar{e}}$ &
$\frac{\tvar{a}}{\tvar{a}+\tvar{e}}\left(1-\frac{\tvar{e}}{\tvar{b}+\tvar{e}}\right)$ &
0
\\
\hline
$(\mu,\avec,\etavec)$ &
$\frac{\tvar{a}}{\tvar{a}+\tvar{b}}\vee\frac{\tvar{e}}{\tvar{b}+\tvar{e}}$ &
$\frac{\tvar{a}}{\tvar{a}+\tvar{b}}\frac{\tvar{e}}{\tvar{b}+\tvar{e}}$ &
0
\\
\hline
\end{tabular}\caption{Rates of convergence of $\delta^{(0)}\bbeta(s))$, 
$\delta^{(1)}\bbeta(s)$
and
$\delta^{(2)}\bbeta(s)$ for $GS(\bbeta)$ under various parametrizations.}\label{table:rates_3levels}
\end{figure}
%

\subsection*{Proofs of Corollaries \ref{thm:rate_eq_skeleton}-\ref{coroll:optimal_param_2_levels} and Theorems \ref{thm:bespoke_parametrizations_2}- \ref{thm:rates_3_uneven_CP}}

\begin{proof}[Proof of Corollary \ref{thm:rate_eq_skeleton}]
By Theorem \ref{thm:factorization_3} the rate of convergence of the whole chain $(\bbeta(s))_{s\in \N}$ coincides with the maximum of the rates of the subchains, meaning that 
$\rho(\bbeta(s))$ equals the maximum of 
$\rho(\delta^{(0)}\bbeta(s))$,
$\rho(\delta^{(1)}\bbeta(s))$ and 
$\rho(\delta^{(2)}\bbeta(s))$.
By Theorem \ref{thm:rates_ordering_3} the latter equals $\rho(\delta^{(0)}\bbeta(s))$.
\end{proof}

\begin{proof}[Proof of Corollary \ref{coroll:optimal_param_3levels}]
Follows from Theorem \ref{thm:rates_3} by checking that for both 
$\balpha=\gavec$ and $\balpha=\avec$, the inequality
$\rho_{(\mu,\balpha,\etavec)}\leq \rho_{(\mu,\balpha,\bvec)}$
holds if and only if 
$\tvar{b}\geq\tvar{e}$;
and for both $\balpha=\gavec$ and $\balpha=\avec$ the inequality
$\rho_{(\mu,\gavec,\balpha)}\leq \rho_{(\mu,\avec,\balpha)}$
holds if and only if 
$\tvar{a}\geq\tvar{b}+\tvar{e}$.
\end{proof}

\begin{proof}[Proof of Theorem \ref{thm:bespoke_parametrizations_2}]
The Markov chain under consideration is a Gibbs Sampler sweeping through $(\mu,\beta_1,\dots,\beta_I)|\textbf{y}$ for some observed data $\textbf{y}=((y_{ij})_{j=1}^{J_i})_{i=1}^I$ assuming Model \ref{model:2_nonsymmetric} and $\beta_i=\gamma_i-\lambda_i\mu$ with $\lambda_i\in\{0,1\}$.
To compute the Gibbs Sampler rate of convergence we first need to compute the $(I+1)\times (I+1)$ matrix $A$ indexed by $(\alpha_1,\alpha_2)\in\{\mu,\beta_1,\dots,\beta_I\}\times \{\mu,\beta_1,\dots,\beta_I\}$ and defined as $A_{\alpha_1\alpha_2}=-\frac{Q_{\alpha_1\alpha_2}}{Q_{\alpha_1\alpha_1}}$
for $\alpha_1\neq \alpha_2$ and $0$ for $\alpha_1=\alpha_2$, where $Q$ is the precision matrix of $(\mu,\beta_1,\dots,\beta_I)|\textbf{y}$.
See the discussion of Theorem \ref{thm:gaussian_AR} for more details and references on the derivation of $A$-matrices.
By computing the precision matrix of $(\mu,\beta_1,\dots,\beta_I)|\textbf{y}$ it can be seen that $A$ is given by 
\begin{equation*}
\begin{pmatrix}
0 & A_{\mu\beta_1} & \cdots & A_{\mu\beta_I}\\
A_{\beta_1\mu} & 0 & \cdots &0\\
\vdots & \vdots & \ddots & \vdots\\
A_{\beta_I\mu} & 0 & \cdots & 0 
\end{pmatrix}
\end{equation*}
where for all $i$
$$
A_{\mu\beta_i}=\frac{-\lambda_i\tprec{i}+(1-\lambda_i)\precision{a}}{\sum_{\ell=1}^I\lambda_\ell\tprec{\ell}+(1-\lambda_\ell)\precision{a}}
\,,\quad
A_{\beta_i\mu}=\frac{-\lambda_i\tprec{i}+(1-\lambda_i)\precision{a}}{\tprec{i}+\precision{a}}\,.
$$
From Theorem \ref{thm:gaussian_AR} the rate of convergence of the Gibbs sampler of interest equals the largest modulus eigenvalue of the autoregressive matrix $B=(\I_{I+1}-L)^{-1}U$, where $\I_{I+1}$ is the $(I+1)$-dimensional identity matrix, $L$ is the lower triangular part of $A$ and $U=A-L$.
In this case $B$ is given by
\begin{equation*}
B=
\begin{pmatrix}
0 &A_{\mu\beta_1}&\cdots&A_{\mu\beta_I}
\end{pmatrix}\varotimes
\begin{pmatrix}
1 \\A_{\beta_1\mu}\\\vdots\\A_{\beta_I\mu}
\end{pmatrix}
=
\begin{pmatrix}
0 & A_{\mu\beta_1} & \cdots & A_{\mu\beta_I}\\
0 &  A_{\mu\beta_1}A_{\beta_1\mu} & \cdots & A_{\mu\beta_I}A_{\beta_1\mu}\\
\vdots &  \vdots & \ddots & \vdots\\
0 &  A_{\mu\beta_1}A_{\beta_I\mu}0 & \cdots & A_{\mu\beta_I}A_{\beta_I\mu}\,.
\end{pmatrix}
\end{equation*}
Finally note that $B$ has $I$ eigenvalues equal to 0 and one equal to 
$$
\sum_{i=1}^I A_{\mu\beta_i}A_{\beta_i\mu}
=
\frac{
\sum_{i\,:\,\lambda_i=1}\tprec{i}\frac{\tprec{i}}{\tprec{i}+\precision{a}}
+
\sum_{i\,:\,\lambda_i=0}\precision{a}\frac{\precision{a}}{\tprec{i}+\precision{a}}
}{
\sum_{i\,:\,\lambda_i=1}\tprec{i}+\sum_{i\,:\,\lambda_i=0}\precision{a}
}\,.
$$
\end{proof}

\begin{proof}[Proof of Corollary \ref{coroll:optimal_param_2_levels}]
Starting from \eqref{eq:rate_2_heterogenous} we can see that for any $i\in\{1,\dots,I\}$
\begin{equation}\label{eq:centering_i}
\rho_{\lambda_1\dots \lambda_{i-1}0\lambda_{i+1}\dots \lambda_I}-\rho_{\lambda_1\dots \lambda_{i-1}1\lambda_{i+1}\dots \lambda_I}
=
\frac{\frac{\precision{a}\tprec{i}}{\precision{a}+\tprec{i}}(\precision{a}-\tprec{i})}{(\rho_{-i}+\precision{a})(\rho_{-i}+\tprec{i})}\,,
\end{equation}
where 
$
\rho_{-i}=\sum_{\ell\neq i \,:\,\lambda_\ell=0}\tprec{\ell}\frac{\tprec{\ell}}{\tprec{\ell}+\precision{a}}
+\sum_{\ell\neq i \,:\,\lambda_\ell=1}\tprec{a}\frac{\tprec{a}}{\tprec{\ell}+\precision{a}}\geq 0
$.
Equation \eqref{eq:centering_i} implies that 
$\rho_{\lambda_1\dots \lambda_{i-1}0\lambda_{i+1}\dots \lambda_I}>\rho_{\lambda_1\dots \lambda_{i-1}1\lambda_{i+1}\dots \lambda_I}$ if and only if $\precision{a}>\tprec{i}$, which in turn implies the statement in Corollary \ref{coroll:optimal_param_2_levels}.
\end{proof}

\begin{proof}[Proof of Theorem \ref{thm:rates_3_uneven_CP}]\label{appendix:proof_bound}
Given an instance of Model \ref{model:3_nonsymmetric} with variance terms $(\var{a},(\var{b,i})_i,(\var{e,ij})_{ij})$ satisfying 
\begin{equation}\label{eq:assumption_for_bound}
\ratio{a}{b}^{(i)}\geq\ratio{a}{b}^{(i')}\ratio{e}{b}^{(i')}, \ \ \ \ \hbox{ for every }i,i'\in\{1,\dots,I\}\ ,
\end{equation} 
the proof will proceed by comparing the original Gibbs Sampler with an auxiliary Gibbs Sampler targeting a different instance of Model \ref{model:3_nonsymmetric} with variance terms $(\var{a},(\var{b,i})_i,(\bar{\sigma}_{e,ij}^2)_{ij})$ satisfying \eqref{eq:k_symmetry_appendix} and thus allowing direct analysis using Corollary \ref{coroll:rate_k_levels_CP}.
In the context of Model \ref{model:3_nonsymmetric}, \eqref{eq:k_symmetry_appendix} reduces to requiring $\sum_{j=1}^{J_i}\rho_{\gamma_i \eta_{ij}}^2$ to be constant over $i$, where $\rho_{\gamma_i \eta_{ij}}$ is the partial correlation $Corr(\gamma_i,\eta_{ij}|\mu,(\gamma_\ell)_{\ell\neq i},(\eta_{\ell s})_{(\ell s)\neq (ij)})$ as in Section \ref{sec:symmetry} of the paper.
By computing the partial correlations of Model \ref{model:3_nonsymmetric} it can be checked that $\sum_{j=1}^{J_i}\rho_{\gamma_i \eta_{ij}}^2=r^{(i)}_{a,b}r^{(i)}_{e,b}$, where $r^{(i)}_{a,b}$ and $r^{(i)}_{e,b}$ are defined in Theorem \ref{thm:rates_3_uneven_CP}.
For each $i=1,\dots,I$ we define auxiliary variance terms $(\bar{\sigma}_{e,ij}^2)_{j=1}^{J_i}$ such that $\bar{\sigma}_{e,ij}^2\geq\var{e,ij}$ for all $j=1,\dots,J_{i}$ and 
\begin{equation}\label{eq:proof_bounds}
r^{(i)}_{a,b}\,
\frac{1}{J_i}\sum_{j=1}^{J_i}
\frac{K_{ij}^{-1}\bar{\sigma}_{e,ij}^2}{\var{b,i}+K_{ij}^{-1}\bar{\sigma}_{e,ij}^2}
=
\max_{\ell=1,\dots,I}
r^{(\ell)}_{a,b}r^{(\ell)}_{e,b}\,.
\end{equation}
Such $(\bar{\sigma}_{e,ij}^2)_{j=1}^{J_i}$ exist because $r^{(i)}_{a,b}\geq \max_{\ell=1,\dots,I}r^{(\ell)}_{a,b}r^{(\ell)}_{e,b}$ by \eqref{eq:assumption_for_bound} and the left hand side of \eqref{eq:proof_bounds} can take any value in $(0,r^{(i)}_{a,b}]$ for $(\bar{\sigma}_{e,ij}^2)_{j=1}^{J_i}$ belonging to $[0,\infty)$.
\eqref{eq:proof_bounds} implies that the instance of Model \ref{model:3_nonsymmetric} with variance terms $(\var{a},(\var{b,i})_i,(\bar{\sigma}_{e,ij}^2)_{ij})$ satisfies \eqref{eq:k_symmetry_appendix} 
with $c_0=\sum_{i=1}^I\rho_{\mu\gamma_i}^2=1-\frac{1}{I}\sum_{i=1}^Ir^{(i)}_{a,b}$
and
$c_1=\sum_{j=1}^{J_i}\rho_{\gamma_i\eta_{ij}}^2=\max_{\ell=1,\dots,I}
r^{(\ell)}_{a,b}r^{(\ell)}_{e,b}$.
As discussed in Example \ref{ex:non_symm}, for models with centred parametrization like Model \ref{model:3_nonsymmetric}, \eqref{eq:k_symmetry_appendix} implies \eqref{eq:k_symmetry} and, after rescaling, \eqref{eq:k_symmetry_rescaled}.
In this case the matrix $C=(c_{dp})_{d,p=0}^2$ is given by
\begin{equation*}
C=
\begin{pmatrix}
1 & \sqrt{c_0} & 0\\
\sqrt{c_0} & 1 & \sqrt{c_1}\\
0 & \sqrt{c_1} & 1
\end{pmatrix}\,.
\end{equation*}
Therefore, by Corollary \ref{coroll:rate_k_levels_CP}, the rate of convergence of the Gibbs Sampler targeting the posterior distribution of Model \ref{model:3_nonsymmetric} with variance terms $(\var{a},(\var{b,i})_i,(\bar{\sigma}_{e,ij}^2)_{ij})$ is given by $c_0+c_1$, which is the largest squared eigenvalue of $C-\I_3$, where $\I_3$ is the 3-dimensional identity matrix.

Finally we show that the Gibbs Sampler rate of convergence induced by the auxiliary variance terms $(\var{a},\var{b,i},\bar{\sigma}_{e,ij}^2)$ is greater or equal than the original one given by $(\var{a},\var{b,i},\var{e,ij})$.
Denote by $Q$ the precision matrix of the original posterior distribution and by $\bar{Q}$ the auxiliary one.
By deriving $Q$ and $\bar{Q}$ from the definition of Model \ref{model:3_nonsymmetric}, it is easy to see that the only terms of $Q$ and $\bar{Q}$ affected by replacing $\var{e,ij}$ with $\bar{\sigma}_{e,ij}^2$ are $(Q_{\eta_{ij}\eta_{ij}})_{ij}$ and  $(\bar{Q}_{\eta_{ij}\eta_{ij}})_{ij}$.
Moreover $\bar{\sigma}_{e,ij}^2\geq\var{e,ij}$ implies $Q_{\eta_{ij}\eta_{ij}}=\frac{1}{\var{b,i}}+\frac{K_{ij}}{\var{e,ij}}\geq \frac{1}{\var{b,i}}+\frac{K_{ij}}{\bar{\sigma}_{e,ij}^2}=\bar{Q}_{\eta_{ij}\eta_{ij}}$.
The result then follows from the fact that the convergence rate of a deterministic scan Gibbs Sampler with single-site update is a non-increasing function of the diagonal elements of the target precision matrix (when the off-diagonal terms are kept constant), see Theorem 7 of \cite{RobertsSahu1997}.
\end{proof}

\section{Proofs for crossed models}

\begin{proof}[Proof of Theorem \ref{thm:centering_two_factors}]
The case $(\lambda_1,\lambda_2)=(1,1)$ follows directly from Theorem \ref{thm:multigrid_decomposition}.
Consider an arbitrary $(\lambda_1,\lambda_2)$ in $\{0,1\}^{2}$ and, analogously to \eqref{eq:residuals_crossed}, define
\begin{align*}
\bar{\beta}^{(s)}=\frac{1}{n_s}\sum_{i=1}^{n_s}\beta^{(s)}_{i}
\quad\hbox{and}\quad
\delta\bbeta^{(s)}=(\bbeta^{(s)}-\bar{\beta}^{(s)})\,,
\end{align*}
for $s\in\{1,2\}$.
It can be checked that the first part of Theorem \ref{thm:multigrid_decomposition} extends to any $(\lambda_1,\lambda_2)\in\{0,1\}^{2}$, meaning that
$\left((\mu,\mbeta^{(1)},\mbeta^{(2)})(t)\right)_{t=1}^\infty$, $\left(\delta\bbeta^{(1)}(t)\right)_{t=1}^\infty$ and $\left(\delta\bbeta^{(2)}(t)\right)_{t=1}^\infty$ are three independent Markov chains.
Also, $\left(\delta\bbeta^{(1)}(t)\right)_{t=1}^\infty$ and $\left(\delta\bbeta^{(2)}(t)\right)_{t=1}^\infty$ perform i.i.d.\ sampling from $\mathcal{L}(\delta\bbeta^{(1)}|\by)$ and $\mathcal{L}(\delta\bbeta^{(2)}|\by)$ respectively and thus have rate of convergence equal to 0. 
It follows that the rate of convergence of the original Gibbs Sampler $\left((\mu,\bbeta^{(1)},\bbeta^{(2)})(t)\right)_{t=1}^\infty$ coicides with the one of the (low dimensional) three component Gibbs sampler targeting $\mathcal{L}(\mu,\bar{\beta}^{(1)},\bar{\beta}^{(2)}|\by)$.
Denote by $A_{\lambda_1\lambda_2}$ the $A$-matrix of the Gibbs Sampler targeting $\mathcal{L}(\mu,\bar{\beta}^{(1)},\bar{\beta}^{(2)}|\by)$, and by $L_{\lambda_1\lambda_2}$ and $U_{\lambda_1\lambda_2}$ its lower and upper triangular parts.
Then it holds
\begin{equation*}
A_{01}=
\begin{pmatrix}
0 & 0  & 1\\
-r_1 & 0  & 0\\
1-r_{2} &  -r_{2} & 0
\end{pmatrix}
\,.
\end{equation*}
and it can be obtained with simple calculations that the 3 eigenvalues of 
$(I-L_{01})^{-1}U_{01}$ are $\{0,0,1-r_2(1-r_1)\}$.
It follows by Theorem \ref{thm:gaussian_AR} that
$\rho_{01}=1-r_1(1-r_2)$
and, by symmetry, 
$\rho_{10}=1-r_2(1-r_1)$.
Consider now $(\lambda_1,\lambda_2)=(0,0)$.
We have
\begin{equation*}
A_{00}=
\left(
\begin{array}{ccc}
 0 & \frac{n_1\tau_1+N\tau_e}{n_1\tau_1+n_2\tau_2+N\tau_e} & \frac{n_2\tau_2+N\tau_e}{n_1\tau_1+n_2\tau_2+N\tau_e} \\
 1 & 0 & -\frac{N\tau_e}{n_1\tau_1+N\tau_e} \\
 1 & -\frac{N\tau_e}{n_2\tau_2+N\tau_e} & 0
\end{array}
\right)
=
\left(
\begin{array}{ccc}
 0 &  \frac{r_2}{r_1+r_2-r_1r_2} &  \frac{r_1}{r_1+r_2-r_1r_2} \\
 1 & 0 & -r_1 \\
 1 & -r_2 & 0
\end{array}
\right)\,.
\end{equation*}
The matrix $(I-L_{00})^{-1}U_{00}$ has one zero eigenvalue and two eigenvalues given by
$\frac{1}{2} \left(1+r_1 r_2+q\right)$
and
$\frac{1}{2} \left(1+r_1 r_2-q\right)$
where
$$
q=\sqrt{(1+r_1 r_2)^2-\frac{4 r_1 r_2}{r_1+r_2-r_1r_2}}\,.
$$
It follows, again by Theorem \ref{thm:gaussian_AR}, that $\rho_{00}=\frac{1}{2} \left(1+r_1 r_2+q\right)$. 
To conclude the proof we now check that the inequality $\rho_{00}\geq 1+r_1r_2-\min\{r_1,r_2\}$ holds for all $0< r_1,r_2<1$.
The latter is equivalent to 
$$
q\geq1+r_1r_2-2\min\{r_1,r_2\}\,.
$$
Squaring both sides, rearranging and dividing by 4 one obtains
$$
-\frac{r_1 r_2}{r_1+r_2-r_1r_2}
-\min\{r_1,r_2\}^2
+
(1+r_1r_2)\min\{r_1,r_2\}
\geq
0\,.
$$
Dividing by $\min\{r_1,r_2\}$ and using $\frac{r_1r_2}{\min\{r_1,r_2\}}=\max\{r_1,r_2\}$
$$
-\frac{\max\{r_1,r_2\}}{r_1+r_2-r_1r_2}
-\min\{r_1,r_2\}
+
1+r_1r_2
\geq
0\,.
$$
The left-hand side equals 
\begin{align*}
&\frac{r_1+r_2-r_1r_2-\max\{r_1,r_2\}}{r_1+r_2-r_1r_2}
-\min\{r_1,r_2\}
+r_1r_2
\\&=
\frac{\min\{r_1,r_2\}-r_1r_2}{r_1+r_2-r_1r_2}
-(\min\{r_1,r_2\}
-r_1r_2)
\\&=
(\min\{r_1,r_2\}
-r_1r_2)\left(\frac{1}{r_1+(1-r_1)r_2}
-1\right)\,,
\end{align*}
which can be easily seen to be non-negative as both terms in the product are non-negative due to $0< r_1,r_2<1$.
\end{proof}

To prove Theorem \ref{thm:rates_identifiable_general} we first need the following lemma.
Consider the linear mapping from $(\mu,\ba)$ to $(\mu,\ma,\tilde{\delta}\ba)$ defined as
$$
\ma=(\ma^{(1)},\dots,\ma^{(k)})
\quad \hbox{and}\quad
\tilde{\delta}\ba=(\tilde{\delta}\ba^{(1)},\dots,\tilde{\delta}\ba^{(k)})
$$
where
$\ma^{(s)}=\frac{1}{n_s}\sum_{i=1}^{n_s}a^{(s)}_{i}$
and 
$\tilde{\delta}\ba^{(s)}=(\tilde{\delta} a^{(s)}_{i})_{i=1}^{n_s}$ for $s=1,\dots,k$ with
$$
\tilde{\delta} a^{(s)}_{i}=a^{(s)}_{i}-v^{(s)}_{i}\ma^{(s)}
\quad \hbox{and}\quad
v^{(s)}_{i}=
\frac{\|w^{(s)}\|^2-w_{i}}{\|w^{(s)}\|^2-n_s^{-1}}\,,
$$
if $\|w^{(s)}\|^2>n_s^{-1}$ and $v^{(s)}_{i}=n_s^{-1}$ if $\|w^{(s)}\|^2=n_s^{-1}$.
\begin{lemma}\label{lemma:orthogonal_identifiable}
Denote the $k$ linear constraints $c_1=\dots=c_k=0$ by $c=0$ for brevity. Then $(\mu,\ma)$ and $\tilde{\delta}\ba$ are conditionally independent given $\by$ and $c=0$, meaning that
\begin{align}\label{eq:orthogonal_identifiable}
\mathcal{L}(\mu,\ma,\tilde{\delta}\ba|\by,c=0)=&\mathcal{L}(\mu,\ma|\by,c=0)\times
\mathcal{L}(\tilde{\delta}\ba|\by,c=0)\,.
\end{align}
Moreover $\tilde{\delta}\ba^{(1)}$, \dots, $\tilde{\delta}\ba^{(k)}$ are conditionally independent given $\by$ and $c=0$, meaning that 
\begin{align}\label{eq:orthogonal_identifiable_2}
\mathcal{L}(\tilde{\delta}\ba|\by,c=0)=&\mathcal{L}(\tilde{\delta}\ba^{(1)}|\by,c=0)\times\dots\times
\mathcal{L}(\tilde{\delta}\ba^{(k)}|\by,c=0)\,.
\end{align}
\end{lemma}
\begin{proof}[Proof of Lemma \ref{lemma:orthogonal_identifiable}]
Since the joint distribution $\mathcal{L}(\mu,\ma,\tilde{\delta}\ba|\by,c=0)$ is multivariate Gaussian it is sufficient to check that the conditional expectations
$\E[\mu|\by,\ma,\tilde{\delta}\ba,c=0]$ and 
$\E[\ma^{(s)}|\by,\mu,\ma^{(-s)},\tilde{\delta}\ba,c=0]$ do not depend on $\tilde{\delta}\ba$ and similarly that 
$\E[\tilde{\delta}\ba^{(s)}|\by,\mu,\ma,\tilde{\delta}\ba^{(-s)},c=0]$
does not depend on $\mu$, $\ma$ and $\tilde{\delta}\ba^{(-s)}$ in order to deduce \eqref{eq:orthogonal_identifiable}.
With standard (but tedious) calculations
 for conditioned multivariate Gaussian distributions one can explicitly compute
\begin{align}
&\E[\tilde{\delta}a^{(s)}_i|\by,\mu,\ma,\tilde{\delta}\ba^{(-s)},c=0]=0\\
&\E[\mu|\by,\ma,\tilde{\delta}\ba,c=0]=\E[\mu|\by,\ma,\tilde{\delta}\ba]=\bar{y}-\sum_{s=1}^k\ma^{(s)}\\
&\E[\ma^{(s)}|\by,\mu,\ma^{(-s)},\tilde{\delta}\ba,c=0]=\tilde{r}_s(\bar{y}-\mu-\sum_{\ell\neq s}\ma^{(\ell)})\,,\label{eq:cond_exp_ma}
\end{align}
where $\tilde{r}_s$ is defined as
$$
\tilde{r}_s=
\frac{N\tau_e}{N\tau_e+n_s\tau_{s}}
\frac{n_s\|w^{(s)}\|^2-1}{n_s\|w^{(s)}\|^2}
\,,$$
with $\|w^{(s)}\|^2=\sum_{i=1}^{n_s}(w_i^{(s)})^2/((\sum_{i=1}^{n_s}w_i^{(s)})^2)$.
\end{proof}

\begin{proof}[Proof of Theorem \ref{thm:rates_identifiable_general}]
Using Lemma \ref{lemma:orthogonal_identifiable} and arguing as in the proof of Theorem \ref{thm:factorization_general} we can deduce that $\big((\mu,\ma)(t)\big)_{t=1}^\infty$, $\big(\tilde{\delta}\ba^{(1)}(t)\big)_{t=1}^\infty$, \dots, $\big(\tilde{\delta}^{(k)}\ba(t)\big)_{t=1}^\infty$ are $(k+1)$ Markov chains and they evolve independently.
Thus the rate of convergence of of the original Markov chain $\left((\mu,\ba)(t)\right)_{t=1}^\infty$ coincides with the supremum of the rate of convergence of these $(k+1)$ Markov chains. 
Since $\big(\tilde{\delta}^{(s)}\ba(t)\big)_{t=1}^\infty$ performs i.i.d.\ sampling for all $s=1,\dots,k$, its rate of convergence is 0 and thus the rate of convergence of $\left((\mu,\ba)(t)\right)_{t=1}^\infty$ coincides with the one of $\big((\mu,\ma)(t)\big)_{t=1}^\infty$.
The latter is a $(k+1)$ components Gaussian Gibbs Sampler with one dimensional components.
Using \eqref{eq:cond_exp_ma} and arguing as in the proof of Proposition 3 of \cite{papaspiliopoulos2018scalable} one can show that the $B$ matrix (as denoted in Theorem \ref{thm:gaussian_AR}) of such Gibbs Sampler scheme is 
\begin{equation*}
B=
\left(
\begin{array}{c|ccc}
0 & -1 & \dots & -1\\\hline
0 &  &  & \\
\vdots &  & L &\\
0  &  &  & 
\end{array}
\right)
\end{equation*}
where $L$ is a $K\times K$ lower triangular matrix with diagonal elements equal to $(\tilde{r}_1,\dots,\tilde{r}_K)$, and that the resulting rate of convergence of $(\mu(t),\ma(t))_t$ is $\max_{s\in\{1,\dots,k\}}\tilde{r}_s$, which is the desired thesis.
\end{proof}

\section{Full conditionals and additional material on the simulations}
\subsection{Full conditional distributions of GS(1,1) and GS(0,0)}\label{appendix:full_conditionals}
This section reports the full conditional distributions involved in Samplers GS(1,1) and GS(0,0) described in Section 2 of the paper. 
\begin{taggedsampler}{GS($1,1$)}
Initialize $\mu(0)$, $\avec(0)$ and $\bvec(0)$ and then iterate
\begin{align*}
\mu(s+1)
&\sim N\left(\ymean-\amean(s)-\bmean(s),\frac{\var{e}}{n}\right)\\
a_i(s+1)
&\sim N\left(\frac{\var{a}(\yimean-\mu(s+1)-\bimean(s))}{\var{a}+\frac{\var{e}}{JK}},\frac{\var{a}\frac{\var{e}}{JK}}{\var{a}+\frac{\var{e}}{JK}}\right)
&\hbox{ for all } i
\\
b_{ij}(s+1)
&\sim N\left(\frac{\var{b}(\yijmean-\mu(s+1)-a_i(s+1))}{\var{b}+\frac{\var{e}}{K}},\frac{\var{b}\frac{\var{e}}{K}}{\var{b}+ \frac{\var{e}}{K} }\right)
&\hbox{ for all } i,j
\end{align*}
where we use the dot subscript to indicate averaging over one dimension, meaning that
$\amean=\sum_{i}\frac{a_i}{I}$,
$\bmean=\sum_{i,j}\frac{b_{ij}}{IJ}$,
$\ymean=\sum_{i,j,k}\frac{y_{ijk}}{IJK}$,
$\bimean=\sum_{j}\frac{b_{ij}}{J}$,
$\yimean=\sum_{j,k}\frac{y_{ijk}}{JK}$
and
$\yijmean=\sum_{k}\frac{y_{ijk}}{K}$.
\end{taggedsampler}
\begin{taggedsampler}{GS($0,0$)}
Initialize $\mu(0)$, $\gavec(0)$ and $\etavec(0)$ and then iterate
\begin{align*}
\mu(s+1)
&\sim N\left(\gamean(s),\frac{\var{a}}{I}\right)
&\\
\gamma_i(s+1)
&\sim N\left(\frac{\frac{\var{b}}{J}\mu(s+1)+\var{a}\etai(s)}{\frac{\var{b}}{J}+\var{a}},\frac{\var{a}\frac{\var{b}}{J}}{\var{a}+\frac{\var{b}}{J}}\right)
&\forall\, i\\
\eta_{ij}(s+1)
&\sim N\left(\frac{\frac{\var{e}}{K}\gamma_i(s+1)+\var{b}\yijmean}{\var{b}+\frac{\var{e}}{K}},\frac{\var{b}\frac{\var{e}}{K}}{\var{b}+\frac{\var{e}}{K}}\right)
&\forall\, i,j
\end{align*}
where as before the dot subscript indicates averaging over indices.
\end{taggedsampler}

\subsection{Traceplots and autocorrelation function for the simulations in Section \ref{sec:gamma_poisson} of the paper}
Figures \ref{fig:GP_trace1}-\ref{fig:GP_trace3} provide traceplots and autocorrelation functions for $\log(\mu)$, $\log(a^{(1)}_2)$ and $\log(a^{(2)}_2)$ for the Gibbs Samplers considered in Table 2 of the paper (NB: we consider $a^{(1)}_2$ and $a^{(2)}_2$ rather than $a^{(1)}_1$ and $a^{(2)}_1$ because the latter are constrained to be equal to 1 in one of the sampler implementation).
\begin{figure}[ht]
\centering
\includegraphics[width=\linewidth]{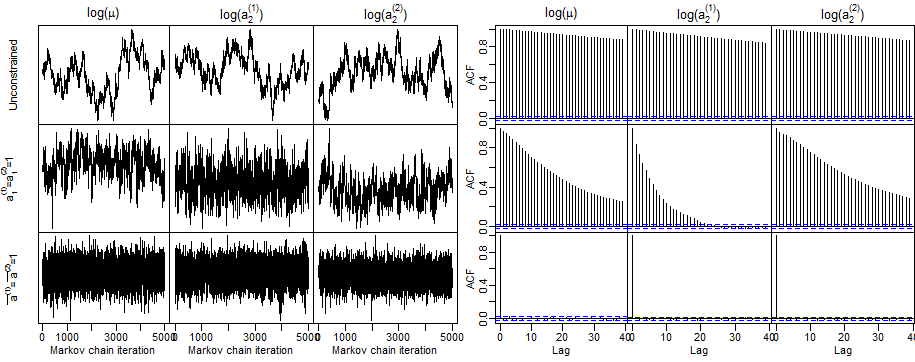}
\caption{Traceplots and autocorrelation functions for the case $n_1=5$ and $n_2=5$ reported in Table 2 of the paper.
First row: standard unconstrained Gibbs Sampler. 
Second row: constrained version with $a^{(1)}_1=a^{(2)}_1=1$. 
Third row: constrained version with $\ma^{(1)}=\ma^{(2)}=1$.}
\label{fig:GP_trace1}
\end{figure}
\begin{figure}[h!]
\centering
\includegraphics[width=\linewidth]{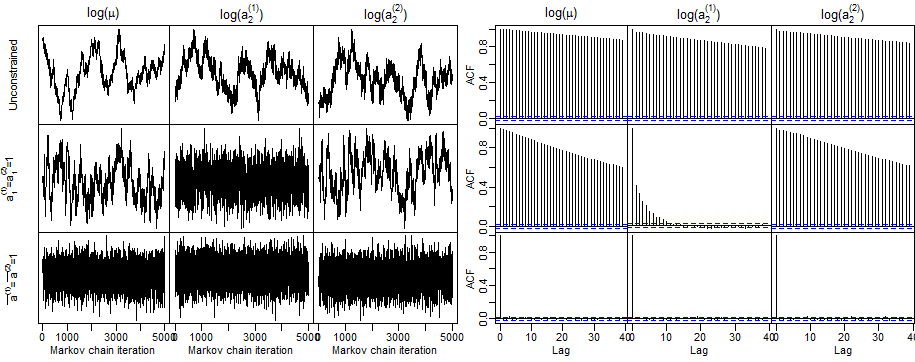}
\caption{Traceplots and autocorrelation functions for the case $n_1=5$ and $n_2=100$ reported in Table 2 of the paper.
First row: standard unconstrained Gibbs Sampler. 
Second row: constrained version with $a^{(1)}_1=a^{(2)}_1=1$. 
Third row: constrained version with $\ma^{(1)}=\ma^{(2)}=1$.}
\label{fig:GP_trace2}
\end{figure}
\begin{figure}[h!]
\centering
\includegraphics[width=\linewidth]{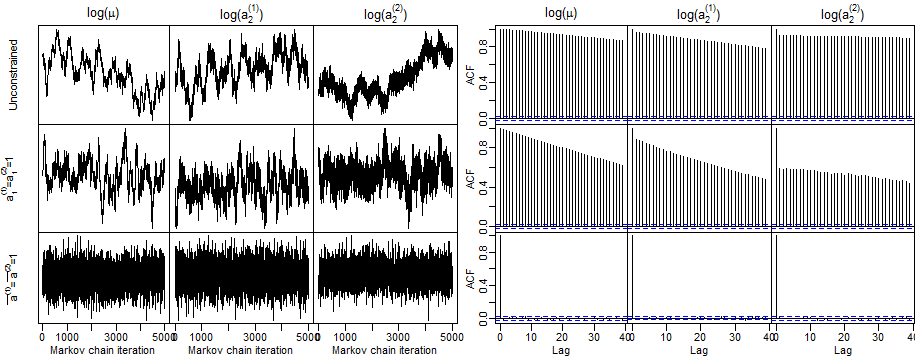}
\caption{Traceplots and autocorrelation functions for the case $n_1=100$ and $n_2=100$ reported in Table 2 of the paper.
First row: standard unconstrained Gibbs Sampler. 
Second row: constrained version with $a^{(1)}_1=a^{(2)}_1=1$. 
Third row: constrained version with $\ma^{(1)}=\ma^{(2)}=1$.}
\label{fig:GP_trace3}
\end{figure}
Similarly, Figures \ref{fig:GP_trace_HMC} and \ref{fig:GP_trace_NUTS} provide traceplots and autocorrelation functions for $\log(\mu)$, $\log(a^{(1)}_2)$ and $\log(a^{(2)}_2)$ for the HMC and NUTS algorithms considered in Table 3 of the paper.
See Table 3 of the paper for corresponding runtimes and effective sample sizes.
\begin{figure}[h!]
\centering
\includegraphics[width=\linewidth]{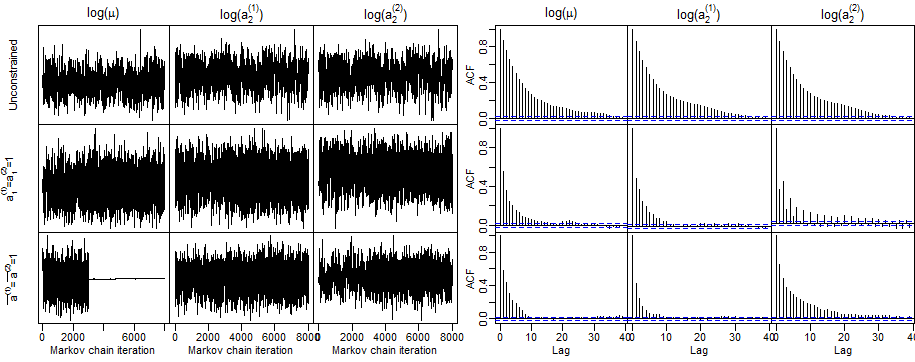}
\caption{Traceplots and autocorrelation functions for the HMC algorithm for the case $n_1=100$ and $n_2=100$ reported in Table 3 of the paper.
First row: unconstrained version. 
Second row: constrained version with $a^{(1)}_1=a^{(2)}_1=1$. 
Third row: constrained version with $\ma^{(1)}=\ma^{(2)}=1$.
Note that for HMC the runtime varies significantly across rows (see Table 3 of paper).
}
\label{fig:GP_trace_HMC}
\end{figure}
\begin{figure}[h!]
\centering
\includegraphics[width=\linewidth]{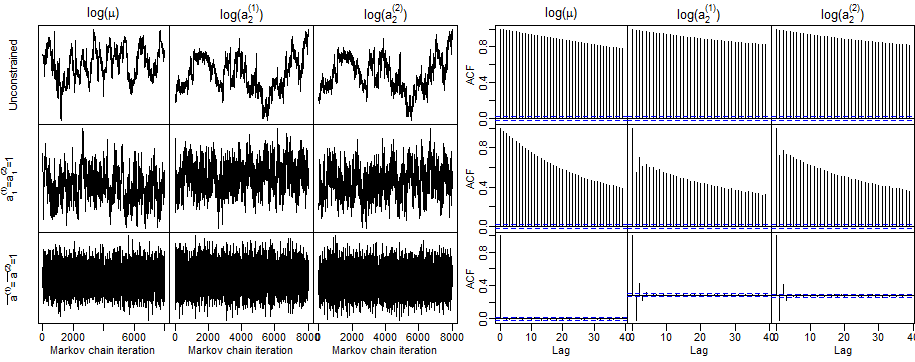}
\caption{Traceplots and autocorrelation functions for the NUTS algorithm for the case $n_1=100$ and $n_2=100$ reported in Table 3 of the paper.
First row: unconstrained version. 
Second row: constrained version with $a^{(1)}_1=a^{(2)}_1=1$. 
Third row: constrained version with $\ma^{(1)}=\ma^{(2)}=1$.}
\label{fig:GP_trace_NUTS}
\end{figure}

\subsection{R code for the simulations in Section \ref{sec:gamma_poisson} of the paper}
In this section we provide the R and Stan code used to perform the simulations reported in Table 2 and 3 of the paper.
First we provide the R code defining the functions for the three Gibbs Samplers implemented in Table 2 of the paper.
\definecolor{dkgreen}{rgb}{0,0.6,0}
\lstset{frame=tb,
  language=R,
  aboveskip=3mm,
  belowskip=3mm,
  showstringspaces=false,
  columns=flexible,
  basicstyle={\small\ttfamily},
  numbers=none,
  numberstyle=\tiny\color{gray},
  keywordstyle=\color{blue},
  commentstyle=\color{dkgreen},
  stringstyle=\color{gray},
  breaklines=true,
  breakatwhitespace=true,
  tabsize=3
}
\begin{lstlisting}[language=R]
## R CODE FOR THE GIBBS SAMPLERS IMPLEMENTED IN TABLE 2
Gibbs_CrossedPoisson_1<-function(y,num_iterations,n1,n2){
  ## Gibbs Sampler for crossed random effect model with Poisson likelihood and Gamma prior
  ## Version 1: unconstrained
  y_bar<-mean(y)
  samples<-matrix(NA,nrow = num_iterations,ncol = 1+n1+n2) # initialize empty matrix for posterior samples
  mu<-1;a<-rep(1,n1);b<-rep(1,n2) # set starting parameters
  for (t in 1:num_iterations){# Gibbs iterations
    mu<-rgamma(1,shape = alpha_mu+y_bar,rate = beta_mu+mean(a)*mean(b))
    a<-rgamma(n1,shape = alpha_a+rowMeans(y),rate =beta_a+ mu*mean(b))
    b<-rgamma(n2,shape = alpha_b+colMeans(y),rate =beta_b+ mu*mean(a))
    samples[t,]<-c(mu,a,b)
  }
  return(samples)
}
Gibbs_CrossedPoisson_2<-function(y,num_iterations,n1,n2){
  ## Gibbs Sampler for crossed random effect model with Poisson likelihood and Gamma prior
  ## Version 2: conditioning on a[1]=b[1]=1
  y_bar<-mean(y)
  samples<-matrix(NA,nrow = num_iterations,ncol = 1+n1+n2)# initialize empty matrix for posterior samples
  mu<-1;a<-rep(1,n1);b<-rep(1,n2) # set starting parameters
  for (t in 1:num_iterations){# Gibbs iterations
    mu<-rgamma(1,shape = alpha_mu+y_bar,rate = beta_mu+mean(a)*mean(b))
    a[-1]<-rgamma(n1,shape = alpha_a+rowMeans(y)[-1],rate =beta_a+ mu*mean(b))
    b[-1]<-rgamma(n2,shape = alpha_b+colMeans(y)[-1],rate =beta_b+ mu*mean(a))
    samples[t,]<-c(mu,a,b)
  } 
  return(samples)
}
Gibbs_CrossedPoisson_3<-function(y,num_iterations,n1,n2){
  ## Gibbs Sampler for crossed random effect model with Poisson likelihood and Gamma prior
  ## Version 3: conditioning on mean(a)=mean(b)=1
  y_bar<-mean(y)
  samples<-matrix(NA,nrow = num_iterations,ncol = 1+n1+n2)# initialize empty matrix for posterior samples
  mu<-1;a<-rep(1,n1);b<-rep(1,n2) # set starting parameters
  for (t in 1:num_iterations){# Gibbs iterations
    mu<-rgamma(1,shape = alpha_mu+y_bar,rate = beta_mu+mean(a)*mean(b))
    a<-rgamma(n1,shape = alpha_a+rowMeans(y),rate =beta_a+ mu*mean(b))
    a<-a/mean(a)
    b<-rgamma(n2,shape = alpha_b+colMeans(y),rate =beta_b+ mu*mean(a))
    b<-b/mean(b)
    samples[t,]<-c(mu,a,b)
  }
  return(samples)
}
\end{lstlisting}
Second we provide the Stan code defining the models used for the HMC and NUTS simulations in Table 3 of the paper.
\begin{lstlisting}[language=Stan]
## STAN CODE DEFINING THE MODELS USED FOR THE SIMULATIONS IN TABLE 3 
/*
Crossed random effect model with Poisson likelihood and Gamma prior
Version 1: unconstrained
*/
data {
  int<lower=0> N;
  int<lower=0> n1;
  int<lower=0> n2;
  int<lower=1,upper=n1> blk1[N];
  int<lower=1,upper=n2> blk2[N];
  int<lower=0> y[N];
}
parameters {
  vector<lower=0>[n1] a;
  vector<lower=0>[n2] b;
  real<lower=0> mu;
}
transformed parameters {
  vector[N] lambda;
  real<lower=0> abar;
  real<lower=0> bbar;
  for (n in 1:N){
    lambda[n] = mu * a[blk1[n]]* b[blk2[n]];
  }
  abar = mean(a);
  bbar = mean(b);
}
model {
  a ~ gamma(2, 0.1);
  b ~ gamma(2, 0.1);
  mu ~ gamma(2, 0.1);
  y ~ poisson(lambda);
}
/*
Crossed random effect model with Poisson likelihood and Gamma prior
Version 2: conditioning on a[1]=b[1]=1
*/
data {
  int<lower=0> N;
  int<lower=0> n1;
  int<lower=0> n2;
  int<lower=1,upper=n1> blk1[N];
  int<lower=1,upper=n2> blk2[N];
  int<lower=0> y[N];
}
parameters {
  vector<lower=0>[n1-1] a;
  vector<lower=0>[n2-1] b;
  real<lower=0> mu;
}
transformed parameters {
  vector[N] lambda;
  vector<lower=0>[n1] a_all;
  vector<lower=0>[n2] b_all;
  real<lower=0> abar;
  real<lower=0> bbar;
  a_all[1]=1;
  for (i in 2:n1){
    a_all[i]=a[i-1];
  }
  bb[1]=1;
  for (j in 2:n2){
    bb[j]=b[j-1];
  }
  for (n in 1:N){
      lambda[n] = mu * a_all[blk1[n]]* b_all[blk2[n]];
  }
  abar = mean(a);
  bbar = mean(b);
}
model {
  a ~ gamma(2, 0.1);
  b ~ gamma(2, 0.1);
  mu ~ gamma(2, 0.1);
  y ~ poisson(lambda);
}
/*
Crossed random effect model with Poisson likelihood and Gamma prior
Version 3: conditioning on mean(a)=mean(b)=1
*/
data {
  int<lower=0> N;
  int<lower=0> n1;
  int<lower=0> n2;
  int<lower=1,upper=n1> blk1[N];
  int<lower=1,upper=n2> blk2[N];
  int<lower=0> y[N];
  vector[n1] alpha_a;
  vector[n2] alpha_b;
}
parameters {
  simplex[n1] a_norm;
  simplex[n2] b_norm;
  real<lower=0> mu;
}
transformed parameters {
  vector<lower=0>[n1] a;
  vector<lower=0>[n2] b;
  vector[N] lambda;
  real<lower=0> abar;
  real<lower=0> bbar;
  a = a_norm * n1;
  b = b_norm * n2;
  for (n in 1:N){
   lambda[n] = mu * a[blk1[n]]* b[blk2[n]]; 
  }
  abar = mean(a);
  bbar = mean(b);
}
model {
  a_norm ~ dirichlet(alpha_a);
  b_norm ~ dirichlet(alpha_b);
  mu ~ gamma(2, 0.1);
  y ~ poisson(lambda);
}
\end{lstlisting}

\end{document}